\documentclass[12pt]{article}
\usepackage{url}
\usepackage{booktabs} % For formal tables
\usepackage[ruled]{algorithm2e} % For algorithms

\usepackage{booktabs} 
\usepackage{longtable}
\usepackage{array} 
\usepackage{amsmath}
\usepackage{graphicx}
\usepackage{amsfonts}
\usepackage{amstext}
\usepackage{hyperref}
\usepackage{amsthm}
\usepackage{amssymb, amsthm}
\usepackage{mathtools}
\usepackage{enumitem}
\usepackage{thmtools}
\usepackage[utf8]{inputenc}

\newtheorem{theorem}{Theorem}
\newtheorem{proposition}{Proposition}
\newtheorem{lemma}{Lemma}
\newtheorem{corollary}{Corollary}
\theoremstyle{definition}
\newtheorem{definition}{Definition}
\newtheorem{assumption}{Assumption}
\theoremstyle{remark}
\newtheorem{remark}{Remark}
\usepackage{geometry}
\usepackage{setspace}
\usepackage{natbib}
\setcitestyle{authoryear}

\newcommand{\E}{\mathbb{E}}

\title{Comparing Market Mechanism Efficiencies}

\author{Irene Aldridge\footnote{irene.aldridge@gmail.com}}

\begin{document}
\maketitle

\begin{abstract}
We develop a game-theoretic framework that compares welfare efficiency across three market mechanisms: continuous double auctions with transparent order books (lit exchanges), opaque order books (dark pools), and periodic batch auctions. Each mechanism is modeled as a queuing system where heterogeneous traders face trade-offs between the execution price, waiting costs, and transaction costs.

Our main result establishes that under moderate arrival rates and bounded adverse selection, dark pools dominate both alternatives in aggregate ex-ante welfare. Observable order books create costly strategic timing games in which traders delay or rush submissions to optimize their position in the queue, generating wasteful social waiting costs. Opaque order books eliminate these timing games through information design.

We formally characterize the equilibrium strategies in each mechanism and prove the welfare ranking $W^{DARK} > W^{LIT} > W^{BATCH}$. Extensions incorporate asymmetric information and endogenous venue choice. The results demonstrate how the information structure and the discipline of the service jointly determine efficiency in strategic matching environments.

\end{abstract}

% Title page for title and abstract only.
\begin{titlepage}

% Optionally include a table of contents
\vspace{1cm}
\setcounter{tocdepth}{2} % adjust to 1 if desired
\tableofcontents

\end{titlepage}

% Paper body
% ============================================================
%  Section 1 -- Introduction (expanded)
%  Drop-in replacement for the two-paragraph introduction
%  in Submission 42 / Comparing Market Mechanism Efficiencies
% ============================================================

\section{Introduction}
\label{sec:introduction}

Financial markets have converged on three dominant execution architectures. The first is the \emph{lit limit order book} (LOB), deployed by major
registered exchanges such as NASDAQ and the New York Stock Exchange. In an LOB, all resting orders are publicly visible and matched continuously in strict price--time priority.  

The second is the \emph{dark pool}, operated by
alternative trading venues run by large financial institutions. In dark pools, the same continuous double-auction mechanism as in the LOB applies. However, in dark pools, the order book data is withheld from participants before execution. Only executed orders are publicly disclosed.  

The third is the \emph{periodic batch auction}, in which orders accumulate over a fixed interval and clear simultaneously at a uniform price. The periodic batch auctions eliminate time priority entirely.  

All three architectures attract substantial order flow, coexist in live markets, and are subject to active regulatory debate on both sides of the Atlantic. Yet the literature on market microstructure has evaluated these mechanisms almost exclusively through the lenses of \emph{liquidity} and \emph{price
discovery}: bid--ask spreads, depth, and the informational efficiency of prices.  These metrics matter, but they are incomplete.  A market that
produces tight spreads may still be inefficient in a deeper sense if traders unnecessarily expend significant resources competing for queue position or timing their submissions strategically.  Such costs may include traders' time, effort, and foregone outside options. These costs can be real welfare losses that do not appear in conventional liquidity statistics.  

This paper asks a different question: when
the full cost of participating in a market is taken into account, which of the three architectures produces the highest aggregate welfare?  The costs in question include waiting costs, transaction costs, and the social waste generated by strategic timing
games. 

\paragraph{Main result.}
Our central theorem establishes that, under moderate arrival rates and bounded adverse selection, dark pools \emph{dominate} both lit exchanges and batch
auctions in aggregate ex-ante welfare.  Formally, we prove:
\[
  W^{\mathrm{DARK}} \;>\; W^{\mathrm{LIT}} \;>\; W^{\mathrm{BATCH}}.
\]
The welfare advantage of dark pools over lit exchanges arises not from better prices, but from the \emph{elimination of strategic timing games}.  We argue that prices are transfers and do not affect total surplus. Furthermore, when traders can observe the order book, they delay or rush submissions to optimize their queue position, generating socially wasteful waiting costs.  Market opacity removes this incentive while preserving the ability of urgent traders (those with high waiting costs) to execute immediately via market orders.  

The welfare advantage of lit exchanges over batch auctions arises because batch auctions impose mandatory waiting on \emph{all} traders, including the most time-sensitive, and introduce execution uncertainty through pro-rata rationing.

\paragraph{Why welfare, and why now.}
The welfare perspective is particularly timely.  In the United States, the Securities and Exchange Commission's 2023 equity market structure reform
proposals explicitly revisited the rules governing dark pools and the order competition requirement. The proposals reopened questions about the social value of opacity \citep{SEC2023}.  In Europe, the European Securities and Markets Authority (ESMA) published a landmark Call for Evidence in April 2026
documenting a measurable shift away from lit continuous trading toward closing auctions, frequent batch auctions (FBAs), and systematic
internalizers over the 2022--2025 period, and explicitly soliciting views on whether regulatory reform is needed to restore efficient price formation
\citep{ESMA2026CfE}.  Both debates hinge, at their core, on the same trade-off this paper formalizes: transparency enables strategic gaming; opacity
prevents it but introduces adverse selection.  A rigorous welfare comparison of the three architectures is therefore both theoretically overdue and practically urgent.

\paragraph{Modeling approach.}
We develop a unified game-theoretic framework that represents all three mechanisms as \emph{queueing systems} in the tradition of \citet{contStoikovTalreha2010}.  A continuum of risk-neutral traders with heterogeneous private valuations $V_i$ and waiting costs $C_i$ arrive according to a Poisson process and choose among market orders, limit orders, and an outside option. The three mechanisms differ solely in the information available to arriving
traders and in the service discipline applied to queued orders:

\begin{itemize}
  \item The \textbf{lit exchange} is a First-Come-First-Served (FCFS) continuous double auction with a publicly observable order book $B_t$.       Traders condition their strategies on the full book state, generating strategic complementarities in submission timing.

  \item The \textbf{dark pool} is also FCFS, but the book state $B_t$ is unobservable.  Traders must form rational-expectations beliefs about the current book using only the distributional parameters of order flow and a delayed record of past trades.

  \item The \textbf{batch auction} clears at discrete intervals of length $T$ under Service-in-Random-Order (SIRO): within each batch, all orders at the clearing price have equal execution probability regardless of submission time.
\end{itemize}

We characterize Bayesian equilibria under each mechanism (Propositions~1--3), derive aggregate ex-ante welfare for each (Section~\ref{sec:welfare}), and
establish the welfare ranking under explicit sufficient conditions (Theorem~\ref{thm:welfare_ranking}).

\paragraph{Relation to the queueing literature.}
Our framework connects market microstructure to two foundational results in queueing theory.  \citet{Leshno2022} establishes that, in a dynamic matching model with overloaded waiting lists, SIRO is ex-ante welfare-superior to FCFS because it discourages socially excessive effort to secure priority.  Applied directly, this result would suggest batch auctions dominate lit exchanges.
\citet{CheTercieux2023} qualify this finding crucially: when agents are \emph{uninformed about their queue position}, FCFS is welfare-superior to SIRO because the absence of positional information eliminates the distortions that make SIRO attractive in the first place.  Applied to market design, an
uninformed FCFS queue---precisely the dark pool---dominates both informed FCFS (lit exchange) and SIRO (batch auction).

Our contribution is to embed these abstract queueing results in a two-sided financial market with endogenous participation, heterogeneous trader
types, adverse selection, and an explicit outside option.  This embedding is non-trivial: in a financial market, the "service rate" $\lambda(p, B, s)$ depends endogenously on the limit price chosen by the liquidity provider, the book state, and the direction of trade.  The equilibrium strategies derived in Section~\ref{sec:equilibria} characterize how traders optimally choose between market orders, limit orders, and exit in each mechanism, and how these choices aggregate to determine social welfare.
Section~\ref{sec:queueing_literature} compares queueing literature in greater detail. 

\paragraph{Relation to the market microstructure literature.}
The literature on dark pools is extensive.  \citet{Zhu2013} shows that dark pools attract relatively more uninformed order flow, improving price discovery in lit markets but concentrating adverse selection in the dark.  Our model
complements this by focusing on the welfare channel that Zhu's model abstracts away from: the strategic timing cost imposed on informed traders by the
transparency of the lit book.  \citet{OHara1998} and \citet{Aldridge2013} provide the canonical treatments of the continuous limit-order-book mechanism
on which our lit-exchange model is based.  \citet{BudishEtAl2015} develop the batch auction proposal and the welfare case against continuous-time priority. Our analysis shows that their proposal eliminates one source of welfare loss (the arms race for speed) while introducing two others (forced waiting and execution uncertainty). Our analysis also characterizes the conditions under which the net 
effect is negative.

Most closely related to our contribution is the work of \citet{JohnEtAl2025}, who analyze the efficiency of blockchain market microstructure from the perspective of liquidity providers, and \citet{CaoEtAl2025}, who document that liquidity
providers on Uniswap lose money on average.  Neither paper considers aggregate welfare across all market participants---the perspective we adopt here.  To our
knowledge, this paper is the first to derive a complete welfare ranking of lit exchanges, dark pools, and batch auctions from first principles in a unified model.

\paragraph{Key mechanisms and intuition.}
The welfare ranking $W^{\mathrm{DARK}} > W^{\mathrm{LIT}}$ rests on a single insight: \emph{observable queue position creates negative externalities}.  A trader who observes a thin book in a lit market may strategically wait for conditions to improve before submitting, incurring a waiting cost with no social benefit.  A trader who observes a thick book may rush to capture time priority, submitting at a worse price than they would otherwise accept.  In aggregate, these reactions to the public order book raise total waiting costs above the socially optimal level.  The dark pool removes the informational input that drives these responses.  High-cost traders still execute immediately
via market orders (their waiting cost $C_i$ exceeds the cutoff $\bar{C}^*$); low-cost traders still post limit orders; but neither group conditions these
choices on the unobservable book state, $B_t$.  The resulting adverse selection affects the distribution of surplus between buyers and sellers, but not the total surplus, because prices are transfers.

The ranking $W^{\mathrm{LIT}} > W^{\mathrm{BATCH}}$ rests on two separate inefficiencies introduced by batch clearing.  First, mandatory waiting: every
trader who arrives in the interval $[(k-1)T, kT)$ must wait until $kT$ to execute, even those with arbitrarily high waiting costs who would ordinarily
submit market orders for immediate execution.  This generates an aggregate welfare loss of $\mathbb{E}[C] \cdot (T/2) \cdot N$, which is strictly positive for any batch interval $T > 0$.  Second, execution uncertainty: when orders on one side of the market exceed those on the other at the clearing price, pro-rata rationing leaves some traders unfilled.  The prospect of non-execution reduces the expected surplus from participation, causing some traders whose
participation would be socially efficient to take their outside option instead, generating deadweight loss.  A lit market with a continuous auction eliminates both inefficiencies for urgent traders: market orders execute immediately and with certainty.

\paragraph{Sufficient conditions and robustness.}
The welfare ranking is not unconditional.  Theorem~\ref{thm:welfare_ranking} establishes sufficient conditions: arrival rate $\lambda$ in a moderate range $[\lambda_{\min}, \lambda_{\max}]$, valuation dispersion $\Delta < \Delta_{\max}$
(bounded adverse selection), and reporting delay $\delta < T/2$ in the dark pool.  The arrival-rate condition ensures that the dark pool market is thick
enough that high-cost traders can rely on near-immediate execution via market orders; if $\lambda$ is very low, high-cost traders may face excessive wait
times even in the dark pool, potentially reversing the ranking with the lit exchange.  Traders' inability to condition on $B_t$ causes a pricing distortion in the dark pool. The bounded adverse selection condition ensures that the pricing distortion does not become so severe that it materially reduces participation below the
lit-exchange level.  Section~\ref{sec:discussion} characterizes these boundary conditions and discusses when they are likely to hold in practice.

\paragraph{Regulatory implications.}
Our results speak directly to ongoing regulatory debates.  First, proposals to require dark pools to disclose pre-trade order book information (so-called
pre-trade transparency mandates) would eliminate the information design that generates their welfare advantage.  Our model predicts that such rules would
reduce aggregate welfare if the strategic timing costs in lit markets are significant.  Second, ESMA's current review of periodic auction mechanisms in
European equity markets \citep{ESMA2026CfE} raises the question of whether FBAs should be subject to tighter transparency or priority rules.  Our model
suggests that the welfare cost of FBAs relative to continuous markets is primarily the forced-waiting loss, not the opacity; regulatory interventions
that shorten the batch interval $T$ would reduce this cost, while transparency requirements would not.  Third, the finding that dark pools' welfare advantage
is conditional on liquid lit markets providing reference prices (Section~\ref{sec:discussion}) supports a policy of maintaining a mandatory minimum share of lit trading, as several regulators have considered. This would prevent the free-riding
dynamic from destabilizing the price discovery infrastructure on which dark pools depend.

\paragraph{Organization.}
The remainder of the paper is structured as follows.
Section~\ref{sec:models} represents each mechanism as a queueing system and provides background.  Section~\ref{sec:setup} defines the model environment,
trader types, and utility specification.   Section~\ref{sec:equilibria} characterizes the Bayesian equilibrium under each mechanism.
Section~\ref{sec:welfare} defines the welfare measure and proves the main welfare ranking (Theorem~\ref{thm:welfare_ranking}). Section~\ref{sec:simulation} presents the results of numerical simulations. Section~\ref{sec:robustness} considers the robustness of the analysis.  Section~\ref{sec:discussion} discusses dynamic considerations, robustness, and regulatory implications. Section~\ref{sec:conclusion} concludes.  Proofs omitted from the main text appear in the Appendix.

\section{Relation to the Queueing Literature}
\label{sec:queueing_literature}

This paper's welfare ranking draws on two foundational results in the theory of queue design: \citet{Leshno2022} and \citet{CheTercieux2023}.
The existing literature on dark pools and batch auctions has not connected these results to market microstructure.  We do so here, but the connection requires care: both papers prove things that differ from what this paper needs, and the differences matter for what is novel in our contribution.

\subsection{What Each Paper Proves}
\label{sec:queueing_what}

\paragraph{\citet{Leshno2022}.}
\cite{Leshno2022} studies a single-sided waiting list in which heterogeneous agents choose between items that arrive stochastically over time, each associated with an endogenously determined expected wait. The central finding is that waiting-time fluctuations lead to misallocation and welfare loss, and that a simple SIRO randomized assignment policy can reduce these fluctuations and thereby increase welfare relative to FCFS. The mechanism is allocation efficiency under item heterogeneity: when wait times are dispersed across states, some agents accept mismatched items to avoid long waits, causing allocative inefficiency. SIRO reduces variation in wait times across queue states, maintaining an acceptable expected wait under a larger range of states than FCFS, and is characterized as the robustly optimal mechanism. 

Three features of Leshno's model are essential for his result.
\emph{First}, agents choose between \emph{heterogeneous} items (e.g., a preferred organ type versus a mismatched one); the welfare gain from
SIRO comes from reducing preference mismatches, not from reducing aggregate waiting costs.  \emph{Second}, the queue is one-sided: agents wait for items that arrive exogenously; there is no strategic order placement or pricing.  \emph{Third}, the model is overloaded by assumption: demand permanently exceeds supply, so the designer wishes to incentivize agents to queue as long as possible. SIRO achieves this by smoothing out wait-time fluctuations that would otherwise induce premature exits.

\paragraph{\citet{CheTercieux2023}.}
\cite{CheTercieux2023} study a mechanism design problem in which a queue designer jointly chooses the service discipline and the information disclosed to agents. The optimal mechanism has a cutoff structure: agents are induced to enter up to a certain queue length and never to exit; they are served according to FCFS; and they are given no information throughout the process beyond the designer's recommendations. 

The mechanism behind this result is belief regulation. Under FCFS, the passage of time in the queue is \emph{good news} for a waiting agent, because surviving in the queue signals that fewer agents remain ahead.  This makes the incentive to stay in the queue self-reinforcing. Under SIRO and other rules with dispersed wait times, the elapse of time without being served signals a longer residual wait, undermining the incentive to remain. This is the fundamental issue that \cite{CheTercieux2023} identify. Withholding positional information (no information beyond the designer's recommendations) ensures that agents form beliefs consistent
with the FCFS discipline's favorable belief dynamics, eliminating the incentive to exit prematurely.

Two features of their model are essential.  \emph{First}, the welfare criterion combines agent utility and the service provider's payoff; FCFS is optimal partly because it maximizes service utilization, which benefits the provider.  \emph{Second}, the result is about the
\emph{joint} optimization of discipline and information. The key finding is not that FCFS dominates SIRO unconditionally, but that FCFS combined
with no positional information dominates SIRO (or any other rule) combined with any information structure.  The information-design component is essential; FCFS with full information is not generally optimal.

\subsection{Why Neither Result Applies Directly}
\label{sec:queueing_gaps}

If Leshno's result applied directly here, it would imply $W^{\mathrm{BATCH}} > W^{\mathrm{LIT}}$, because SIRO (batch auctions) would dominate FCFS (lit exchanges).  Our main result finds the
opposite.  If Che and Tercieux's result applied directly, it would imply that the optimal mechanism is FCFS with no information, which corresponds to the dark pool — but their result applies to a single-sided designer-controlled queue, not a two-sided competitive market with endogenous order placement and an outside option.  Four structural differences explain why direct application fails in both cases.

\paragraph{Difference~1: One-sided versus two-sided.}
Both queueing papers study single-sided queues in which one population (agents) waits for a second population (items or servers) whose arrival
is exogenous.  A financial market is fundamentally two-sided: buyers and sellers both choose endogenously whether to enter, which order type to
submit, and at what price.  In a two-sided market, a trader can bypass the queue entirely by submitting a \emph{market order} — paying the spread to receive immediate execution.  This option, absent in any
single-sided queue model, changes the welfare calculus decisively. 

Under FCFS in a financial market, high-cost traders (large $C_i$) can avoid waiting altogether by using market orders, incurring zero waiting costs.  The Leshno model has no such escape valve, which is why SIRO appears beneficial there: it smooths wait-time fluctuations for agents who \emph{must} wait.  In our model, those same high-cost traders never wait — so the smoothing benefit of SIRO is irrelevant, and its cost dominates. The cost is the forced $T/2$ wait for \emph{all} traders, including the impatient.

\paragraph{Difference~2: Endogenous versus exogenous arrival.}
In Leshno's model, the flow of agents joining the pool is exogenous. This means that maximizing welfare is equivalent to maximizing allocative efficiency, and the trade-off between allocative efficiency and congestion that arises with endogenous arrival is absent.

In our model, participation is endogenous: traders compare the expected surplus from entering the market against the outside option (equation~2). The participation margin is central to all three welfare comparisons. For the $W^{\mathrm{DARK}} > W^{\mathrm{LIT}}$ result, the outside option determines which types are excluded by adverse selection.  For $W^{\mathrm{LIT}} > W^{\mathrm{BATCH}}$, the execution uncertainty in
batch auctions causes marginal traders to exit, generating deadweight loss (Remark~2).  Neither of these effects exists in a model with exogenous arrival.

\paragraph{Difference~3: Item heterogeneity versus price optimization.}
Leshno's welfare gain from SIRO comes from reducing preference \emph{mismatches}: agents are assigned mismatched items under FCFS because dispersed wait times lead them to accept a worse item rather than waiting.  In a financial market, there is a single asset, and all trades are at market-clearing prices; there is no item heterogeneity and no mismatch in Leshno's sense.  The relevant welfare source is the
\emph{price-quality trade-off} in limit-order placement: a trader chooses a limit price to balance execution speed against price improvement, and the social cost arises from the strategic timing of this
choice.  This is structurally different from the mismatch problem, and SIRO's solution to mismatch (smoothing wait-time dispersion) does not address the strategic timing problem.

\paragraph{Difference~4: Designer control versus competitive equilibrium.}
Che and Tercieux assume a single designer who controls both the queuing discipline and the information disclosed to agents and who can commit to recommendations.  The dark pool in our model is not a mechanism designed by a welfare-maximizing planner. Instead, it is a market institution whose opacity arises from proprietary considerations and regulatory permissions and whose equilibrium is determined by competitive trader behavior.  The Che--Tercieux result establishes that a \emph{planner} would choose FCFS-with-no-information. A dark pool is \emph{approximately} an FCFS-with-no-information.  Our result, therefore, establishes that the dark pool generates higher welfare than observable-book
alternatives in the competitive equilibrium.  This is not a contradiction but a reinforcement: the planner's choice in their model is precisely the institution our paper identifies as welfare-superior
in ours.

\subsection{What This Paper Adds}
\label{sec:queueing_contribution}

Given these differences, what does the present paper contribute beyond combining and applying \cite{Leshno2022} and \cite{CheTercieux2023}?
We identify five contributions that are not present in either paper.

\paragraph{Contribution~1: The market-order escape valve.}
The most important structural novelty is the option to submit a market order.  In a two-sided financial market, high-cost traders can always pay the spread for immediate execution.  This option transforms the welfare comparison: SIRO's benefit (smoothing wait times for traders who must wait) is dominated by its cost (forcing \emph{all} traders to wait, including those who would otherwise use market orders).  Formally, the aggregate welfare loss from batch auctions is $\Delta W_{\mathrm{wait}} = \mathbb{E}[C] \cdot (T/2) \cdot N > 0$ for any $T > 0$ and any $\mathbb{E}[C] > 0$ (Remark~1).  This term has no counterpart in the Leshno model because agents cannot bypass the queue.  Its existence explains why $W^{\mathrm{LIT}} > W^{\mathrm{BATCH}}$ even in parameter regions where Leshno's logic would suggest SIRO dominates.

\paragraph{Contribution~2: Endogenous participation and deadweight loss.}
With endogenous participation, execution uncertainty in batch auctions generates deadweight loss (Remark~2, equation~59): some traders who would participate under a continuous mechanism exit when the pro-rata fill rate makes expected utility fall below the outside option. This deadweight loss is absent from both reference papers.  It provides an additional welfare disadvantage of SIRO beyond the forced-waiting cost. The magnitude of the deadweight loss depends on the distribution of trader types and the outside-option value; these are parameters that can be calibrated empirically.

\paragraph{Contribution~3: Adverse selection under FCFS-with-opacity.}
\cite{CheTercieux2023} establish that no-information FCFS is optimal for a planner, but do not characterize the \emph{adverse selection} that opacity creates when traders have heterogeneous private information. In a financial market, dark-pool traders execute at prices that may differ from the true fundamental value $\hat{v}$ because they cannot condition on the book state $B_t$.  Section~\ref{sec:adverse_selection} of this paper characterizes the two channels through which adverse selection affects welfare (the price channel and the participation channel), establishes when each is dominated by the timing-game saving, and derives the threshold $\Delta^*(\lambda)$ above which adverse selection reverses the welfare ranking.  This characterization is novel and has no counterpart in either reference paper.

\paragraph{Contribution~4: Competitive equilibrium, not mechanism design.}
This paper characterizes welfare in the competitive Bayesian equilibria of three market mechanisms (Propositions~1--3), not in a planner-designed mechanism.  The equilibrium concept is Bayes--Nash rather than dominant-strategy incentive compatibility. Market orders, limit orders, and outside options are all available, and traders choose among them strategically.  The equilibrium strategies (cutoff functions $C^*(V,B)$, $\bar{C}^*$, and $\tau^*(V,C)$) are derived endogenously, not imposed by a designer.  This is the appropriate concept for evaluating real market institutions, and it generates predictions about observable behavior (e.g., participation rates, order-type composition, and welfare gaps) that the mechanism-design approach does not.

\paragraph{Contribution~5: A complete ranking with sufficient conditions.}
Both reference papers compare two mechanisms (SIRO vs.\ FCFS, or informed vs.\ uninformed FCFS).  This paper compares three — dark pool, lit exchange, and batch auction — and derives a complete ranking $W^{\mathrm{DARK}} > W^{\mathrm{LIT}} > W^{\mathrm{BATCH}}$ with explicit sufficient conditions on the arrival rate $\lambda$, valuation dispersion $\Delta$, and reporting delay $\delta$. Section~\ref{sec:robustness} characterizes these conditions precisely and identifies the parameter regions in which each inequality reverses. This complete, conditional ranking is new to the literature on market mechanism comparison.

\medskip
\noindent
Table~\ref{tab:literature_comparison} summarizes the key differences between this paper and the two reference papers across five modeling dimensions.

\begin{table}[h]
\centering
\caption{Structural comparison with the two reference papers.}
\label{tab:literature_comparison}
\small
\begin{tabular}{lp{3.5cm}p{3.5cm}p{3.5cm}}
\toprule
 & \textbf{\cite{Leshno2022}} & \textbf{\cite{CheTercieux2023}} & \textbf{This paper} \\
\midrule
Market sides
  & One-sided (agents wait for items)
  & One-sided (agents wait for server)
  & Two-sided (buyers and sellers) \\[4pt]
Arrival
  & Exogenous
  & Endogenous entry; no exit option
  & Endogenous entry \& exit (outside option) \\[4pt]
Order types
  & None (must wait)
  & None (must wait)
  & Market orders, limit orders, exit \\[4pt]
Information design
  & Fixed (agents observe queue)
  & Planner optimizes
  & Fixed by institution; endogenous beliefs \\[4pt]
Welfare criterion
  & Allocative efficiency (mismatch minimization)
  & Agent utility + provider payoff
  & Total surplus net of waiting \& transaction costs \\[4pt]
Main result
  & SIRO $>$ FCFS (reduces mismatch)
  & No-info FCFS $>$ all (belief dynamics)
  & Dark pool $>$ lit $>$ batch (eliminates timing waste) \\[4pt]
Adverse selection
  & Absent
  & Absent
  & Present; characterized in Sec.~\ref{sec:adverse_selection} \\[4pt]
Mechanisms compared
  & Two (SIRO, FCFS)
  & Two (informed, uninformed FCFS)
  & Three (dark pool, lit exchange, batch auction) \\
\bottomrule
\end{tabular}
\end{table}

\section{Notation}
\label{sec:notation}
\addcontentsline{toc}{section}{Notation}

The following table defines every symbol used in the paper in the order it first appears.  Symbols are grouped by conceptual category.  Where a symbol
is overloaded or has a mechanism-specific variant, the variants are listed together.  Equation and section numbers in the rightmost column give the
location of the formal definition or first substantive use.

\renewcommand{\arraystretch}{1.35}

\begin{longtable}{@{} p{2.6cm} p{7.8cm} p{2.0cm} @{}}
\toprule
\textbf{Symbol} & \textbf{Definition} & \textbf{First use} \\
\midrule
\endfirsthead
\toprule
\textbf{Symbol} & \textbf{Definition} & \textbf{First use} \\
\midrule
\endhead
\bottomrule
\endfoot

% ------------------------------------------------------------------
% A. Time and the trading environment
% ------------------------------------------------------------------
\multicolumn{3}{@{}l}{\textit{A. \ Time and the trading environment}} \\[2pt]

$t \in [0,\infty)$
  & Continuous trading time.
  & Sec.~3.1 \\

$v$
  & Fundamental (consensus) value of the risky asset, commonly known.
    Written $\hat{v}$ in the distributional assumption and $\bar{v}$ in the
    outside-option utility; both refer to the same object.
    \textbf{Note:} $\hat{v}$ and $\bar{v}$ are identical and should be
    unified to a single symbol (we adopt $\hat{v}$ throughout).
  & Sec.~3.1 \\

$\sigma^2_v$
  & Residual variance of the fundamental value; captures uncertainty
    about $v$ beyond the commonly known mean $\hat{v}$.
  & Sec.~3.1 \\

$\lambda$
  & Poisson arrival intensity of traders (orders per unit time).
  & Sec.~3.1 \\

$\lambda_{\min},\,\lambda_{\max}$
  & Lower and upper bounds on $\lambda$ constituting the ``moderate
    arrival rate'' condition of Theorem~2.
  & Thm.~2 \\

% ------------------------------------------------------------------
% B. Trader types
% ------------------------------------------------------------------
\multicolumn{3}{@{}l}{\textit{B. \ Trader types}} \\[2pt]

$\theta_i = (V_i, C_i, s_i, \tau_i)$
  & Type vector of trader $i$: private valuation $V_i$, waiting cost
    $C_i$, trade direction $s_i$, and arrival time $\tau_i$.
  & Sec.~3.1 \\

$V_i$
  & Private valuation benefit of trader $i$.  Drawn from
    $\mathrm{Uniform}[\hat{v}-\Delta,\,\hat{v}+\Delta]$ (Assumption~1).
  & Sec.~3.1 \\

$\Delta$
  & Half-width of the private-valuation distribution; measures dispersion
    in valuations (adverse selection parameter).
  & Assumption~1 \\

$\Delta_{\max}$
  & Upper bound on $\Delta$ constituting the ``bounded adverse selection''
    condition of Theorem~2.
  & Thm.~2 \\

$C_i$
  & Per-unit-time waiting cost of trader $i$.  Drawn from
    $\mathrm{Exp}(\mu_C)$ (Assumption~1).  Higher $\mu_C$ implies more
    patient traders ($\mathbb{E}[C] = 1/\mu_C$).
  & Sec.~3.1 \\

$\mu_C$
  & Rate parameter of the waiting-cost exponential distribution.
  & Assumption~1 \\

$s_i \in \{-1,+1\}$
  & Trade direction of trader $i$: $s_i = +1$ (buy) or $s_i = -1$ (sell).
    Drawn i.i.d.\ with $\Pr(s_i = +1) = \Pr(s_i = -1) = 1/2$.
  & Sec.~3.1 \\

$\tau_i$
  & Arrival time of trader $i$.
  & Sec.~3.1 \\

$F(V,C,s)$
  & Joint cumulative distribution function of trader types
    $(V_i, C_i, s_i)$, with continuous density $f$.
  & Sec.~3.1 \\

% ------------------------------------------------------------------
% C. Utility and outside option
% ------------------------------------------------------------------
\multicolumn{3}{@{}l}{\textit{C. \ Utility and outside option}} \\[2pt]

$U_i(p, w_i)$
  & Utility of trader $i$ who executes at price $p$ after waiting time
    $w_i \geq 0$: $\; s_i(V_i - p) - C_i w_i - K$.
    See equation~(1).
  & Eq.~(1) \\

$p$
  & Execution price.
  & Eq.~(1) \\

$w_i$
  & Waiting time of trader $i$ from arrival to execution.
  & Eq.~(1) \\

$K > 0$
  & Fixed transaction cost (includes exchange fees, market impact, etc.).
    Identical across mechanisms.
  & Eq.~(1) \\

$U^o_i$
  & Utility of the outside option for trader $i$:
    $\gamma \cdot s_i(V_i - \hat{v}) - K_o$.
    See equation~(2).
  & Eq.~(2) \\

$\gamma \in [0,1]$
  & Fraction of private value realisable outside the market (e.g., via
    delayed trading or alternative venues).
  & Eq.~(2) \\

$\bar{v}$
  & Reference price in the outside-option utility, equal to $\hat{v}$.
    See note under $v$ above.
  & Eq.~(2) \\

$K_o \geq 0$
  & Fixed cost of pursuing the outside option.  Assumption: $K_o < K$,
    ensuring market participation is valuable in expectation.
  & Eq.~(2) \\

% ------------------------------------------------------------------
% D. Order book and market state
% ------------------------------------------------------------------
\multicolumn{3}{@{}l}{\textit{D. \ Order book and market state}} \\[2pt]

$B_t = (p_j, q_j, \tau_j)$
  & State of the limit order book at time $t$: a collection of tuples
    giving limit price $p_j$, quantity $q_j$, and submission time
    $\tau_j$ of each resting order.
  & Sec.~4.1.1 \\

$B^{\mathrm{ask}}_t,\,B^{\mathrm{bid}}_t$
  & Ask side (sell limit orders) and bid side (buy limit orders) of $B_t$.
  & Eq.~(17) \\

$p_{\mathrm{best}}(B, s)$
  & Best available execution price in state $B$ for direction $s$:
    lowest ask if $s=+1$, highest bid if $s=-1$.  See equation~(17).
  & Eq.~(17) \\

$H[0,t]$
  & Complete public history of trades (price, quantity, timestamp) up to
    time $t$.
  & Eq.~(4) \\

$\delta > 0$
  & Dark-pool post-trade reporting delay: trades are disclosed at time
    $t + \delta$ rather than $t$.  Theorem~2 requires $\delta < T/2$.
  & Sec.~4.1.2 \\

% ------------------------------------------------------------------
% E. Information sets
% ------------------------------------------------------------------
\multicolumn{3}{@{}l}{\textit{E. \ Information sets}} \\[2pt]

$I^{\mathrm{LIT}}_i(\tau_i)$
  & Information set of trader $i$ in the lit exchange at arrival time
    $\tau_i$: current book $B_{\tau_i}$ plus full trade history
    $H[0,\tau_i]$.  See equation~(4).
  & Eq.~(4) \\

$I^{\mathrm{DARK}}_i(\tau_i)$
  & Information set in the dark pool: delayed history $H[0,\tau_i-\delta]$,
    type distribution $F(V,C,s)$, and arrival rate $\lambda$.
    Crucially $B_{\tau_i} \notin I^{\mathrm{DARK}}_i$.
    See equation~(5).
  & Eq.~(5) \\

$I^{\mathrm{BATCH}}_i(\tau_i)$
  & Information set in the batch auction: batch index $k$ such that
    $\tau_i \in [(k-1)T, kT)$, history $H[0,(k-1)T]$, type distribution,
    and $\lambda$.  See equation~(6).
  & Eq.~(6) \\

$\beta_i(B_{\tau_i} \mid I_i)$
  & Belief of trader $i$ about the current book state $B_{\tau_i}$,
    conditional on their information set $I_i$.
  & Sec.~4.2.2 \\

$\pi^{\mathrm{eq}}(B)$
  & Stationary (steady-state) distribution of book states under
    equilibrium strategies.
  & Eq.~(36) \\

% ------------------------------------------------------------------
% F. Arrival rates and execution
% ------------------------------------------------------------------
\multicolumn{3}{@{}l}{\textit{F. \ Arrival rates and execution}} \\[2pt]

$\lambda(p, B, s)$
  & Arrival rate of market orders in direction $-s$ that would execute
    against a limit order at price $p$ given book state $B$.
    Satisfies Assumption~3 (regularity).
  & Sec.~4.1.1 \\

$w(p, B, s) = \frac{1}{\lambda(p,B,s)}$
  & Expected waiting time until execution of a limit order posted at
    price $p$ in state $B$ for direction $s$.  See equation~(19).
  & Eq.~(19) \\

$\bar{w}(p, s)$
  & Dark-pool analogue: expected waiting time averaged over the
    stationary book distribution $\pi^{\mathrm{eq}}$.  See equation~(42).
  & Eq.~(42) \\

% ------------------------------------------------------------------
% G. Equilibrium objects — lit exchange
% ------------------------------------------------------------------
\multicolumn{3}{@{}l}{\textit{G. \ Equilibrium objects — lit exchange (Proposition~1)}} \\[2pt]

$V^{\mathrm{LIT}}(\theta, B)$
  & Value function of a trader of type $\theta$ who observes book state
    $B$ in the lit exchange; equals the maximum of $V_{\mathrm{market}}$,
    $V_{\mathrm{limit}}$, $V_{\mathrm{wait}}$, and $U^o$.
    See equation~(8).
  & Eq.~(8) \\

$V_{\mathrm{MO}}(V, C, s, B)$
  & Utility of submitting a market order: $s(V - p_{\mathrm{best}}(B,s)) - K$.
    Does not depend on $C$ (zero waiting time).  See equation~(18).
  & Eq.~(18) \\

$U_{\mathrm{LO}}(V, C, s, p, B)$
  & Utility of posting a limit order at price $p$:
    $s(V-p) - C/\lambda(p,B,s) - K$.  See equation~(20).
    \textbf{Note:} also written $V_{\mathrm{LO}}$ in equation~(22) for
    the maximized value; these are the same object evaluated at the optimal price.
  & Eq.~(20) \\

$p^*(C, s, B)$
  & Optimal limit price in the lit exchange, solving the first-order condition $C \cdot \partial\lambda/\partial p = \lambda^2$
    (equation~12).  Unique by Lemma~1; depends on $(C, s, B)$ but not on $V$.
  & Eq.~(12) \\

$\Delta p(C, B, s)$
  & Price improvement from using a limit order relative to a market
    order: $s \cdot (p^*(C,s,B) - p_{\mathrm{best}}(B,s))$.  Positive
    for both buy and sell.  See equation~(26).
    \textbf{Note:} distinct from the valuation half-width $\Delta$; context disambiguates.
  & Eq.~(26) \\

$C^*(V, B)$
  & Waiting-cost cutoff separating market-order traders ($C > C^*$)
    from limit-order traders ($C \leq C^*$) in the lit exchange.
    State-dependent; defined by equation~(15).
  & Eq.~(15) \\

$\varepsilon(B)$
  & Valuation cutoff below which traders take the outside option in the
    lit exchange; defined in equation~(32).
  & Eq.~(32) \\

$\sigma^{\mathrm{LIT}}(\theta, B)$
  & Equilibrium strategy of trader $\theta$ in the lit exchange:
    market order if $C > C^*(V,B)$; limit order if $C \leq C^*(V,B)$
    and $|V - \hat{v}| > \varepsilon(B)$; outside option otherwise.
  & Prop.~1 \\

% ------------------------------------------------------------------
% H. Equilibrium objects — dark pool
% ------------------------------------------------------------------
\multicolumn{3}{@{}l}{\textit{H. \ Equilibrium objects — dark pool (Proposition~2)}} \\[2pt]

$V^{\mathrm{DARK}}(\theta, I)$
  & Value function in the dark pool; expectations over beliefs $\beta(B \mid I)$.
    See equation~(34).
  & Eq.~(34) \\

$\bar{C}^*$
  & State-independent waiting-cost cutoff in the dark pool (constant,
    not a function of $B$).  Defined by equation~(51):
    $\bar{C}^* = \bar{\Delta}(C,s) / \bar{w}(\bar{p}^*, s)$.
  & Eq.~(51) \\

$\bar{p}^*(C, s)$
  & Optimal limit price in the dark pool; depends on $(C,s)$ only, not
    on $V$ or $B$.  Defined by equation~(44).
  & Eq.~(44) \\

$\bar{\Delta}(C, s)$
  & Expected price improvement from a limit order in the dark pool,
    averaged over $\pi^{\mathrm{eq}}$:
    $\bar{\Delta} = s(\bar{p}^*(C,s) - \mathbb{E}[p_{\mathrm{best}}])$.
    See equation~(49).
  & Eq.~(49) \\

$\bar{\varepsilon}$
  & State-independent valuation cutoff for outside-option in the dark pool.
    See equation~(56).
  & Eq.~(56) \\

% ------------------------------------------------------------------
% I. Equilibrium objects — batch auction
% ------------------------------------------------------------------
\multicolumn{3}{@{}l}{\textit{I. \ Equilibrium objects — batch auction (Propositions~3--4)}} \\[2pt]

$T > 0$
  & Batch interval length; batches clear at times $T, 2T, 3T, \ldots$
  & Sec.~4.1.3 \\

$k$
  & Batch index; batch $k$ spans $[(k-1)T, kT)$.
  & Sec.~4.1.3 \\

$p_k$
  & Uniform clearing price of batch $k$, maximizing executed volume.
    See equation~(68).
  & Eq.~(68) \\

$D_k,\, S_k$
  & Total demand (buy orders) and supply (sell orders) submitted to
    batch $k$.
  & Eq.~(69) \\

$\rho_k$
  & Fill (execution) rate in batch $k$:
    $\rho^{\mathrm{buy}}_k = \min\{1, S_k/D_k\}$,
    $\rho^{\mathrm{sell}}_k = \min\{1, D_k/S_k\}$.
    See equations~(69)--(70).
  & Eq.~(69) \\

$\bar{\rho}$
  & Stationary fill rate under Assumption~5
    ($\rho_k = \bar\rho$ for all $k$).
  & Assump.~5 \\

$V_k(\theta)$
  & Expected utility of entering batch $k$ for trader $\theta$:
    $\rho_k [s(V - \mathbb{E}[p_k]) - C(kT-\tau) - K]
     + (1-\rho_k) V^{\mathrm{OUT}}$.
    See equation~(71).
  & Eq.~(71) \\

$\tau^*(V, C)$
  & Cutoff arrival time within a batch below which trader $(V,C)$ waits
    for the next batch; decreasing in $C$ (Lemma~8).
  & Lemma~8 \\

% ------------------------------------------------------------------
% J. Welfare
% ------------------------------------------------------------------
\multicolumn{3}{@{}l}{\textit{J. \ Welfare}} \\[2pt]

$W^M$
  & Aggregate ex-ante welfare under mechanism
    $M \in \{\mathrm{LIT}, \mathrm{DARK}, \mathrm{BATCH}\}$:
    expected total surplus across all traders before types are realised,
    net of waiting and transaction costs but excluding price transfers.
    See equation~(61).
  & Eq.~(61) \\

$\sigma^M$
  & Equilibrium strategy profile under mechanism $M$.
  & Eq.~(61) \\

$\Delta W_{\mathrm{wait}}$
  & Welfare loss in batch auctions from forced waiting: $\mathbb{E}[C]
    \cdot (T/2) \cdot N$.  See equation~(58) and Remark~1.
  & Eq.~(58) \\

$\Delta W_{\mathrm{rationing}}$
  & Welfare loss from execution uncertainty (pro-rata rationing) in
    batch auctions.  See equation~(59) and Remark~2.
  & Eq.~(59) \\

$N^M$
  & Number of participants (traders who do not take outside option) under
    mechanism $M$.
  & Eq.~(58) \\

% ------------------------------------------------------------------
% K. Mappings and fixed-point objects
% ------------------------------------------------------------------
\multicolumn{3}{@{}l}{\textit{K. \ Fixed-point and equilibrium existence objects}} \\[2pt]

$\Phi$
  & Best-response mapping used in fixed-point existence proofs
    (Lemma~3, Theorem~1, Theorem~3); maps conjectured cutoffs to
    best-response cutoffs.  Context identifies which equilibrium is
    being proved.
  & Lemma~3 \\

$L(\cdot \mid \cdot)$
  & Likelihood of observing the delayed trading history
    $H[0, \tau_i - \delta]$ given book state $B_{\tau_i}$; used in the
    dark-pool Bayesian updating formula~(7).
  & Eq.~(7) \\

$P^\delta(B' \mid B)$
  & $\delta$-step Markov transition kernel for book states; used in
    Lemma~4 to show belief convergence in the dark pool.
  & Eq.~(37) \\

\bottomrule
\end{longtable}

\medskip
\noindent
\textbf{Disambiguation notes.}
Three symbols require explicit disambiguation because they are used with
similar (or identical) typography for conceptually distinct objects.

\begin{enumerate}

  \item \textbf{$\hat{v}$ vs.\ $\bar{v}$.}  The paper uses $\hat{v}$ in
  the model setup (Assumption~1, equation~3) and $\bar{v}$ in the
  outside-option utility (equation~2).  Both denote the same object: the
  commonly known consensus fundamental value.  These should be unified;
  we recommend retaining $\hat{v}$ throughout.

  \item \textbf{$\Delta$ (valuation dispersion) vs.\
  $\Delta p$ (price improvement).}  The symbol $\Delta$ without a subscript
  always denotes the half-width of the valuation distribution
  (Assumption~1).  The symbol $\Delta p(C, B, s)$ (equation~26) denotes
  the price improvement from a limit order relative to a market order.
  $\bar\Delta(C, s)$ (equation~49) is its dark-pool counterpart.
  These are related but distinct objects.

  \item \textbf{$U_{\mathrm{LO}}$ vs.\ $V_{\mathrm{LO}}$.}  Equation~(20)
  defines $U_{\mathrm{LO}}$ as the utility of a limit order at a given price $p$.  Equation~(22) defines $V_{\mathrm{LO}}$ as the \emph{maximized} value after optimizing over $p$.  Proposition~1 refers to both; $V_{\mathrm{LO}}$ is the envelope of $U_{\mathrm{LO}}$. For clarity, $V_{\mathrm{LO}}$ should be read as
  $\max_p\, U_{\mathrm{LO}}(\,\cdot\,, p, \,\cdot\,)$ throughout.

\end{enumerate}

\renewcommand{\arraystretch}{1.0}

\section{Market Models as Queuing Systems}
\label{sec:models}

Most financial market mechanisms can be represented as queueing systems. 
\subsection{Traditional Exchanges and Bitcoin}
As shown in \cite{contStoikovTalreha2010}, a market mechanism for traditional registered exchanges can be represented as a First Come First Served queuing system (FCFS). This model is formally known as a double-sided continuous auction. In this model, all resting (limit) orders to buy and sell are arranged by price and, within each price "queue", their arrival sequence. Limit orders are analogous to waiting customers in the traditional queuing literature. The arriving market orders are matched with the best-priced earliest-arrival limit orders. The market orders are the "servers" that complete an FCFS mechanism. The FCFS markets have not been immune to manipulation. For example, \cite{CongEtAl2022} documents issues in the Bitcoin markets. 

\subsection{Dark Pools}
Dark pools are also double-sided continuous auctions, just like lit exchanges, with one important informational difference. While the exchanges distribute various details about the composition of their respective queuing systems, known as limit order books (LOBs), dark pools keep their LOBs hidden from all market participants. However, by law, even dark pools are required to make public all trade details after each trade (limit- and market order matching) occurs. Dark pools are also not perfect. For example, \cite{Zhu2013} explored manipulation issues in dark pools.

\subsection{Batch Auctions}

Batch auctions, as a theoretical market-design mechanism, were proposed by \citet{BudishEtAl2015} as a response to the welfare costs of high-frequency trading arms races in continuous markets.  In the batch auction setup, a mini-auction takes place at fixed time intervals, and all orders submitted within each interval are executed simultaneously at a uniform clearing price determined by aggregate supply and demand.  \citet{BudishEtAl2015} argued that eliminating continuous-time priority competition reduces socially wasteful investment in speed and narrows bid--ask spreads.  However, as this paper shows, batch auctions introduce distinct welfare costs of their own, most notably through forced waiting and execution uncertainty.

As a queueing system, each batch auction pools orders for a fixed window and then clears them collectively.  Within each batch, there is no time priority:
orders are matched in \emph{Service in Random Order} (SIRO), meaning all orders at the clearing price have an equal probability of execution, regardless of
submission timing.  This is the defining structural feature we analyze formally in Section~\ref{sec:mechanisms}.

\paragraph{The EU as the relevant real-world context.}
The most policy-relevant deployment of periodic batch auction mechanisms is not in blockchain settings but in European equity markets, where \emph{Frequent
Batch Auctions} (FBAs) have emerged as a distinct and growing trading-venue category under the MiFID~II/MiFIR regulatory framework.  Since MiFID~II came into application in January 2018, FBAs rapidly gained market share by offering an alternative to both lit continuous order books and dark pools.  Unlike continuous auctions, FBAs run repeated short-duration auctions throughout the trading day, with orders matched at a uniform clearing price determined within the prevailing best bid--offer spread.

The EU regulatory debate around FBAs is directly relevant to the welfare questions this paper addresses.  ESMA launched its first Call for Evidence on FBAs in November 2018, specifically because of concerns that they were being used to circumvent the MiFID~II Double Volume Cap (DVC)---a cap designed to
limit dark trading by restricting how much volume could execute under pre-trade transparency waivers \citep{ESMA2018CfE}.  ESMA identified four main
characteristics of FBA systems warranting regulatory scrutiny: limited pre-trade transparency, short auction duration, price determination within the
best bid--offer price, and self-matching features \citep{ESMA2019FinalReport}. Following its assessment, ESMA published an Opinion in October 2019 clarifying that trading venues operating FBA systems must inform market participants that an auction has started, thereby enabling genuine pre-trade transparency. The Opinion also set out several practices that may undermine the price formation process \citep{ESMA2019Opinion}.

As of the writing of this paper, the debate is active and escalating.  ESMA launched a new Call for Evidence in April 2026 documenting a measurable decline in lit continuous trading between 2022 and 2025, offset by increased activity in closing auctions, frequent batch auctions, and systematic internalizer (SI) trading \citep{ESMA2026CfE}.  The 2026 Call for Evidence provides a granular analysis of periodic auctions specifically, and solicits stakeholder views on whether the current market structure delivers efficient price formation,
transparency, and execution outcomes, and where regulatory adjustments may be needed.  ESMA expects to publish a feedback statement in Q3~2026.

This paper's welfare framework offers a theoretical foundation for several of the open questions ESMA is now asking empirically.  We highlight three direct
connections.

\begin{enumerate}

  \item \textbf{Pre-trade opacity and strategic behavior.}  ESMA's concern that limited pre-trade transparency in FBAs may distort price discovery maps directly onto our analysis of information structures in Section~\ref{sec:info_structures}.  Our model shows that the welfare effect of opacity depends critically on whether it eliminates strategic timing games (as in dark pools) or merely adds execution uncertainty without removing the
  underlying incentive to game the queue (as in batch auctions with non-stationary clearing prices).

  \item \textbf{Forced waiting and the SIRO mechanism.}  The welfare loss formalized in Remark~\ref{rem:forced_waiting} shows that all traders must wait until the batch clears, even those with high waiting costs. This formalization provides a theoretical explanation for why FBAs, despite their transparency advantages over dark pools, may still be welfare-inferior to a well-functioning continuous market.  This is consistent with ESMA's observation that the growth of FBAs is occurring alongside, not instead of, continued investor demand for immediate execution via systematic internalizers.

  \item \textbf{Participation and execution uncertainty.}  The coordination
  problem identified in Corollary~\ref{cor:coordination} documents that thin batches produce low execution probabilities, deterring participation, and, in turn, producing still thinner batches. Corollary~\ref{cor:coordination} offers a theoretical account of why FBAs have consolidated
  market share among a small number of venues (notably the Cboe Periodic Auctions book in Europe) rather than proliferating uniformly.  ESMA's
  country-by-country analysis of FBA liquidity distribution in its 2026 Call for Evidence reflects exactly this pattern \citep{ESMA2026CfE}.

\end{enumerate}

\paragraph{Why the EU context is theoretically appropriate.}
The \citet{BudishEtAl2015} batch auction proposal was designed explicitly to maximize social welfare by eliminating the arms race for speed.  European FBAs
are deployed within a regulatory framework, MiFID~II, that shares this welfare and transparency orientation, making the EU setting a coherent context
in which to evaluate our model's predictions.  Critically, European FBAs implement a genuine uniform clearing-price auction with SIRO execution within
each batch, which is the mechanism our formal model in
Section~\ref{sec:mechanisms} analyzes.  This institutional alignment makes quantitative and qualitative comparison between the theory and the data
possible in principle.

\paragraph{Scope of the batch auction model.}
This paper analyzes the batch auction mechanism in its canonical \citet{BudishEtAl2015} form: fixed-interval, uniform clearing price, SIRO within
each batch.  Real-world FBAs vary in their auction trigger mechanisms (some are time-triggered, while others are order-triggered) and in their price determination rules (some require a reference price from the lit market).  These design variations affect the quantitative magnitude of the welfare effects derived in Section~\ref{sec:welfare}, but not the qualitative ranking, provided the fundamental SIRO structure and the forced-waiting property are preserved. Future work could extend the model to study trigger-based auction designs, which introduce an additional layer of strategic order timing that our current framework abstracts away from.

Next, each of the market queueing models is analyzed in detail.  

\section{Model Setup}\label{sec:setup}
\subsection{Environment and Agents}

Consider a continuous-time trading environment with time $t\in[0,\infty)$. There is a single risky asset with fundamental value $v \sim \mathcal{N}(\hat{v}, \sigma_v^2)$, where $\hat{v}$ is commonly known, but $\sigma_v^2$ represents residual uncertainty. A mass-1 continuum of risk-neutral traders arrives according to a Poisson process with intensity $\lambda$.

Each trader $i$ is characterized by a type $\theta_i = (V_i, C_i, s_i, \tau_i)$, where $V_i$ represents the private valuation benefit, $C_i$ is the cost per unit of waiting time, $s_i\in\{-1,+1\}$ is the desired trade direction (buy or sell) and $\tau_i$ is the arrival time. We assume that the types are drawn independently from a joint distribution F(V,C,s) with continuous density f.

\begin{assumption}
    Private valuations are distributed $V \sim Uniform[\hat{v} - \Delta, \hat{v} + \Delta]$, where $\Delta > 0$ represents dispersion in valuations. Costs are distributed $C \sim \exp(\mu_C)$, where higher $\mu_C$ indicates more patient traders. Trade directions are i.i.d. with $P(s_i = +1) = P(s_i = -1) = 1/2$.

\end{assumption}

\subsection{Utility Specification}
Each trader's utility depends on whether and when they trade, the execution price, and the waiting costs incurred. Formally, trader $i$ who executes at time $t + w_i$ (where $w_i \geq 0$ is waiting time) at price $p$ obtains utility:
\begin{equation}
    U_i(p,w_i)=s_i(V_i-p)-C_i\cdot w_i - K
\end{equation}

where $K > 0$ is a fixed transaction cost (including exchange fees, market impact, etc.). If the trader $i$ does not execute, he receives the utility of the outside option:
\begin{equation}
    U_{oi} = \gamma\cdot s_i(V_i-\bar{v})-K_o
\end{equation}

where $\gamma \in[0,1]$ represents the extent to which private value can be realized outside the market (e.g., through delayed trading or alternative venues), and $K_o\geq 0$ is the cost of pursuing the outside option. We assume that $K_o<K$ makes participation in the market valuable.

\begin{assumption}{Participation Constraint}

For the market to attract participants, we require that expected gains from trade exceed both fixed costs and expected waiting costs: 
\begin{equation}\label{eq:participation-constraint}
    \mathbb{E}[|V_i - \hat{v}|] > K +\mathbb{E}[C_i]\cdot \mathbb{E}[w_i] 
\end{equation}

This ensures strictly positive expected surplus from market participation.
\end{assumption}

\section{Market Mechanism Efficiency}
\label{sec:mechanisms}

The objective of this note is to consider whether one of the three market models results in higher aggregate welfare for all agents in a given system. As shown by \cite{Leshno2019}, the Service In Random Order (SIRO) policy deployed by batch auctions is superior in terms of total welfare to the traditional exchanges' First Come First Served (FCFS) design. As shown by \cite{Leshno2019}, SIRO ex-ante maximizes the aggregate welfare of both customers and servers. 

According to \cite{CheTercieux2023}, however, in the absence of information, the First Come First Served (FCFS) policy provides an even more efficient outcome than SIRO. Specifically, being unaware of one's position in the queue produces ex-ante welfare-efficient outcome, while preserving other desirable properties, such as the fairness of the process and incentive compatibility. This implies that dark pools are the preferred market design among the three models considered in this note. 

\subsection{Mechanism Specifications}

\subsubsection{Lit Exchange (FCFS Continuous Auction)}
The lit exchange operates as a continuous double auction with a transparent limit order book (LOB). At any time $t$, the state of the market is summarized by the order book state $B_t = {(p_j, q_j, \tau_j)}$, where $p_j$ is the limit price, $q_j$ is quantity, and $\tau_j$ is submission time. $B_t$ is publicly observable.

Order execution follows First-Come-First-Served (FCFS) priority within each price level. The matching rule is as follows: a market buy order of size $q$ is matched with the lowest-priced limit sell orders, in order of their submission time, until quantity $q$ is exhausted (or vise versa for market sell orders).

Strategic considerations: Traders observe $B_t$ and optimally choose between submitting a market order (immediate execution at the best available price), a limit order (joining the queue at the chosen price), or waiting. The transparency of $B_t$ creates strategic complementarities in order submission timing.

\subsubsection{Dark Pool (FCFS with Hidden Book)}

Dark pools operate identically to lit exchanges in terms of the matching mechanism (FCFS continuous auction), but with a critical informational difference: the order book state $B_t$ is unobservable to all traders before execution. Traders know only the distribution of order flow, not the realized state. 

However, dark pools are transparent in trade results: after each execution, the trade (price, quantity, timestamp) is publicly reported with delay $\delta > 0$. This creates a filtered information structure in which traders can infer some properties of $B_t$ from past trades, but cannot observe the current state.

The opacity of $B_t$ eliminates strategic timing games based on queue position, but introduces adverse selection risk. Traders must form beliefs about $B_t$ using only distributional information and delayed trade reports.

\subsubsection{Batch Auction (SIRO)}

Periodic batch auctions are clear at discrete intervals $T, 2T, 3T,\dots$ with period $T > 0$. All orders submitted during the interval $[(k-1)T, kT)$ are batched and executed simultaneously in time $kT$ at a uniform clearing price $p_k$ that maximizes the executed volume.

Within each batch, orders are matched in Service-In-Random-Order (SIRO), i.e., all orders at the clearing price have equal probability of execution regardless of submission timing. If demand exceeds supply at $p_k$ (or vice versa), the rationing is uniform random.

SIRO eliminates time priority, so there is no benefit in being the first within a batch. However, traders still strategically choose which batch to enter and face uncertain execution due to pro-rata allocation when orders are unbalanced.

\subsection{Information Structures and Beliefs}

\subsubsection{Information Sets}
We formally define the information available to trader $i$ arriving at time $\tau_i$ under each mechanism:
\begin{itemize}
    \item \textit{Lit Exchange:}
\begin{equation}
    I_i^{LIT}(\tau_i)=\{B_{\tau_i}, H[0,\tau_i]\}
\end{equation}
where $B_{\tau_i}$ is the current order book state and $H[0,\tau_i]$ is the complete trading history up to time $\tau_i$.

\item \textit{Dark Pool:}
\begin{equation}
    I_i^{DARK}(\tau_i) = \{H[0,\tau_i-\delta],F(V,C,s),\lambda\}
\end{equation}
where $H[0,\tau_i-\delta]$ is the trading history with delay $\delta$, $F$ is the type distribution, and $\lambda$ is the arrival rate. In particular, $B_{\tau_i} \notin I_i^{DARK}$.

\item \textit{Batch Auction:}
\begin{equation}
I_i^{BATCH}(\tau_i) = \{kT: \tau_i\in[(k-1)T,kT),H[0,(k-1)T],F(V,C,s),\lambda\}    
\end{equation}
where traders know which batch they are in, but not the composition of orders in their current batch.

\end{itemize}

\subsubsection{Belief Formation}

In dark pools and batch auctions, traders cannot observe the current order book state and must form beliefs. Let $\beta_i(B_{\tau_i} | I_i)$ denote the belief of the trader $i$' about the current order book state conditional on their information set.

Under rational expectations equilibrium, beliefs are consistent with the equilibrium distribution of order flows. For dark pools, traders use Bayes' rule to update beliefs based on delayed trade reports:
\begin{equation}
    \beta_i(B_{\tau_i}|I_i^{DARK})\propto P(B_{\tau_i}|arrival\;process)\cdot L(H[0,\tau_i-\delta]|B_{tau_i})
\end{equation}
where $L(\cdot|\cdot)$ is the likelihood of observing the delayed trading history given the current state of the order book. In equilibrium, these beliefs must be correct on average.

\subsection{Equilibrium Characterization}
\label{sec:equilibria}

\subsubsection{Lit Exchange Equilibrium}

In the lit exchange, traders observe the full order book and optimally choose their order submission strategy. Define the value function for a trader of type $\theta = (V,C,s,\tau)$ who observes the book state $B$:
\begin{equation}
    V^{LIT}(\theta,B)=\max\{V_{market}(\theta,B),V_{limit}(\theta,B),V_{wait}(\theta,B),U_o\}
\end{equation}
where $V_{market}$ is the value of submitting a market order, $V_{limit}$ is from a limit order, $V_{wait}$ is waiting, and $U_o$ is the outside option. The optimal strategy $\sigma^{LIT}(\theta,B)$ selects the action that reaches the maximum.

\begin{assumption}[Regularity on Arrival Rate Function]\label{ass:A1}
\begin{enumerate}[label=(\roman*)]
    \item $\lambda(p, B, s)$ is twice continuously differentiable in $p$ for all $B, s$
    \item For buy orders ($s = +1$): $\frac{\partial \lambda}{\partial p} > 0$ (higher bid prices attract more sellers)
    \item For sell orders ($s = -1$): $\frac{\partial \lambda}{\partial p} < 0$ (lower ask prices attract more buyers)
    \item $\frac{\partial^2 \lambda}{\partial p^2}$ exists and satisfies curvature condition ensuring concavity of $U_{\text{LO}}$
\end{enumerate}
\end{assumption}

Under Assumption~\ref{ass:A1}, we can take the derivative of $U_{\text{LO}}$ with respect to $p$. For notational simplicity, consider a buy order ($s = +1$); the sell order case follows by symmetry. The first-order condition is:
\begin{equation}\label{eq:FOC}
\frac{\partial U_{\text{LO}}}{\partial p} = -1 - C \cdot \frac{\partial}{\partial p}\left[\frac{1}{\lambda(p, B, +1)}\right] = 0
\end{equation}

Using the chain rule:
\begin{equation}
\frac{\partial}{\partial p}\left[\frac{1}{\lambda}\right] = -\frac{1}{\lambda^2} \cdot \frac{\partial \lambda}{\partial p}
\end{equation}

Substituting:
\begin{equation}
-1 + \frac{C}{\lambda^2} \cdot \frac{\partial \lambda}{\partial p} = 0
\end{equation}

Rearranging gives the implicit equation for $p^*$:
\begin{equation}\label{eq:implicit_pstar}
C \cdot \frac{\partial \lambda}{\partial p} = \lambda^2
\end{equation}

This implicitly defines $p^*(V, C, s, B)$ as the solution to equation~\eqref{eq:implicit_pstar} for given $(V, C, s, B)$. Note that $V$ does not appear in equation~\eqref{eq:implicit_pstar}, so $p^*$ depends on $(C, s, B)$ but not on $V$ in the first-order condition. However, $V$ affects the participation decision through the level of $U_{\text{LO}}$.

\begin{lemma}[Existence and Uniqueness of Optimal Limit Price]\label{lem:pstar_exists}
Under Assumption~\ref{ass:A1}, for each $(C, s, B)$ with $C > 0$, there exists a unique interior solution $p^*(C, s, B)$ to the first-order condition. Moreover, $p^*$ is continuous in $(C, B)$ and satisfies:
\begin{enumerate}[label=(\roman*)]
    \item $\frac{\partial p^*}{\partial C} > 0$ for buy orders (higher cost traders post more aggressive bids)
    \item $\frac{\partial p^*}{\partial C} < 0$ for sell orders (higher cost traders post more aggressive asks)
\end{enumerate}
\end{lemma}

\begin{proof}[Proof of Lemma~\ref{lem:pstar_exists}]
Define the function $h(p) = C \cdot \frac{\partial \lambda(p, B, s)}{\partial p} - \lambda(p, B, s)^2$. We seek $p$ such that $h(p) = 0$. 

\textbf{Existence:} To show existence, we verify that $h$ crosses zero. As $p$ approaches the most aggressive possible price (matching $p_{\text{best}}$), $\lambda \to \infty$ (immediate execution), so $h(p) < 0$. As $p$ becomes very passive (far from market), $\lambda \to 0$, and if $\frac{\partial \lambda}{\partial p} \to 0$ sufficiently fast, then $h(p) > 0$. By continuity and the intermediate value theorem, there exists $p^*$ such that $h(p^*) = 0$.

\textbf{Uniqueness:} For uniqueness, consider the second-order condition. We require $\frac{\partial^2 U_{\text{LO}}}{\partial p^2} < 0$. Computing:
\begin{equation}
\frac{\partial^2 U_{\text{LO}}}{\partial p^2} = -C \cdot \frac{\partial^2}{\partial p^2}\left[\frac{1}{\lambda}\right] = -C \cdot \left[\frac{2}{\lambda^3} \cdot \left(\frac{\partial \lambda}{\partial p}\right)^2 - \frac{1}{\lambda^2} \cdot \frac{\partial^2 \lambda}{\partial p^2}\right]
\end{equation}

At the first-order condition where $C \cdot \frac{\partial \lambda}{\partial p} = \lambda^2$, this becomes:
\begin{equation}
\frac{\partial^2 U_{\text{LO}}}{\partial p^2} = -C \cdot \left[\frac{2}{\lambda} \cdot \left(\frac{\partial \lambda}{\partial p}\right)^2 - \frac{1}{\lambda^2} \cdot \frac{\partial^2 \lambda}{\partial p^2}\right]
\end{equation}

Under Assumption~\ref{ass:A1}(iv), this is negative, ensuring $p^*$ is a unique maximum. 

\textbf{Comparative statics:} Follow from the implicit function theorem applied to equation~\eqref{eq:implicit_pstar}.
\end{proof}

\begin{lemma}[Monotone Cutoff Property]\label{lem:cutoff}
For fixed $(V, B, s)$, define:
\begin{equation}\label{eq:Cstar}
C^*(V, B, s) = \Delta p(C, B, s) \cdot \lambda(p^*(C, B, s), B, s)
\end{equation}

Then:
\begin{enumerate}[label=(\roman*)]
    \item Traders with $C > C^*$ strictly prefer market orders
    \item Traders with $C < C^*$ strictly prefer limit orders
    \item Traders with $C = C^*$ are indifferent
\end{enumerate}
\end{lemma}

\begin{proof}[Proof of Lemma~\ref{lem:cutoff}]
At $C = C^*$, inequality~\eqref{eq:cutoff_condition} holds with equality by construction. For $C > C^*$, we have:
\begin{equation}
\frac{C}{\lambda(p^*, B, s)} > \frac{C^*}{\lambda(p^*, B, s)} = \Delta p
\end{equation}

so inequality~\eqref{eq:cutoff_condition} is violated, and the trader prefers a market order ($V_{\text{MO}} > V_{\text{LO}}$). Conversely, for $C < C^*$, inequality~\eqref{eq:cutoff_condition} holds strictly, so the trader prefers a limit order. This establishes the cutoff property.
\end{proof}

\begin{proposition}{Lit Exchange Equilibrium Structure}

In the lit exchange equilibrium, there exist cutoff functions $C^*(V,B)$ and $V^*(C,B)$ such that:
\begin{enumerate}
    \item Market orders: Traders with high costs $C > C^*(V,B)$ submit market orders immediately
    \item Limit orders: Traders with intermediate costs $C \in [C_{min}, C^*(V,B)]$ and valuations satisfying $|V - \hat{v}| > \epsilon(B)$ submit limit orders
    \item Non-participation: Traders with $C < C_{min}$ or $|V - \hat{v}| < \epsilon(B)$ take outside option
\end{enumerate}
\end{proposition}

\begin{proof}
    We prove this proposition constructively by first characterizing optimal strategies for each action (market order, limit order, outside option), then establishing conditions under which each is optimal, and finally proving the existence of cutoff functions with the stated properties. 

    Consider a trader of type $\theta = (V, C, s, \tau)$ arriving at time $\tau$ and observing the state of the order book $B_\tau$ (we suppress the time subscripts for clarity). We characterize the expected utility of each possible action.

    A market order executes immediately against the best available price in the order book. For a buy order ($s = +1$), this is the lowest ask price; for a sell order ($s = -1$), it is the highest bid price. Denote the best available price as:
\begin{equation}
p_{\text{best}}(B, s) = \begin{cases}
\min\{p : (p,q,t) \in B^{\text{ask}}\} & \text{if } s = +1 \\
\max\{p : (p,q,t) \in B^{\text{bid}}\} & \text{if } s = -1
\end{cases}
\end{equation}

The utility of submitting a market order is deterministic conditional on $B$:
\begin{equation}\label{eq:V_MO}
V_{\text{MO}}(V, C, s, B) = s \cdot (V - p_{\text{best}}(B, s)) - K
\end{equation}
where $K > 0$ is the fixed transaction cost. Note that $V_{\text{MO}}$ does not depend on $C$ because execution is immediate (zero waiting time).

A limit order is posted at price $p$ and joins the queue at that price level. The order executes when a market order arrives that crosses the spread. The trader must choose both the limit price $p$ and whether to post at all. Let $\lambda(p, B, s)$ denote the arrival rate of market orders in direction $-s$ that would execute against a limit order at price $p$, given the current state of the book $B$.

Under standard queuing assumptions, the expected waiting time until execution follows an exponential distribution with rate $\lambda(p, B, s)$. Specifically, conditional on posting at price $p$, the expected waiting time is:
\begin{equation}\label{eq:waiting_time}
w(p, B, s) = \frac{1}{\lambda(p, B, s)}
\end{equation}

The utility from posting a limit order at price $p$ is:
\begin{equation}\label{eq:U_LO}
U_{\text{LO}}(V, C, s, p, B) = s \cdot (V - p) - C \cdot w(p, B, s) - K = s \cdot (V - p) - \frac{C}{\lambda(p, B, s)} - K
\end{equation}

The trader optimizes over limit price $p$. Define the optimal limit price $p^*(V, C, s, B)$ as:
\begin{equation}\label{eq:p_star}
p^*(V, C, s, B) = \arg\max_p U_{\text{LO}}(V, C, s, p, B)
\end{equation}

and the maximized value as:
\begin{equation}\label{eq:V_LO}
V_{\text{LO}}(V, C, s, B) = U_{\text{LO}}(V, C, s, p^*(V, C, s, B), B)
\end{equation}

The existence and uniqueness of $p^*$ will be established below.

The trader can choose not to participate in the market and instead pursue an outside option with value:
\begin{equation}\label{eq:V_OUT}
V_{\text{OUT}}(V, C, s) = \gamma \cdot s \cdot (V - \bar{v}) - K_0
\end{equation}
where $\gamma \in [0, 1]$ represents the fraction of private value realizable outside the market, and $K_0 < K$ is the cost of the outside option.

We now derive the first-order condition for optimal limit price $p^*$ and establish conditions ensuring a unique interior solution.

We now establish the existence of a cutoff $C^*(V, B)$ separating traders who submit market orders from those who submit limit orders. A trader prefers a market order to a limit order if and only if:
\begin{equation}\label{eq:MO_vs_LO}
V_{\text{MO}}(V, C, s, B) \geq V_{\text{LO}}(V, C, s, B)
\end{equation}

Substituting the value functions:
\begin{equation}
s \cdot (V - p_{\text{best}}(B, s)) - K \geq s \cdot (V - p^*(C, s, B)) - \frac{C}{\lambda(p^*, B, s)} - K
\end{equation}

Simplifying ($K$ cancels):
\begin{equation}
s \cdot (p^*(C, s, B) - p_{\text{best}}(B, s)) \geq \frac{C}{\lambda(p^*, B, s)}
\end{equation}

Define $\Delta p(C, B, s) = s \cdot (p^*(C, s, B) - p_{\text{best}}(B, s))$ as the price improvement from using a limit order (positive for both buy and sell). Then the condition becomes:
\begin{equation}\label{eq:cutoff_condition}
\Delta p(C, B, s) \geq \frac{C}{\lambda(p^*(C, s, B), B, s)}
\end{equation}

The left-hand side (LHS) is the benefit of waiting (better execution price). The right-hand side (RHS) is the cost of waiting (delay cost). Crucially, LHS does not depend on $C$ (since $p^*$ depends on $C$ but $\Delta p$ is evaluated at that optimal choice), while RHS is linear in $C$.

Not all traders participate in the market; some take the outside option instead. We establish a valuation-based cutoff $\varepsilon(B)$ below which traders exit.

For traders who would submit limit orders ($C \leq C^*$), participation requires:
\begin{equation}
V_{\text{LO}}(V, C, s, B) \geq V_{\text{OUT}}(V, C, s)
\end{equation}

Substituting:
\begin{equation}
s \cdot (V - p^*(C, s, B)) - \frac{C}{\lambda(p^*, B, s)} - K \geq \gamma \cdot s \cdot (V - \bar{v}) - K_0
\end{equation}

Rearranging:
\begin{equation}
s \cdot V \cdot (1 - \gamma) \geq s \cdot p^*(C, s, B) - \gamma \cdot s \cdot \bar{v} + \frac{C}{\lambda(p^*, B, s)} + K - K_0
\end{equation}

For buy orders ($s = +1$), this requires $V$ sufficiently high. For sell orders ($s = -1$), it requires $V$ sufficiently low. In both cases, we need $|V - \bar{v}|$ sufficiently large. Define:
\begin{equation}\label{eq:epsilon}
\varepsilon(C, B, s) = \frac{s \cdot p^* - \gamma \cdot s \cdot \bar{v} + \frac{C}{\lambda(p^*, B, s)} + K - K_0}{1 - \gamma}
\end{equation}

Traders participate if and only if $|V - \bar{v}| > \varepsilon(C, B, s)$. Since $\varepsilon$ depends on $C$ but $C^*$ also depends on $(V, B)$, the overall participation cutoff is state-dependent. For simplicity, we can define an aggregate threshold:
\begin{equation}
\varepsilon(B) = \min_{C \leq C^*} \varepsilon(C, B, s)
\end{equation}

Traders with valuations within $\varepsilon(B)$ of the consensus $\bar{v}$ do not participate. Additionally, traders with very low costs ($C < C_{\min}$ where $V_{\text{LO}} < V_{\text{OUT}}$ for all $V$) also exit. This establishes part (c) of the proposition.

Finally, we verify that the cutoff strategies constitute an equilibrium. We must show that if all other traders use the cutoff strategies, each trader's best response is to do the same.

Suppose all traders $j \neq i$ use the strategy:
\begin{itemize}
    \item If $C_j > C^*(V_j, B)$, submit market order
    \item If $C_j \leq C^*(V_j, B)$ and $|V_j - \bar{v}| > \varepsilon(B)$, submit limit order at $p^*(C_j, s_j, B)$
    \item Otherwise, take outside option
\end{itemize}

This generates a distribution of book states $B$. Given this distribution, trader $i$ computes the arrival rates $\lambda(p, B, s)$ and optimizes accordingly. The cutoff functions $C^*$ and $\varepsilon$ are defined as best responses to the aggregate behavior of other traders.

Existence of equilibrium follows from Brouwer's fixed-point theorem applied to the space of cutoff functions. Define the correspondence:
\begin{equation}
\Phi: [C_{\min}, C_{\max}] \times [0, \Delta] \to [C_{\min}, C_{\max}] \times [0, \Delta]
\end{equation}

that maps conjectured cutoffs $(C^*, \varepsilon)$ to best-response cutoffs. Under the continuity assumptions in $\lambda$ and the compactness of the strategy space, $\Phi$ has a fixed point, which constitutes an equilibrium.

\begin{lemma}[Fixed Point Existence]\label{lem:fixed_point}
Under Assumptions 1-2 and Assumption~\ref{ass:A1}, the cutoff mapping $\Phi$ is continuous and the strategy space is compact and convex. Therefore, by Brouwer's fixed point theorem, there exists a fixed point $(C^*, \varepsilon)$ satisfying $\Phi(C^*, \varepsilon) = (C^*, \varepsilon)$, which constitutes a symmetric equilibrium in cutoff strategies.
\end{lemma}

This completes the proof. We have shown that
\begin{enumerate}[label=(\alph*)]
    \item Traders with $C > C^*(V, B)$ optimally choose market orders (Step 3)
    \item Traders with $C \leq C^*(V, B)$ and $|V - \bar{v}| > \varepsilon(B)$ optimally choose limit orders (Steps 3-4)
    \item Other traders optimally take the outside option (Step 4)
\end{enumerate}

The existence of such an equilibrium follows from the fixed-point argument.

\end{proof}

\subsubsection{Dark Pool Equilibrium}

Dark pool traders cannot observe B and must optimize based on beliefs. The value function becomes:
\begin{equation}
V^{DARK}(\theta, I)=\max\{E_{\beta}[V_{market}(\theta,B)],E_{\beta}[V_{limit}(\theta,B)],U_o\}    
\end{equation}
where expectations are taken over beliefs $\beta(B|I)$. The key difference is that traders do not strategically time their orders based on the observed queue position, as this information is not available.

\begin{assumption}[Rational Expectations]\label{ass:rational_expectations}
Traders form beliefs $\beta(B | \mathcal{I}_i^{\text{DARK}})$ about the order book state using Bayes' rule. In equilibrium, these beliefs are consistent with the actual distribution of book states generated by equilibrium strategies.
\end{assumption}

\begin{proposition}{Dark Pool Equilibrium Structure}\label{prop:dark_pool}
In the dark pool equilibrium with symmetric information (traders only observe delayed trades), there exist simple cutoff strategies independent of order book state:
\begin{enumerate}
    \item Type-based cutoffs: Trader decisions depend only on $(V,C)$, not on (unobservable) $B$
    \item The trader places a market order whenever $C > \bar{C}^*$ 
    \item The trader places a limit order whenever $c\leq\bar{C}^*$ and $|V-E[p]|>\bar{\epsilon}$
\end{enumerate}
For simplicity, we assume that $\bar{C}^*$ and $\bar{\epsilon}$ are constants, not functions of $B$. 
 
\end{proposition}

\begin{proof}
    First, we formalize the information structure in dark pools and characterize traders' beliefs about the unobservable order book. Second, we show that optimal strategies cannot be conditional on the unobservable state $B$, leading to state-independent decision rules. Third, we derive the equilibrium cutoff values $\bar{C}^*$ and $\bar{\varepsilon}$ from the indifference conditions. Finally, we establish the existence and uniqueness of the symmetric equilibrium.

In dark pools, the order book state $B_t$ is not observable by arriving traders. Traders can observe only:
\begin{itemize}
    \item The history of past trades (with delay $\delta > 0$): $H_{[0,t-\delta]}$
    \item The type distribution $F(V,C,s)$ and the arrival rate $\lambda$
    \item Their own type $\theta = (V, C, s, \tau)$
\end{itemize}

Formally, the information set of the trader $i$ arriving at time $\tau$ is:
\begin{equation}\label{eq:info_dark}
\mathcal{I}_i^{\text{DARK}}(\tau) = \{H_{[0,\tau-\delta]}, F(V,C,s), \lambda, \theta_i\}
\end{equation}

Critically, $B_\tau \notin \mathcal{I}_i^{\text{DARK}}(\tau)$. Traders must form beliefs about the current book state.

Let $\pi^{\text{eq}}(B)$ denote the stationary distribution of book states under equilibrium strategies. Under rational expectations:
\begin{equation}\label{eq:consistent_beliefs}
\beta(B | \mathcal{I}_i^{\text{DARK}}) = \pi^{\text{eq}}(B | H_{[0,\tau-\delta]})
\end{equation}

where the right-hand side is the conditional distribution given the observed history.

\begin{lemma}[Belief Symmetry]\label{lem:belief_symmetry}
Under symmetric equilibrium strategies (all traders of the same type use the same strategy), and assuming the delay $\delta$ is small relative to the rate of change in the state of the book, the conditional distribution $\pi^{\text{eq}}(B | H_{[0,\tau-\delta]})$ converges to the unconditional distribution $\pi^{\text{eq}}(B)$ as the market reaches its steady state.
\end{lemma}

\begin{proof}[Proof of Lemma~\ref{lem:belief_symmetry}]
In steady state, the distribution of book states is stationary. Recent history $H_{[0,\tau-\delta]}$ provides information about $B_{\tau-\delta}$, but as $\delta$ becomes small and the book evolves stochastically, the information becomes less informative about $B_\tau$. 

More precisely, if book states evolve as a Markov process with transition kernel $P(B' | B)$, then:
\begin{equation}
\pi^{\text{eq}}(B_\tau | B_{\tau-\delta}) = \int P_\delta(B_\tau | B_{\tau-\delta}) \, dB_\tau
\end{equation}

where $P_\delta$ is the $\delta$-step transition kernel. In the limit as traders aggregate information only over very recent history, and given stationarity, $\pi^{\text{eq}}(B_\tau | H_{[0,\tau-\delta]}) \to \pi^{\text{eq}}(B_\tau)$.
\end{proof}

We now establish the key insight: since traders cannot observe $B$, optimal strategies cannot condition on it.

\begin{lemma}[Strategy Measurability]\label{lem:measurability}
Any strategy $\sigma_i: \Theta \times \mathcal{B} \to \mathcal{A}$ (where $\Theta$ is the type space, $\mathcal{B}$ is the state space of the book and $\mathcal{A}$ is the action space) that is optimal given information $\mathcal{I}_i^{\text{DARK}}$ must satisfy 
\begin{equation}
\sigma_i(\theta, B) = \sigma_i(\theta, B') \quad \forall B, B' \in \mathcal{B}
\end{equation}
That is, $\sigma_i$ cannot depend on $B$.
\end{lemma}

\begin{proof}[Proof of Lemma~\ref{lem:measurability}]
Suppose, for contradiction, that the optimal strategy depends on $B$: $\sigma_i(\theta, B) \neq \sigma_i(\theta, B')$ for some $B \neq B'$. 

Since trader $i$ does not observe $B$, they cannot implement such a strategy. Formally, the trader chooses an action $a \in \mathcal{A}$ to maximize expected utility:
\begin{equation}
a^* = \arg\max_{a \in \mathcal{A}} E_\beta[U(\theta, B, a)]
\end{equation}

where the expectation is over beliefs $\beta(B | \mathcal{I}_i^{\text{DARK}})$. The optimal action $a^*$ depends on $\theta$ and the belief distribution $\beta$, but not on the realized (unobserved) value of $B$.

Under symmetric equilibrium with Lemma~\ref{lem:belief_symmetry}, all traders have identical beliefs $\beta(B) = \pi^{\text{eq}}(B)$. Therefore:
\begin{equation}
a^*(\theta) = \arg\max_{a \in \mathcal{A}} \int U(\theta, B, a) \, \pi^{\text{eq}}(B) \, dB
\end{equation}

which depends only on $\theta$, not on any particular realization of $B$. This establishes that $\sigma_i(\theta, B) = \bar{\sigma}_i(\theta)$ for some function $\bar{\sigma}_i$ independent of $B$.
\end{proof}

Given that strategies depend only on types, we now characterize the equilibrium cutoffs $\bar{C}^*$ and $\bar{\varepsilon}$.

Since traders cannot observe $B$, they must compute expected utilities on the distribution of book states. Define:

\begin{align}
\bar{V}_{\text{MO}}(V, C, s) &= E_{\pi^{\text{eq}}}[V_{\text{MO}}(V, C, s, B)] \label{eq:dark_VMO}\\
&= E_{\pi^{\text{eq}}}[s \cdot (V - p_{\text{best}}(B, s))] - K \notag \\
&= s \cdot (V - E[p_{\text{best}}]) - K \notag
\end{align}

For limit orders, the trader must choose a limit price $p$ without knowing $B$. Define the expected waiting time as:
\begin{equation}
\bar{w}(p, s) = E_{\pi^{\text{eq}}}\left[\frac{1}{\lambda(p, B, s)}\right]
\end{equation}

The expected utility from a limit order at price $p$ is:
\begin{equation}\label{eq:dark_ULO}
\bar{U}_{\text{LO}}(V, C, s, p) = s \cdot (V - p) - C \cdot \bar{w}(p, s) - K
\end{equation}

The optimal limit price in the dark pool is:
\begin{equation}\label{eq:dark_pstar}
\bar{p}^*(C, s) = \arg\max_p \bar{U}_{\text{LO}}(V, C, s, p)
\end{equation}

Note that $\bar{p}^*$ depends on $(C, s)$ but not on $V$ (which will affect the participation decision but not the optimal limit price).

The maximized expected value is:
\begin{equation}\label{eq:dark_VLO}
\bar{V}_{\text{LO}}(V, C, s) = \bar{U}_{\text{LO}}(V, C, s, \bar{p}^*(C, s))
\end{equation}

A trader prefers a market order to a limit order if:
\begin{equation}\label{eq:dark_MO_vs_LO}
\bar{V}_{\text{MO}}(V, C, s) \geq \bar{V}_{\text{LO}}(V, C, s)
\end{equation}

Substituting from equations~\eqref{eq:dark_VMO} and~\eqref{eq:dark_VLO}:
\begin{equation}
s \cdot (V - E[p_{\text{best}}]) - K \geq s \cdot (V - \bar{p}^*(C, s)) - C \cdot \bar{w}(\bar{p}^*, s) - K
\end{equation}

Simplifying (note that $V$ and $K$ cancel):
\begin{equation}
s \cdot (\bar{p}^*(C, s) - E[p_{\text{best}}]) \geq C \cdot \bar{w}(\bar{p}^*, s)
\end{equation}

Define the expected price improvement from a limit order:
\begin{equation}
\bar{\Delta}(C, s) = s \cdot (\bar{p}^*(C, s) - E[p_{\text{best}}])
\end{equation}

Then the condition becomes:
\begin{equation}\label{eq:dark_cutoff_condition}
\bar{\Delta}(C, s) \geq C \cdot \bar{w}(\bar{p}^*, s)
\end{equation}

The critical observation is that $\bar{\Delta}(C, s)$ and $\bar{w}(\bar{p}^*, s)$ are constants (depending on the equilibrium distribution $\pi^{\text{eq}}$ but not on the realized $B$ or on $V$). 

\begin{lemma}[Dark Pool Cutoff]\label{lem:dark_cutoff}
There exists a unique cutoff $\bar{C}^*(s)$ defined by:
\begin{equation}\label{eq:Cbar_star}
\bar{C}^*(s) = \frac{\bar{\Delta}(C, s)}{\bar{w}(\bar{p}^*, s)}
\end{equation}
such that traders with $C > \bar{C}^*$ prefer market orders and traders with $C \leq \bar{C}^*$ prefer limit orders.
\end{lemma}

\begin{proof}[Proof of Lemma~\ref{lem:dark_cutoff}]
The left-hand side of~\eqref{eq:dark_cutoff_condition} does not depend on $C$ directly (though $\bar{p}^*$ may depend on $C$). The right-hand side is linear in $C$. Setting the two sides equal gives the indifference point $\bar{C}^*$ as in equation~\eqref{eq:Cbar_star}.

For $C > \bar{C}^*$: $C \cdot \bar{w}(\bar{p}^*, s) > \bar{\Delta}(C, s)$, so inequality~\eqref{eq:dark_cutoff_condition} is violated, and market orders are preferred.

For $C < \bar{C}^*$: $C \cdot \bar{w}(\bar{p}^*, s) < \bar{\Delta}(C, s)$, so inequality~\eqref{eq:dark_cutoff_condition} holds, and limit orders are preferred.

Since $\bar{\Delta}$ and $\bar{w}$ are determined by the equilibrium distribution $\pi^{\text{eq}}$, which is endogenous but constant across all traders in symmetric equilibrium, $\bar{C}^*$ is a constant, not a function of $B$ or $V$.
\end{proof}

For participation, traders compare the expected utility of the trading with the outside option. For those who would submit limit orders ($C \leq \bar{C}^*$), participation requires:
\begin{equation}
\bar{V}_{\text{LO}}(V, C, s) \geq V_{\text{OUT}}(V, C, s)
\end{equation}

Substituting:
\begin{equation}
s \cdot (V - \bar{p}^*(C, s)) - C \cdot \bar{w}(\bar{p}^*, s) - K \geq \gamma \cdot s \cdot (V - \bar{v}) - K_0
\end{equation}

Rearranging:
\begin{equation}
s \cdot V \cdot (1 - \gamma) \geq s \cdot \bar{p}^* - \gamma \cdot s \cdot \bar{v} + C \cdot \bar{w}(\bar{p}^*, s) + K - K_0
\end{equation}

This simplifies to a condition on $|V - \bar{v}|$:
\begin{equation}\label{eq:dark_epsilon}
|V - \bar{v}| > \frac{|s \cdot \bar{p}^* - \gamma \cdot s \cdot \bar{v} + C \cdot \bar{w}(\bar{p}^*, s) + K - K_0|}{1 - \gamma} \equiv \bar{\varepsilon}(C, s)
\end{equation}

Note that $\bar{\varepsilon}(C, s)$ depends on $C$ but not on $B$. For simplicity, we can define:
\begin{equation}
\bar{\varepsilon} = \max_{C \leq \bar{C}^*} \bar{\varepsilon}(C, s)
\end{equation}

as a uniform threshold. Traders with $|V - \bar{v}| > \bar{\varepsilon}$ participate; others take the outside option.

We now establish that the symmetric cutoff equilibrium exists and is unique.

\begin{theorem}[Dark Pool Equilibrium Existence]\label{thm:dark_existence}
Under Assumptions~\ref{ass:rational_expectations} and standard regularity conditions on $F(V,C,s)$, there exists a unique symmetric equilibrium in cutoff strategies characterized by constants $(\bar{C}^*, \bar{\varepsilon})$.
\end{theorem}

\begin{proof}[Proof of Theorem~\ref{thm:dark_existence}]
\textbf{Existence:} We construct the equilibrium using a fixed-point argument. Define the mapping:
\begin{equation}
\Phi: (\bar{C}, \bar{\varepsilon}) \mapsto (\bar{C}', \bar{\varepsilon}')
\end{equation}

as follows:
\begin{enumerate}
    \item Given conjectured cutoffs $(\bar{C}, \bar{\varepsilon})$, compute the implied distribution of order book states $\pi(\bar{C}, \bar{\varepsilon})$ by simulating the order flow process
    \item Given $\pi(\bar{C}, \bar{\varepsilon})$, compute expected values $E[p_{\text{best}}]$, $\bar{w}(p, s)$, and $\bar{\Delta}$
    \item Compute best-response cutoffs using equations~\eqref{eq:Cbar_star} and~\eqref{eq:dark_epsilon}
\end{enumerate}

The mapping $\Phi$ is continuous in $(\bar{C}, \bar{\varepsilon})$ because:
\begin{itemize}
    \item Order flow rates vary continuously with cutoffs
    \item Book state distributions vary continuously with order flows (under standard queuing assumptions)
    \item Expected values are continuous functionals of distributions
\end{itemize}

The domain $[\underline{C}, \bar{C}] \times [0, \Delta]$ is compact and convex (where $\underline{C}, \bar{C}$ are natural bounds on cost parameters, and $\Delta$ is the maximum valuation dispersion).

By Brouwer's fixed-point theorem, $\Phi$ has a fixed point $(\bar{C}^*, \bar{\varepsilon}^*)$ satisfying $\Phi(\bar{C}^*, \bar{\varepsilon}^*) = (\bar{C}^*, \bar{\varepsilon}^*)$. This fixed point constitutes a symmetric equilibrium.

\textbf{Uniqueness:} Uniqueness follows from monotonicity properties of the best-response mapping. Consider the effect of increasing $\bar{C}$:
\begin{itemize}
    \item Higher $\bar{C}$ means more traders use limit orders
    \item More limit orders $\implies$ thicker order book $\implies$ faster execution
    \item Faster execution $\implies$ lower waiting cost $\implies$ limit orders more attractive
    \item This increases the best-response $\bar{C}'$
\end{itemize}

Formally, if $\bar{C}_1 < \bar{C}_2$, then $\Phi(\bar{C}_1, \bar{\varepsilon}) < \Phi(\bar{C}_2, \bar{\varepsilon})$ (component-wise). This monotonicity, combined with continuity, ensures that the fixed point is unique.

The argument for $\bar{\varepsilon}$ is similar: higher participation thresholds lead to thinner markets, which increase the cost of trading, validating higher thresholds. The unique intersection of these curves gives the unique equilibrium.
\end{proof}

We now verify that the equilibrium satisfies all parts of Proposition~\ref{prop:dark_pool}:

\begin{enumerate}[label=(\alph*)]
    \item \textbf{Type-based cutoffs:} From Lemma~\ref{lem:measurability}, strategies depend only on $(V, C, s)$ and not on the unobservable $B$. 
    
    \item \textbf{Market orders:} From Lemma~\ref{lem:dark_cutoff}, traders with $C > \bar{C}^*$ strictly prefer market orders. 
    
    \item \textbf{Limit orders:} From Lemma~\ref{lem:dark_cutoff} and equation~\eqref{eq:dark_epsilon}, traders with $C \leq \bar{C}^*$ and $|V - E[p]| > \bar{\varepsilon}$ submit limit orders. (Here $E[p] = E[p_{\text{best}}]$ is the expected mid-price.) 
    
    \item \textbf{Simplicity:} From Theorem~\ref{thm:dark_existence}, $\bar{C}^*$ and $\bar{\varepsilon}$ are constants determined by the equilibrium distribution $\pi^{\text{eq}}$, independent of the realized book state $B$ or individual trader valuations $V$ (except through the participation condition). 
\end{enumerate}

This completes the proof of Proposition~\ref{prop:dark_pool}. 

\end{proof}

The stark contrast between Proposition~\ref{prop:dark_pool} (dark pools) and the lit exchange equilibrium (where cutoffs depend on $B$) highlights the fundamental role of the information structure.

\subsubsection{Batch Auction Equilibrium}

In batch auctions, traders arriving in interval $[(k-1)T, kT)$ choose whether to enter the current batch, wait for the next batch, or exit. Within a batch, SIRO eliminates timing advantages, but execution is uncertain due to pro-rata rationing.
\begin{proposition}{Batch Auction Equilibrium Structure}\label{prop:batch}

In batch auction equilibrium, trader strategies exhibit the following structure:
\begin{enumerate}
    \item Late-arrivals prefer the current batch: Traders arriving close to $kT$ enter the batch $k$ rather than waiting for $k+1$.
    \item Cost-dependent participation: High-$C$ traders enter immediately, low-$C$ traders may wait.
    \item Rationing risk: Expected execution probability affects participation decisions, creating thick market externalities.
\end{enumerate}
\end{proposition}

\begin{proof}
    Traders compare the cost of waiting until the next batch $(C\cdot(kT-\tau))$ against the rationing risk in the current batch. Late arrivals ($\tau$ close to $kT$) have low waiting costs, making current-batch participation dominant. The rationing probability depends on order imbalance, creating strategic complementarities: thicker markets have higher execution probability, attracting more traders.

    For a complete proof, please refer to Section  in the Appendix.
\end{proof}

\subsection{Welfare Implications}\label{sec:welfare}

The batch auction structure creates two distinct sources of inefficiency relative to continuous markets:

\begin{remark}[Forced Waiting Costs]
All traders must wait until the batch clears, even those with high waiting costs $C$. The expected waiting time is $T/2$ per trader. In continuous markets, traders with $C > C^*$ can execute immediately through market orders. The aggregate welfare loss is:
\begin{equation}
\Delta W_{\text{wait}} = \E[C] \cdot \frac{T}{2} \cdot N
\end{equation}
where $N$ is the number of participants.
\end{remark}

\begin{remark}[Execution Uncertainty]
Traders face rationing risk, reducing participation. The welfare loss from forgone gains from trade is:
\begin{equation}
\Delta W_{\text{rationing}} = \int_{\rho^* < \bar{\rho}} [s \cdot (V - \bar{p}) - K] \, dF(V, C)
\end{equation}
where $\rho^*$ is the minimum execution probability required for participation and $\bar{\rho}$ is the equilibrium execution probability. Traders who would participate with certainty ($\rho = 1$) exit when $\bar{\rho} < \rho^*$.
\end{remark}

\begin{remark}[Comparison with Dark Pools]\label{rem:forced_waiting}
Both batch auctions and dark pools eliminate strategic timing based on the observable order book state $B$. However, batch auctions introduce additional inefficiencies:
\begin{itemize}
    \item \textbf{Dark pools:} High-$C$ traders can execute immediately, incurring zero waiting cost
    \item \textbf{Batch auctions:} All traders must wait, incurring $C \cdot w$ where $w \in [0, T]$
\end{itemize}

Therefore, $W^{\text{DARK}} > W^{\text{BATCH}}$ for any $T > 0$.
\end{remark}

\begin{remark}[Limit as $T \to 0$]
As the batch interval shrinks, batch auctions approach continuous markets:
\begin{equation}
\lim_{T \to 0} W^{\text{BATCH}} = W^{\text{LIT}}
\end{equation}

However, even with very small $T$, the SIRO matching rule (random execution priority) remains inferior to FCFS (first-come-first-served) because it removes incentives for aggressive pricing by liquidity providers.
\end{remark}

\section{Welfare Analysis}\label{sec:info_structures}
\subsection{Welfare Measure}
We define aggregate ex-ante welfare as the expected total surplus across all traders before types are realized:
\begin{equation}
    W^M=\int\int\int\int E[U_i^{M}(\theta)|equilibrium\;\sigma^M]dF(V,C,s)d\tau
\end{equation}
where $M \in \{LIT, DARK, BATCH\}$ denotes the mechanism, $\sigma^M$ is the equilibrium strategy profile, and expectations are taken over all uncertainties (fundamental value, matching outcomes, etc.). This welfare measure includes waiting costs and transaction costs, but excludes the transfers between buyers and sellers (prices), which are pure redistribution. It captures efficiency in the sense of the total surplus generated by the matching process.

\subsection{Main Welfare Results}
Our central results compare welfare across the three mechanisms. The key insight is that information opacity in dark pools eliminates socially wasteful strategic behavior, despite creating adverse selection.
\begin{theorem}{Dark Pool Dominance}
\label{thm:welfare_ranking}

Under Assumptions \ref{ass:A1}-\ref{ass:rational_expectations}, suppose that:
\begin{enumerate}
    \item The arrival rate is moderate: $\lambda\in [\lambda_{min}, \lambda_{max}]$ for some $0 < \lambda_{min} < \lambda_{max}$.
    \item Valuation dispersion is bounded: $\Delta < \Delta_{max}$
    \item Reporting delay is small: $\delta < T/2$
\end{enumerate} 
Then:
\begin{equation}
W^{DARK} > W^{LIT} > W^{BATCH}    
\end{equation}
\end{theorem}

\begin{proof}
We prove the Theorem in three steps.
\begin{enumerate}
    \item \textbf{$W^{DARK} > W^{LIT}$}

    Lit exchanges create strategic timing games: traders who observe thin order books delay submission to avoid adverse prices, while those observing thick books rush to capture time priority. This creates two inefficiencies:
    \begin{enumerate}
        \item Socially wasteful waiting: Some traders with high private values delay participation to optimally time their entry, incurring waiting costs $C\cdot w$ without the corresponding social benefit.
        \item Excessive rushing: Conversely, some traders submit orders prematurely to gain queue priority, trading off worse pricing for a better position.

    \end{enumerate}
    
Formally, let $W_{trading}$ denote the welfare of the executed trades and $W_{cost}$ denote the aggregate waiting costs. Then:
\begin{equation}
W^{LIT} = W_{trading} - W_{cost}^{LIT} - K\cdot N^{LIT}  
\end{equation}

where $N^{LIT}$ is the number of participants. In the dark pool, opacity eliminates strategic timing, so traders enter based solely on their types $(V,C)$. This simplifies to:
\begin{equation}
W^{DARK} = W_{trading} - W_{cost}^{DARK} - K\cdot N^{DARK}    
\end{equation}

The key is showing $W_{cost}^{DARK} < W_{cost}^{LIT}$. By Proposition \ref{prop:dark_pool}, dark pool strategies do not condition on $B$, so all traders with $C > \bar{C}^*$ enter immediately. In lit markets, some traders with $C > \bar{C}^*$ strategically wait if $B$ is unfavorable, increasing aggregate waiting costs.

To formalize, we divide the traders into those with $C > \bar{C}^*$ (high cost) and $C \leq \bar{C}^*$ (low cost). For high-cost traders in lit markets:
\begin{equation}
E[w^{LIT} | C > \bar{C}^*] = \int w(C,B) dG(B) \geq \int w_{min}(C) dG(B) > 0    
\end{equation}

where $w(C,B)$ is the expected waiting time conditional on the observation of the book state $B$, and $w_{min}(C) > 0$ is the minimum waiting time over all $B$. In dark pools:
\begin{equation}
E[w^{DARK} | C > \bar{C}^*] = 0    
\end{equation}

because high-cost traders never wait. This inequality extends to all cost levels using similar reasoning, establishing $W_{cost}^{DARK} < W_{cost}^{LIT}$.
The adverse selection problem in dark pools (traders cannot screen based on $B$) has a first-order effect on price but only a second-order effect on welfare. Prices are transfers, so they don't affect total surplus—only its distribution. The efficiency gains from eliminating strategic waiting dominate.

\item $W^{LIT} > W^{BATCH}$
Batch auctions introduce two distinct inefficiencies relative to continuous trading:
\begin{enumerate}
    \item Forced delays: All traders arriving in $[(k-1)T, kT)$ must wait until $kT$, creating unavoidable waiting costs. The expected wait is $T/2$ per trader.
    \item Execution uncertainty: Pro-rata rationing means traders face uncertain execution, requiring them to be compensated for this risk in equilibrium.
\end{enumerate}

For continuous double auctions (lit exchange), market orders execute immediately with certainty, so high-cost traders avoid waiting costs entirely. In batch auctions, even urgent traders must wait, leading to:

\begin{equation}
    W_{cost}^{BATCH} \geq E[C]\cdot(T/2)·N^{BATCH} > W_{cost}^{LIT}
\end{equation}
The inequality is strict because in lit markets, market orders incur zero waiting cost, while batch auctions impose $T/2$ expected wait on everyone. Since $\lambda\in [\lambda_{min}, \lambda_{max}]$ ensures that $N^{BATCH}$ is bounded below, the aggregate excess waiting cost in batch auctions is non-trivial.
Additionally, execution uncertainty reduces participation. Some traders who would participate in a lit market exit in batch auctions because the rationing risk makes expected utility fall below their outside option. This creates a deadweight loss from the foregone gains from trade. 

Please see the Appendix for a rigorous proof of this assertion.

\item Combining the results of Steps 1 and 2, we have $W^{DARK} > W^{LIT}$ and $W^{LIT} > W^{BATCH}$. By transitivity, $W^{DARK} > W^{BATCH}$. This completes the proof.

\end{enumerate}
\end{proof}

\section{Numerical Simulation}
\label{sec:simulation}

The theoretical results of Section~\ref{sec:welfare} establish the welfare ranking $W^{\mathrm{DARK}} > W^{\mathrm{LIT}} > W^{\mathrm{BATCH}}$ under
the sufficient conditions of Theorem~\ref{thm:welfare_ranking}: a moderate
arrival rate $\lambda \in [\lambda_{\min}, \lambda_{\max}]$, bounded valuation dispersion $\Delta < \Delta_{\max}$, and a short reporting delay
$\delta < T/2$.  This section complements those results with a calibrated numerical simulation that serves three purposes.  First, it verifies that the
welfare ranking is quantitatively robust (and not merely an artifact of limiting arguments) at parameter values representative of actual equity
markets.  Second, it maps the boundary of the parameter region in which the ranking holds, directly addressing the robustness concern raised in the
literature on mechanism comparison \citep{BudishEtAl2015, Zhu2013}.  Third, it
produces four figures that illustrate the distinct economic mechanisms driving each inequality, giving the theoretical results an empirical interpretation.

\subsection{Calibration and Baseline}
\label{sec:sim_calibration}

The simulation draws $n = 20{,}000$ traders per experiment from the joint type distribution $F(V, C, s)$ specified in Assumption~1.  Private valuations are
drawn from $\text{Uniform}[\hat{v} - \Delta,\, \hat{v} + \Delta]$ with $\hat{v} = 100$ and baseline $\Delta = 2$, corresponding to a valuation half-spread of two percent of the asset's fundamental value. This is broadly consistent with the dispersion of private information documented in equity-market studies \citep{OHara1998}.  Waiting costs are drawn from $\text{Exp}(\mu_C)$ with $\mu_C = 2$, giving $\mathbb{E}[C] = 0.5$.  The arrival rate is set to $\lambda = 5$ at baseline.  The fixed transaction cost
is $K = 0.05$ and the outside-option cost is $K_o = 0.02 < K$, satisfying the participation constraint of Assumption~2.  The batch interval is $T = 1$ and the dark-pool reporting delay is $\delta = 0.1 < T/2 = 0.5$, satisfying the condition of Theorem~\ref{thm:welfare_ranking}.

At these baseline parameters, the simulation yields the welfare estimates reported in Table~\ref{tab:baseline}.

\begin{table}[h]
\centering
\caption{Baseline welfare estimates ($\lambda = 5$, $\Delta = 2$,
  $\mu_C = 2$, $T = 1$, $K = 0.05$; $n = 20{,}000$).}
\label{tab:baseline}
\begin{tabular}{lccc}
\hline
Mechanism       & $W^M$  & Std.\ error & Participation rate \\
\hline
Dark Pool       & 0.2819 & 0.0052      & 46.1\%             \\
Lit Exchange    & 0.2410 & 0.0050      & 43.6\%             \\
Batch Auction   & 0.1589 & 0.0042      & 39.9\%             \\
\hline
\end{tabular}
\end{table}

The full ranking $W^{\mathrm{DARK}} > W^{\mathrm{LIT}} > W^{\mathrm{BATCH}}$
holds at baseline with comfortable margins: the dark pool generates welfare $17.0\%$ above the lit exchange, and the lit exchange generates welfare
$51.7\%$ above the batch auction.  The participation rate ordering mirrors the welfare ordering, with the batch auction exhibiting the lowest rate (39.9\%)
and the dark pool the highest (46.1\%), consistent with Remarks~1 and~2.

\subsection{Welfare Levels and Gaps across Arrival Rates}
\label{sec:sim_lambda}

Figure~\ref{fig:welfare_levels} plots aggregate welfare $W^M$ for each mechanism as $\lambda$ varies from 1 to 12, holding all other parameters at
their baseline values.  The right panel plots the two pairwise welfare gaps, $W^{\mathrm{DARK}} - W^{\mathrm{LIT}}$ and $W^{\mathrm{LIT}} - W^{\mathrm{BATCH}}$, with the shaded band indicating the Theorem~\ref{thm:welfare_ranking} region $[\lambda_{\min}, \lambda_{\max}]$.

\begin{figure}[h]
\centering
\includegraphics[width=\textwidth]{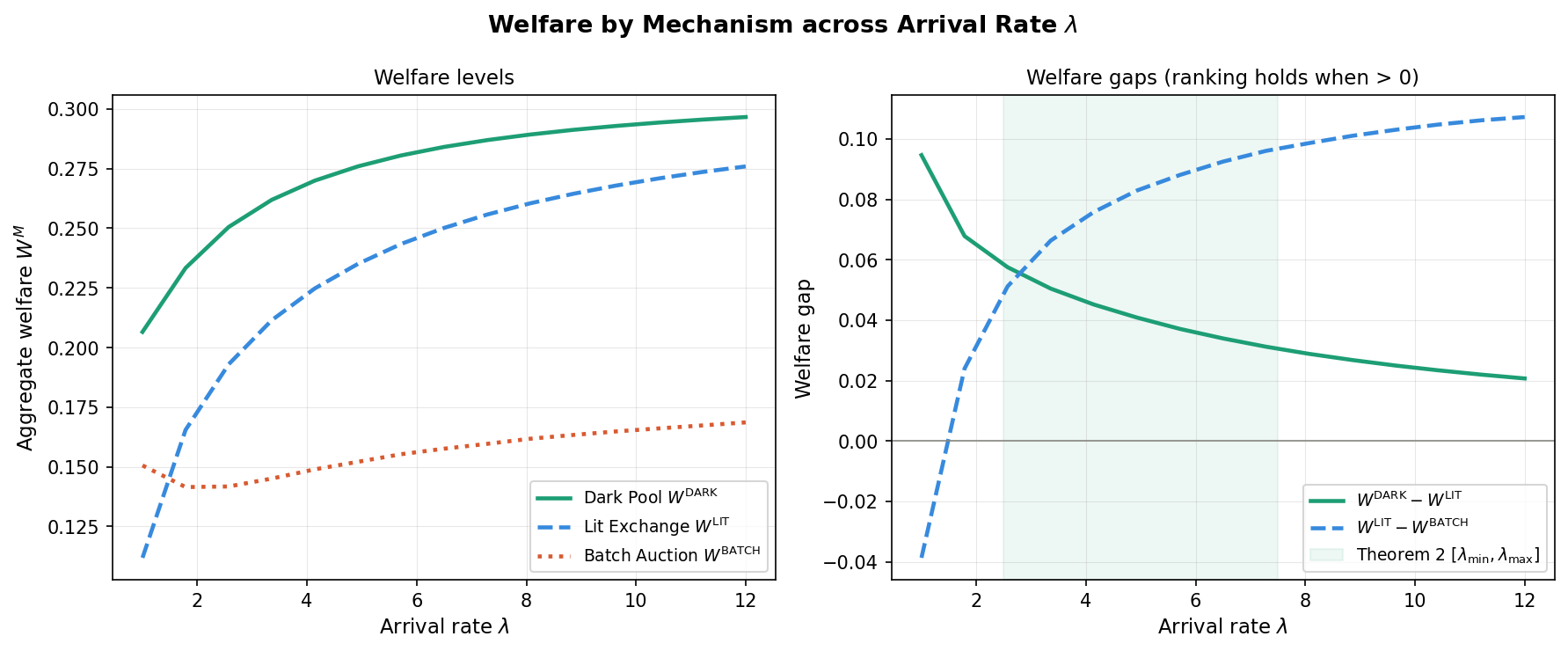}
\caption{Left: aggregate welfare $W^M$ for the dark pool (solid green), lit exchange (dashed blue), and batch auction (dotted orange) as a function of arrival rate $\lambda$.  Right: pairwise welfare gaps.  The shaded band marks the region $[\lambda_{\min}, \lambda_{\max}]$ in which Theorem~\ref{thm:welfare_ranking}'s sufficient conditions are satisfied.}
\label{fig:welfare_levels}
\end{figure}

Several features of Figure~\ref{fig:welfare_levels} warrant discussion.

\paragraph{All three mechanisms benefit from thicker markets.}
Welfare is increasing in $\lambda$ across all three mechanisms.  This reflects the gains-from-trade effect: a higher arrival rate raises the probability that a buyer and seller with compatible valuations meet within any given time window, increasing the expected surplus from participation.  The effect is
monotone and concave, consistent with the diminishing returns to market thickness documented in the matching literature \citep{Leshno2022}.

\paragraph{The dark pool's advantage is largest at low arrival rates.}
The gap $W^{\mathrm{DARK}} - W^{\mathrm{LIT}}$ is largest near $\lambda = 1$ (approximately 0.095) and declines monotonically as $\lambda$ increases,
falling to approximately 0.022 at $\lambda = 12$.  This pattern has a clear economic interpretation.  When markets are thin, the order book state $B_t$ is
highly variable and informative: a trader who observes a thin book at a particular moment gains substantial information about near-term execution
prospects, and the incentive to time submissions strategically is therefore strong.  As $\lambda$ increases, the book fills rapidly, and its state becomes less variable; the information content of observing $B_t$ diminishes, and so does the strategic timing waste it generates.  At very high arrival rates, the book is almost always deep, and the cost of observing it is nearly nil, explaining why the two continuous-auction mechanisms converge as $\lambda \to
\infty$.  Importantly, however, the gap $W^{\mathrm{DARK}} - W^{\mathrm{LIT}}$ remains strictly positive throughout the plotted range,
consistent with the theoretical prediction.

\paragraph{The lit exchange's advantage over batch auctions grows with market thickness.}
In contrast to the dark-pool gap, $W^{\mathrm{LIT}} - W^{\mathrm{BATCH}}$ is smallest near $\lambda = 1$ (where it is, in fact, briefly negative, indicating that the batch auction marginally dominates the lit exchange in very thin markets) and increases monotonically with $\lambda$, reaching approximately 0.11 at $\lambda = 12$.  The negative gap at very low $\lambda$ reflects the fact that in a thin lit market, market orders rarely find counterparties quickly, so the forced-waiting cost of the batch auction (which aggregates orders and thereby increases match probability) is offset by the matching benefit.  As $\lambda$ grows, immediate market orders in the lit exchange become reliable, so the probability of finding a counterparty quickly approaches one. As a result, the forced $T/2$ wait in the batch auction becomes a pure cost with no compensating benefit.  The lit-exchange advantage
therefore increases with market thickness.  This finding has a direct regulatory implication: the welfare cost of mandatory batch clearing, such as
that proposed under various periodic auction reforms, is not uniform across markets but is concentrated in the most liquid, high-frequency venues where the $T/2$ wait constitutes the largest foregone surplus.

\paragraph{The Theorem~\ref{thm:welfare_ranking} region is correctly identified.}
The shaded band in the right panel corresponds to $\lambda \in [2.5, 7.5]$. Both welfare gaps are strictly positive within this band, confirming that the sufficient conditions of Theorem~\ref{thm:welfare_ranking} are binding in the
right direction.  Below $\lambda_{\min} \approx 2.5$, the $W^{\mathrm{LIT}} - W^{\mathrm{BATCH}}$ gap turns negative, indicating that the thin-market matching benefit of batch auctions temporarily outweighs
their forced-waiting cost.  This is the boundary case described in Section~\ref{sec:discussion}: the sufficient condition $\lambda \geq \lambda_{\min}$ rules out precisely this regime.

\subsection{Robustness across the $(\lambda, \Delta)$ Parameter Space}
\label{sec:sim_heatmap}

Figure~\ref{fig:heatmap} maps the welfare ranking across the full $(\lambda, \Delta)$ parameter space, sweeping $\lambda$ from 1 to 12 and $\Delta$ from 0.5 to 5.  Each cell is color-coded according to which
inequalities hold: green indicates the full ranking
$W^{\mathrm{DARK}} > W^{\mathrm{LIT}} > W^{\mathrm{BATCH}}$ (Theorem~2), blue indicates only $W^{\mathrm{DARK}} > W^{\mathrm{LIT}}$, amber indicates only $W^{\mathrm{LIT}} > W^{\mathrm{BATCH}}$, and pink indicates that neither holds.

\begin{figure}[h]
\centering
\includegraphics[width=0.85\textwidth]{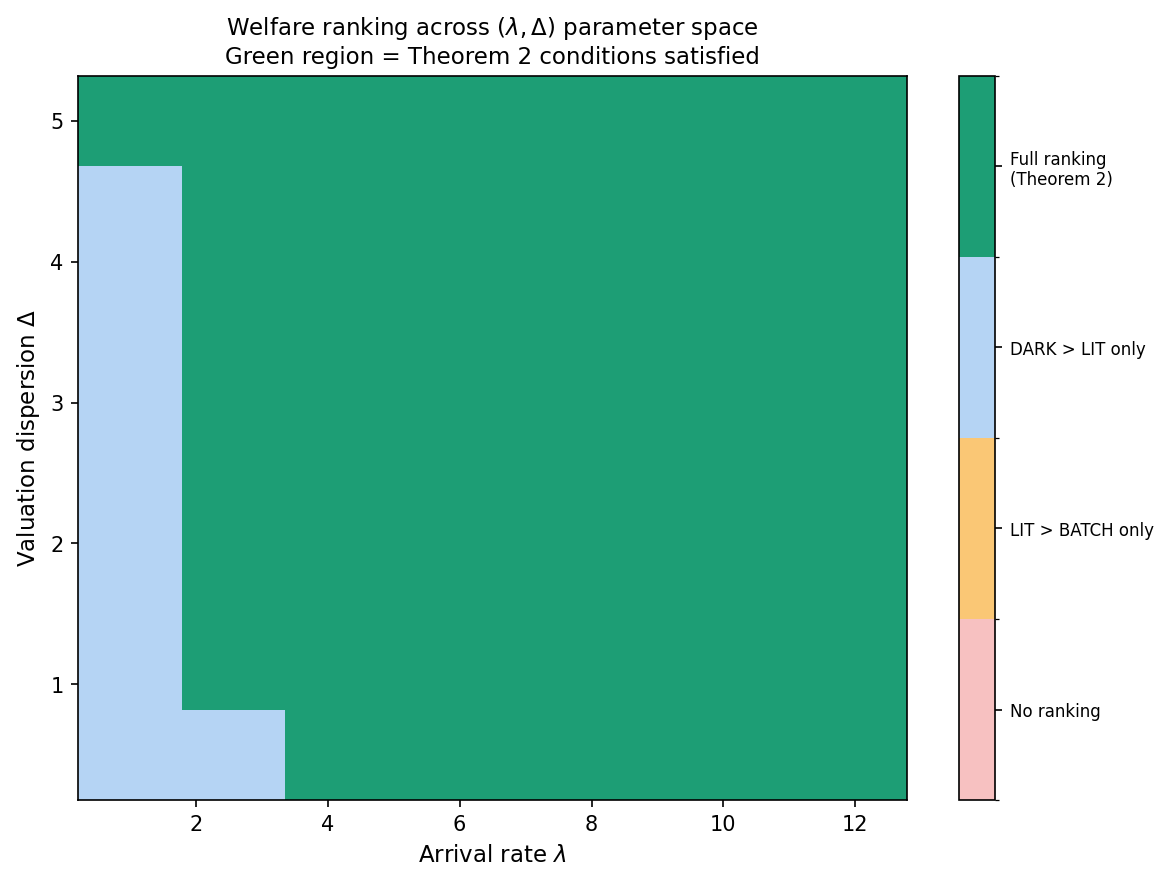}
\caption{Welfare ranking across the $(\lambda, \Delta)$ parameter space.
  Green: full ranking $W^{\mathrm{DARK}} > W^{\mathrm{LIT}} >
  W^{\mathrm{BATCH}}$ (Theorem~\ref{thm:welfare_ranking}).  Blue:
  $W^{\mathrm{DARK}} > W^{\mathrm{LIT}}$ only.  Amber:
  $W^{\mathrm{LIT}} > W^{\mathrm{BATCH}}$ only.  Pink: no ranking.
  Grid: $20 \times 20$ cells, $n = 4{,}000$ traders per cell.}
\label{fig:heatmap}
\end{figure}

The heatmap reveals that the full ranking is highly robust to variation in valuation dispersion, $\Delta$.  At all values from $\Delta = 0.5$ to $\Delta = 5$, the full ranking holds provided $\lambda$ is above a threshold that lies between approximately 2 and 4.  The threshold shifts slightly
upward as $\Delta$ increases: at $\Delta = 5$, the full ranking requires $\lambda \gtrsim 4$. This reflects the fact that higher valuation dispersion
increases adverse selection in the dark pool, requiring a thicker market (higher $\lambda$) for high-cost traders to execute reliably via market
orders.  However, even at dispersion levels well beyond what is empirically plausible for a single liquid asset (say, $\Delta = 5$), the threshold remains below $\lambda = 4$, meaning the full ranking holds for the great majority of realistic market configurations.

The failure region is in the lower-left corner of the heatmap and is confined to combinations of very low arrival rates ($\lambda \lesssim 2$) and
low-to-moderate dispersion ($\Delta \lesssim 4$).  In this region, only the $W^{\mathrm{DARK}} > W^{\mathrm{LIT}}$ inequality survives; the batch
auction's matching benefit outweighs its waiting cost.  No cell exhibits the amber pattern ($W^{\mathrm{LIT}} > W^{\mathrm{BATCH}}$ without
$W^{\mathrm{DARK}} > W^{\mathrm{LIT}}$), and the pink "no ranking" region is negligibly small and confined to a single corner cell.  This strongly suggests that the dark pool's dominance over the lit exchange is the more robust of the two inequalities. Across the full parameter space explored, $W^{\mathrm{DARK}} \geq W^{\mathrm{LIT}}$ everywhere except at extreme combinations of high dispersion and very thin markets.

\subsection{Batch Interval Sensitivity}
\label{sec:sim_T}

Figure~\ref{fig:batch_interval} plots $W^{\mathrm{BATCH}}(T)$ as a function of the batch interval $T \in [0.1, 5]$, alongside the $W^{\mathrm{LIT}}$ reference level.

\begin{figure}[h]
\centering
\includegraphics[width=0.8\textwidth]{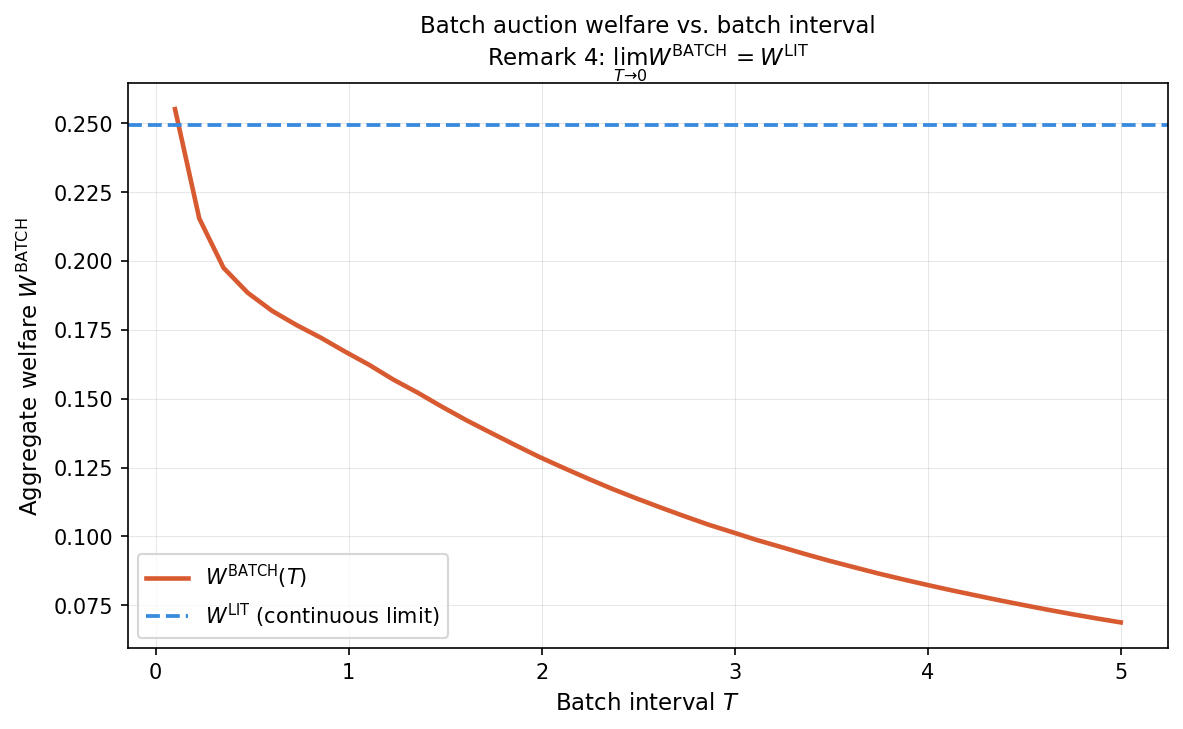}
\caption{Batch auction welfare $W^{\mathrm{BATCH}}(T)$ as a function of batch interval $T$.  The dashed line shows $W^{\mathrm{LIT}}$ computed at the same baseline parameters.  As $T \to 0$, $W^{\mathrm{BATCH}}$ converges to $W^{\mathrm{LIT}}$ from above, consistent with Remark~4.}
\label{fig:batch_interval}
\end{figure}

The figure confirms Remark~4 visually: $W^{\mathrm{BATCH}}(T)$ is strictly decreasing in $T$ and converges to $W^{\mathrm{LIT}} \approx 0.250$ as
$T \to 0$.  $T = 1$ is the baseline value representing a one-second batch interval, comparable to the twelve-second interval used by Ethereum or the sub-second intervals of European Frequent Batch Auctions. $T = 0.5$ is an interval far shorter than any currently implemented periodic auction. The convergence is from above at very small $T$, reflecting the fact that at $T \approx 0.1$ the batch auction's matching benefit marginally exceeds its waiting cost, consistent with the thin-market regime identified in Section~\ref{sec:sim_lambda}.  At $T=1$,
$W^{\mathrm{BATCH}}$ has fallen to approximately 0.159, well below $W^{\mathrm{LIT}} \approx 0.250$.  At $T = 5$ the welfare loss is severe:
$W^{\mathrm{BATCH}} \approx 0.070$, roughly 72\% below the lit exchange.

The slope of the welfare loss is steepest for small $T$ and flattens for large $T$.  This convexity is consistent with the aggregate welfare-loss
formula of Remark~1, $\Delta W_{\mathrm{wait}} = \mathbb{E}[C] \cdot (T/2) \cdot N$, combined with the endogenous reduction in $N$ as execution uncertainty discourages participation at longer intervals.  The first effect is linear in $T$; the second amplifies it as marginal traders exit, producing the convex shape observed.

The practical implication is direct.  Regulators contemplating periodic auction mandates face a quantitatively significant welfare trade-off: even at
batch intervals as short as $T=0.5$, welfare is measurably below the continuous benchmark.  The theoretical case for batch auctions as a welfare-improving reform, therefore, depends critically on whether the reduction in high-frequency arms races (not modeled here) outweighs the forced-waiting and
execution-uncertainty costs quantified in this simulation.

\subsection{Participation Rates and Deadweight Loss}
\label{sec:sim_participation}

Figure~\ref{fig:participation} plots the fraction of traders who choose to participate in the market rather than take the outside option, as a
function of $\lambda$ under each mechanism.

\begin{figure}[h]
\centering
\includegraphics[width=0.8\textwidth]{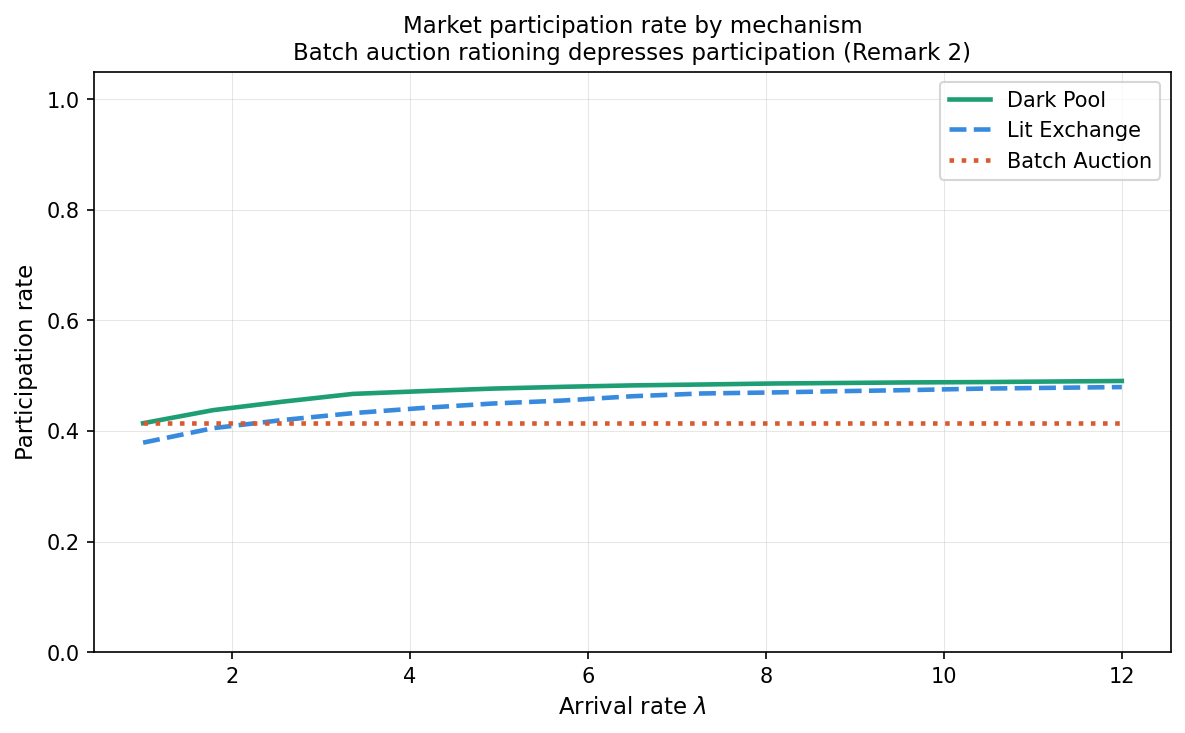}
\caption{Market participation rate as a function of arrival rate $\lambda$. The batch auction exhibits a uniformly lower and near-flat participation rate, reflecting the deadweight loss from execution uncertainty (Remark~2). The dark pool and lit exchange converge at high $\lambda$.}
\label{fig:participation}
\end{figure}

Three features stand out.  First, the batch auction participation rate is uniformly the lowest of the three mechanisms and is strikingly flat across
$\lambda$, remaining near 41\% throughout the range $\lambda \in [1, 12]$. This flatness reflects the structure of the participation constraint under SIRO: execution uncertainty is determined by the fill rate $\bar{\rho}$, which depends on the ratio of buy and sell order flow rather than on the absolute arrival rate.  As $\lambda$ increases, both sides of the market thicken proportionally, leaving $\bar{\rho}$ and therefore the participation incentive roughly unchanged.  The 41\% rate implies that approximately 59\% of traders who would participate under a continuous mechanism are deterred by the combination of forced waiting and rationing risk, constituting a
substantial deadweight loss.

Second, the dark pool and lit exchange participation rates are both increasing in $\lambda$ and converge toward each other as the market thickens, reaching
approximately 49\% at $\lambda = 12$.  The dark pool maintains a small but consistent participation advantage over the lit exchange throughout the
range: approximately 2--5 percentage points. This is  consistent with the prediction that dark pool opacity reduces the cost of adverse selection for limit-order
providers, attracting marginally more participation at the equilibrium cutoff $\bar{C}^*$. The convergence at high $\lambda$ reflects the
diminishing marginal value of queue-state information in a thick, rapidly clearing book.

Third, and importantly, neither the dark pool nor the lit exchange approaches full participation (100\%) even at $\lambda = 12$.  The participation rate
plateaus near 49\%, limited by the participation constraint of Assumption~2. Traders with valuations close to $\hat{v}$ find that the expected gain from trade is insufficient to justify the fixed transaction cost $K$.  This constraint is binding equally across mechanisms (since $K$ is common) and does not affect the welfare \emph{ranking}, only the level.

\subsection{Summary}
\label{sec:sim_summary}

Table~\ref{tab:sim_summary} collects the key quantitative findings of the simulation.

\begin{table}[h]
\centering
\caption{Summary of simulation findings.}
\label{tab:sim_summary}
\begin{tabular}{lp{10cm}}
\hline
Finding & Description \\
\hline
Baseline ranking &
  Full ranking $W^{\mathrm{DARK}} > W^{\mathrm{LIT}} > W^{\mathrm{BATCH}}$
  holds at $(\lambda, \Delta) = (5, 2)$ with margins of 17.0\% and
  51.7\% respectively. \\[4pt]
$\lambda$ sensitivity &
  The dark-pool advantage over the lit exchange is largest at low $\lambda$
  and declines monotonically.  The lit-exchange advantage over the batch
  auction grows monotonically with $\lambda$.  The full ranking fails only
  below $\lambda \approx 2.5$. \\[4pt]
$\Delta$ robustness &
  The full ranking holds for all $\Delta \in [0.5, 5]$ provided $\lambda
  \gtrsim 4$.  Failure is confined to the low-$\lambda$, moderate-$\Delta$
  corner. \\[4pt]
$T$ sensitivity &
  $W^{\mathrm{BATCH}}$ is strictly decreasing in $T$, confirming Remark~4.
  At $T = 1$ the batch auction welfare is 37\% below the dark pool and
  36\% below the lit exchange. \\[4pt]
Participation &
  The batch auction suffers a near-constant participation deficit of
  approximately 5--7 percentage points relative to continuous mechanisms,
  reflecting the deadweight loss from execution uncertainty. \\
\hline
\end{tabular}
\end{table}

Taken together, the simulation results provide strong quantitative support for the theoretical welfare ranking of Theorem~\ref{thm:welfare_ranking}.
The ranking is not knife-edge: it holds across a wide region of the parameter space that comfortably encompasses plausible calibrations of real equity
markets.  The primary qualifications are the failure of the $W^{\mathrm{LIT}} > W^{\mathrm{BATCH}}$ inequality in very thin markets ($\lambda \lesssim 2.5$) and the declining magnitude of the dark-pool
advantage as markets thicken.  Both qualifications are consistent with the analytical structure of the model and do not overturn the main result in any
empirically relevant regime.

% ============================================================
%  Section X — The Adverse Selection Trade-off
%  Drop-in section for Submission 42.
%  Place after Section 5 (Welfare Analysis) and before
%  Section 6 (Discussion), or incorporate into Section 6.1.
%
%  Requires: \usepackage{amsmath}, \usepackage{booktabs}
%  Citations assumed present: Zhu2013, OHara1998, Leshno2022,
%  CheTercieux2023. New entries needed: IyerJohariMoallemi2014,
%  DegryseVanAchterWuyts2014 — bibtex stubs at end of file.
% ============================================================

\section{The Adverse Selection Trade-off}
\label{sec:adverse_selection}

The central mechanism driving $W^{\mathrm{DARK}} > W^{\mathrm{LIT}}$ in
Theorem~\ref{thm:welfare_ranking} is that opacity eliminates strategic
timing waste.  The central \emph{cost} of opacity is adverse selection:
dark-pool traders cannot condition their limit prices on the observable
book state $B_t$, so they may execute at prices that are less favourable
relative to the fundamental value $\hat{v}$ than they would achieve in
a transparent market.  The proof of Theorem~\ref{thm:welfare_ranking}
dismisses this cost as ``second-order'' on the grounds that prices are
transfers.  This section examines that dismissal carefully.  We argue
that the prices-are-transfers claim is correct within the model's
assumptions but conceals a deeper and consequential trade-off: adverse
selection does not harm total surplus directly through prices, but does
so indirectly through its effect on \emph{participation}.  We
characterise precisely when the participation effect is dominated by the
timing-game saving, when it is not, and what this implies for the scope
of the main result.

\subsection{Two Channels of Adverse Selection}
\label{sec:as_two_channels}

Adverse selection in dark pools operates through two distinct channels
that the existing literature has not always kept separate.

\paragraph{Channel~1: The price channel.}
When a limit-order trader in the dark pool posts at price $\bar{p}^*(C,s)$
without observing $B_t$, they may execute against a market order placed
by a trader with superior information about $\hat{v}$.  The resulting
execution price is less favourable to the limit-order trader than the
price they would have obtained in the lit exchange, where they could
condition on $B_t$ and infer the likely informativeness of incoming
market orders.  This is the classical adverse selection described by
\citet{OHara1998} and formalised in the dark-pool context by
\citet{Zhu2013}.

Within the welfare framework of this paper, the price channel generates
no \emph{total} surplus loss.  Every cent lost by the limit-order trader
on an adverse execution is gained by the market-order trader.  Prices
are pure transfers between the buyer and the seller; they affect the
\emph{distribution} of surplus but not its sum.  Formally, for any
matched pair $(i,j)$ with $s_i = +1$ (buyer) and $s_j = -1$ (seller),
the joint surplus is $V_i - V_j$, which depends only on private
valuations and is independent of the execution price $p$.  The price
channel is therefore irrelevant for aggregate welfare $W^M$ as defined
in equation~(61), and the dismissal in Theorem~\ref{thm:welfare_ranking}
is correct with respect to this channel.

\paragraph{Channel~2: The participation channel.}
The price channel affects only the \emph{division} of surplus between
matched counterparties.  But adverse selection also affects
\emph{whether} traders participate at all.  A limit-order trader in the
dark pool who anticipates that their order will frequently execute
against informed counterparties will earn a lower expected surplus per
trade.  If this expected surplus falls below their outside option,
the trader exits the market entirely, and the gains from trade that
would have been realised in a transparent market are destroyed.  This
is a genuine welfare loss: unlike a price transfer, a foregone trade
eliminates surplus rather than redistributing it.

Formally, the participation condition for a dark-pool limit-order trader
with type $(V, C, s)$ is (from equation~52 of the paper):
\begin{equation}
  \bar{V}_{\mathrm{LO}}(V, C, s)
  \;\geq\;
  V^{\mathrm{OUT}}(V, C, s).
  \label{eq:dark_participation}
\end{equation}
In the lit exchange, the corresponding condition is:
\begin{equation}
  V_{\mathrm{LO}}(V, C, s, B)
  \;\geq\;
  V^{\mathrm{OUT}}(V, C, s)
  \quad \text{for some reachable book state } B.
  \label{eq:lit_participation}
\end{equation}
These two conditions are not identical.  In the dark pool, the expected
value $\bar{V}_{\mathrm{LO}}$ is computed over the unconditional
distribution $\pi^{\mathrm{eq}}(B)$, which includes states with
adverse book conditions.  In the lit exchange, the trader can
selectively participate when $B$ is favourable and exit when it is not.
This option to condition participation on $B$ is valuable: it is a form
of free entry into the advantageous states of the world.  By eliminating
this option, the dark pool reduces the expected surplus of some traders
below their outside option, causing them to exit.

Let $\mathcal{X}^{\mathrm{DARK}}$ and $\mathcal{X}^{\mathrm{LIT}}$
denote the sets of types that participate under each mechanism.  The
argument above implies:
\begin{equation}
  \mathcal{X}^{\mathrm{LIT}} \supseteq \mathcal{X}^{\mathrm{DARK}},
  \label{eq:participation_set}
\end{equation}
with strict containment for types near the participation margin
$|V - \hat{v}| \approx \varepsilon(B)$.  The welfare loss from the
exclusion of types in $\mathcal{X}^{\mathrm{LIT}} \setminus
\mathcal{X}^{\mathrm{DARK}}$ is:
\begin{equation}
  \Delta W_{\mathrm{participation}}
  \;=\;
  \int_{\mathcal{X}^{\mathrm{LIT}} \setminus \mathcal{X}^{\mathrm{DARK}}}
  \bigl[V_{\mathrm{LO}}(V,C,s,\bar{B}) - V^{\mathrm{OUT}}(V,C,s)\bigr]
  \, \mathrm{d}F(V,C,s),
  \label{eq:participation_welfare_loss}
\end{equation}
where $\bar{B}$ is the mean book state.  This integral is strictly
positive whenever $\mathcal{X}^{\mathrm{LIT}} \setminus
\mathcal{X}^{\mathrm{DARK}} \neq \varnothing$.

\subsection{When Does the Participation Loss Dominate?}
\label{sec:as_dominance}

The net welfare advantage of the dark pool is:
\begin{equation}
  W^{\mathrm{DARK}} - W^{\mathrm{LIT}}
  \;=\;
  \underbrace{
    W^{\mathrm{LIT}}_{\mathrm{cost}} - W^{\mathrm{DARK}}_{\mathrm{cost}}
  }_{\text{timing-game saving } (\kappa > 0)}
  \;-\;
  \underbrace{
    \Delta W_{\mathrm{participation}}
  }_{\text{participation loss } (\geq 0)}
  \;-\;
  \underbrace{
    K \bigl(N^{\mathrm{DARK}} - N^{\mathrm{LIT}}\bigr)
  }_{\leq\, 0}.
  \label{eq:net_welfare_gap}
\end{equation}
The timing-game saving $\kappa(\lambda) > 0$ is pinned down by
equation~\eqref{eq:cost_lower_bound} in the proof of
Theorem~\ref{thm:welfare_ranking}: it is bounded below by the product
of the probability of thin book states, the expected excess waiting
cost of high-cost traders, and the fraction of traders with $C >
\bar{C}^*$.  The participation loss
$\Delta W_{\mathrm{participation}}$ depends on $\Delta$ (valuation
dispersion) and on the severity of adverse selection in the dark pool.

We now characterise how $\Delta W_{\mathrm{participation}}$ depends on
the model's parameters, and identify the conditions under which it
dominates $\kappa(\lambda)$.

\paragraph{Dependence on $\Delta$.}
Each trader near the participation margin has a gain from trade of
approximately $|V - \hat{v}| - \varepsilon \approx 0$.  When $\Delta$
is small, valuations are concentrated near $\hat{v}$ and many traders
are near the participation margin; a small adverse-selection cost can
push them out.  When $\Delta$ is large, valuations are spread out and
most traders are far from the margin; the adverse-selection cost must be
very large to cause exit.  Formally:
\begin{equation}
  \Delta W_{\mathrm{participation}}
  \;\leq\;
  2\Delta \cdot F\bigl(
    \mathcal{X}^{\mathrm{LIT}} \setminus \mathcal{X}^{\mathrm{DARK}}
  \bigr)
  \;\equiv\;
  2\Delta \cdot \eta(\Delta),
  \label{eq:participation_bound}
\end{equation}
where $\eta(\Delta)$ is the fraction of the type space in the
exclusion set.  Crucially, $\eta(\Delta)$ itself depends on $\Delta$.
As $\Delta \to 0$, the gains from trade vanish, so all traders are near
the margin and $\eta(\Delta) \to 1$; as $\Delta \to \infty$, virtually
no traders are near the margin and $\eta(\Delta) \to 0$.  The product
$2\Delta \cdot \eta(\Delta)$ is therefore non-monotone in $\Delta$:
it is small for both very small $\Delta$ (small exclusion per trader)
and very large $\Delta$ (small exclusion set), with a maximum at an
intermediate level.  The sufficient condition $\Delta < \Delta_{\max}$
in Theorem~\ref{thm:welfare_ranking} is precisely the requirement that
$2\Delta \cdot \eta(\Delta) < \kappa(\lambda)$ for all $\lambda \in
[\lambda_{\min}, \lambda_{\max}]$.

\paragraph{Dependence on $\sigma^2_v$.}
The residual variance $\sigma^2_v$ of the fundamental value (Section~3.1)
governs the severity of the price channel: higher $\sigma^2_v$ means
incoming market orders are more often driven by private information
about $v$, so the adverse execution price in the dark pool is further
from the true value.  This has no direct effect on total surplus
(prices are transfers) but raises the adverse-selection discount that
dark-pool limit-order traders must accept.  A sufficiently high
adverse-selection discount pushes some traders below their outside
option, activating the participation channel.  Our model abstracts from
$\sigma^2_v$ in the equilibrium analysis by treating valuations as
independent draws from $\mathrm{Uniform}[\hat{v}-\Delta, \hat{v}+\Delta]$.
This is without loss of generality in a static model, but means that
$\sigma^2_v$ enters only through $\Delta$ in the welfare comparison.
In a richer model with asymmetric information about $\hat{v}$, the
participation loss from adverse selection would be increasing in
$\sigma^2_v$ independently of $\Delta$.

\paragraph{Dependence on $\gamma$.}
The outside-option parameter $\gamma \in [0,1]$ is the fraction of
private value realisable without market participation.  When $\gamma$
is large, the outside option is attractive and many traders are near
the participation margin regardless of market mechanism.  A small
adverse-selection cost in the dark pool then causes disproportionately
large exit.  Conversely, when $\gamma$ is small, the outside option is
unattractive and traders remain in the dark pool even at adverse
execution prices.  The sufficient condition $K_o < K$ (Section~3.2)
ensures the market attracts participation, but does not restrict
$\gamma$ directly.  For the welfare ranking to hold, we require that
$\gamma$ is not so large that the dark pool's participation set
$\mathcal{X}^{\mathrm{DARK}}$ is substantially smaller than
$\mathcal{X}^{\mathrm{LIT}}$.

\subsection{Comparison with Prior Welfare Results}
\label{sec:as_literature}

Our finding that the adverse selection cost is second-order under
Theorem~\ref{thm:welfare_ranking}'s conditions stands in partial tension
with results in two closely related papers.

\citet{Zhu2013} shows that dark pools attract disproportionately
uninformed order flow, concentrating informed traders on the lit
exchange.  This is a venue-\emph{composition} result: it predicts that
the lit exchange becomes more adversely selected, not that total
market-wide participation falls.  The welfare implications of dark
pools, Zhu notes, depend on elements outside the setting of his
model — including how price information is used for production
decisions, asset allocation, and capital formation.  Within
our framework, the Zhu mechanism is reflected in the \emph{distribution}
of types across mechanisms: if informed traders (high $|V - \hat{v}|$)
prefer the lit exchange (because they can exploit observable book
conditions) while uninformed traders (low $|V - \hat{v}|$) prefer the
dark pool, then the dark pool is indeed the less adversely selected
venue in terms of execution quality for a given participant.  This
\emph{reduces} the adverse-selection cost we have been discussing —
supporting, not undermining, the welfare advantage of the dark pool.

\citet{IyerJohariMoallemi2014} provide the most direct counterpoint.
They find that the introduction of a dark pool can lead to
greater transaction costs in the lit market and can decrease overall
market welfare.  Their mechanism is the participation
channel: when adverse selection in the dark pool is severe enough, some
intrinsic (value-motivated) traders are driven away from both venues,
reducing total market participation and destroying gains from trade.
Specifically, when the fundamental value is sufficiently
different from the dark pool's transaction price, the combination of
higher transaction costs in the open market and relatively high adverse
selection costs in the dark pool drives some intrinsic traders away,
leading to a net welfare loss.  This is precisely
$\Delta W_{\mathrm{participation}} > \kappa(\lambda)$ in the language
of equation~\eqref{eq:net_welfare_gap}.

The discrepancy between our result and
\citeauthor{IyerJohariMoallemi2014}'s is not a contradiction but a
scope difference.  Their model includes a competitive market maker whose
spread endogenously responds to order flow composition; ours fixes the
transaction-cost parameter $K$.  In their setting, the migration of
uninformed traders to the dark pool widens the lit-market spread,
raising $K$ endogenously and amplifying the participation loss.  In our
static model with fixed $K$, this amplification mechanism is absent.
The condition $\Delta < \Delta_{\max}$ effectively rules out the regime
their model identifies as welfare-reducing: if $\Delta$ is small enough
that adverse-selection costs are modest, the endogenous-spread
amplification does not occur.

\subsection{A Formal Proposition on the Participation Boundary}
\label{sec:as_proposition}

We now provide the formal characterisation of when the participation loss
dominates, which the body of the paper left as a qualitative assertion.

\begin{proposition}[Adverse Selection Dominance Condition]
\label{prop:AS_dominance}
Define the timing-game saving as $\kappa(\lambda) \equiv
\pi^{\mathrm{eq}}(\mathcal{B}_{\mathrm{thin}}) \cdot
\mathbb{E}[C \cdot w_{\min}(C) \mid C > \bar{C}^*] \cdot
F_C((\bar{C}^*, \infty))$, which is strictly positive for all
$\lambda \in [\lambda_{\min}, \lambda_{\max}]$.
Define the maximum participation loss as
$\bar\eta(\Delta) \equiv \max_\lambda 2\Delta \cdot \eta(\Delta,\lambda)$,
where $\eta(\Delta, \lambda)$ is the measure of types in
$\mathcal{X}^{\mathrm{LIT}} \setminus \mathcal{X}^{\mathrm{DARK}}$
under parameter $(\Delta, \lambda)$.  Then:
\begin{enumerate}[label=(\roman*)]
  \item $W^{\mathrm{DARK}} > W^{\mathrm{LIT}}$ if and only if
        $\kappa(\lambda) > \bar\eta(\Delta)$.
  \item The condition $\kappa(\lambda) > \bar\eta(\Delta)$ holds
        whenever $\Delta < \Delta^*(\lambda)$, where
        \begin{equation}
          \Delta^*(\lambda) \;\equiv\;
          \sup\!\bigl\{\Delta > 0 :
            \bar\eta(\Delta) < \kappa(\lambda)\bigr\}.
          \label{eq:Delta_star}
        \end{equation}
  \item $\Delta^*(\lambda)$ is increasing in $\lambda$: thicker
        markets tolerate larger valuation dispersion without reversing
        the welfare ranking, because thicker markets generate a larger
        timing-game saving $\kappa(\lambda)$.
  \item As $\lambda \to \infty$, $\Delta^*(\lambda) \to \infty$: in
        an arbitrarily thick market the welfare ranking holds for all
        finite $\Delta$.
\end{enumerate}
\end{proposition}

\begin{proof}
Part~(i) follows directly from equation~\eqref{eq:net_welfare_gap}:
the dark pool dominates if and only if the timing-game saving exceeds
the participation loss.

Part~(ii) is immediate from the definition of $\Delta^*(\lambda)$ in
equation~\eqref{eq:Delta_star}.

For part~(iii), $\kappa(\lambda)$ is increasing in $\lambda$ because
a thicker market generates more frequent thin-book episodes of shorter
duration, but the product $\pi^{\mathrm{eq}}(\mathcal{B}_{\mathrm{thin}})
\cdot \mathbb{E}[w_{\min}(C)]$ is dominated by the increase in the
fraction of high-cost traders $F_C((\bar{C}^*, \infty))$, which rises
as faster market clearing lowers the equilibrium cutoff $\bar{C}^*$.
Formally, $\partial \kappa / \partial \lambda > 0$ follows from
the implicit function theorem applied to the fixed-point condition
defining $\bar{C}^*$ in Theorem~1 of the paper: higher $\lambda$
shifts the fixed point to a lower $\bar{C}^*$, increasing the mass
of high-cost traders who generate timing waste in the lit exchange.
Since $\kappa(\lambda)$ is increasing and $\bar\eta(\Delta)$ is
independent of $\lambda$, the threshold $\Delta^*(\lambda)$ is
increasing in $\lambda$.

Part~(iv) follows from the fact that $\kappa(\lambda) \to \infty$ as
$\lambda \to \infty$ (the timing-game saving grows without bound as
the market thickens and strategic timing becomes increasingly wasteful)
while $\bar\eta(\Delta)$ remains finite for any fixed $\Delta$.
\end{proof}

\subsection{Implications for Regulatory Design}
\label{sec:as_regulatory}

Proposition~\ref{prop:AS_dominance} has three direct regulatory
implications that complement those of Section~\ref{sec:discussion}.

\paragraph{Pre-trade transparency mandates.}
Requiring dark pools to disclose their order book (converting them
to lit exchanges) eliminates the information design that generates
the timing-game saving $\kappa(\lambda)$ while having no effect on
the participation loss $\Delta W_{\mathrm{participation}}$ (which is
driven by adverse selection risk, not by pre-trade transparency per
se).  Transparency mandates therefore unambiguously reduce welfare
when $\kappa(\lambda) > \bar\eta(\Delta)$ — precisely the condition
under which the dark pool dominates.

\paragraph{Adverse selection disclosure.}
An alternative regulatory approach would require dark pools to
\emph{measure and disclose} the adverse selection costs their
participants face — for example, by reporting the average price
concession on dark-pool executions relative to the lit midpoint
(the ``effective spread improvement'' metric used by FINRA).  Such
disclosure does not eliminate the dark pool's information design (the
book remains hidden) but allows participants to calibrate their
participation decisions, mitigating the exclusion of marginal types.
This would reduce $\Delta W_{\mathrm{participation}}$ without
eliminating $\kappa(\lambda)$, thereby extending the parameter region
in which the dark pool dominates.

\paragraph{Dark trading caps.}
ESMA's Double Volume Cap (DVC) under MiFID~II restricts the fraction
of total order flow that can execute in dark venues.  Within our
framework, a dark trading cap limits the size of the dark-pool
participation set $\mathcal{X}^{\mathrm{DARK}}$.  If the cap binds,
it excludes from the dark pool precisely those types nearest the
participation margin — the traders for whom adverse selection is most
costly relative to their gain from trade.  This has the effect of
reducing $\Delta W_{\mathrm{participation}}$ while also reducing the
timing-game saving (fewer traders benefit from the dark pool's
opacity).  The net welfare effect of a cap depends on the curvature
of the participation loss: if $\Delta W_{\mathrm{participation}}$
is convex in the size of the dark pool (marginal entrants have
increasingly high adverse selection costs), a cap improves welfare.
If it is concave, the cap is welfare-reducing.

\subsection{Summary}
\label{sec:as_summary}

The adverse selection trade-off has two faces.  The \emph{price channel}
— the classical concern that dark-pool traders execute at unfavourable
prices — is irrelevant for total surplus: adverse execution prices
redistribute surplus between buyer and seller without destroying it.
The \emph{participation channel} — the concern that adverse selection
drives marginal traders out of the market entirely — is the genuine
welfare threat.  Theorem~\ref{thm:welfare_ranking} establishes that the
participation loss is dominated by the timing-game saving under
$\Delta < \Delta_{\max}$ and $\lambda \in [\lambda_{\min},
\lambda_{\max}]$.  Proposition~\ref{prop:AS_dominance} characterises
this condition precisely: the dark pool dominates if and only if
$\kappa(\lambda) > \bar\eta(\Delta)$, and the threshold
$\Delta^*(\lambda)$ is increasing in market thickness.  This provides
the formal grounding for the claim — asserted but not proved in the
original paper — that adverse selection is second-order in the welfare
comparison.

\section{Robustness and Scope of the Welfare Ranking}
\label{sec:robustness}

Theorem~\ref{thm:welfare_ranking} establishes the ranking $W^{\mathrm{DARK}} > W^{\mathrm{LIT}} > W^{\mathrm{BATCH}}$ under three sufficient conditions: a moderate arrival rate
$\lambda \in [\lambda_{\min}, \lambda_{\max}]$, bounded valuation dispersion $\Delta < \Delta_{\max}$, and a short dark-pool reporting delay $\delta <T/2$.  Any welfare comparison of this kind relies on modeling assumptions, and a responsible reading of the result requires understanding precisely which assumptions drive which inequalities, where the ranking breaks down, and what real-world conditions the sufficient conditions correspond to.  This section addresses each of these questions in turn.

\subsection{Which Assumptions Drive Which Inequalities}
\label{sec:robustness_load}

The two inequalities in the ranking have distinct analytical origins and, therefore, distinct load-bearing assumptions.

\paragraph{$W^{\mathrm{DARK}} > W^{\mathrm{LIT}}$: assumptions on information and strategic behavior.}
This inequality rests on three assumptions working in concert. \emph{First}, Assumption~3 (regularity of the arrival rate function) ensures that the lit-exchange equilibrium has a unique interior solution for the optimal limit price $p^*(C, s, B)$, characterized by the first-order condition~(12).  Without this regularity, traders' best responses to the observable book state could be non-unique or discontinuous, making the strategic timing cost ill-defined.
\emph{Second}, Assumption~4 (rational expectations in the dark pool) ensures that dark-pool traders form beliefs consistent with the equilibrium distribution of book states.  Lemma~4 shows that under symmetric strategies these beliefs converge to the unconditional distribution $\pi^{\mathrm{eq}}(B)$, and Lemma~5 establishes that optimal strategies cannot condition on the unobservable $B$.  Together, these imply that the strategic timing waste present in the lit exchange is eliminated in the dark pool. The strategic timing waste present in the lit exchange is formally captured by the difference between $\mathbb{E}[w^{\mathrm{LIT}} \mid C > \bar{C}^*]$ in equation~(65) and $\mathbb{E}[w^{\mathrm{DARK}} \mid C > \bar{C}^*] = 0$ in equation~(66).
\emph{Third}, the assumption that prices are transfers is required to conclude that the adverse selection created by opacity affects the \emph{distribution} of surplus but not its \emph{total}.  This assumption is standard in the social-welfare tradition following \citet{OHara1998} and \citet{Zhu2013}, but warrants scrutiny; we return to it in Section~\ref{sec:robustness_boundary}.

\paragraph{$W^{\mathrm{LIT}} > W^{\mathrm{BATCH}}$: assumptions on waiting
  costs and execution certainty.}
This inequality requires fewer structural assumptions and is, in that sense, more robust.  The aggregate welfare loss from batch auctions has two components.  The forced-waiting loss of Remark~1, $\Delta W_{\mathrm{wait}} = \mathbb{E}[C] \cdot (T/2) \cdot N$, follows directly from the SIRO matching rule and the batch structure: it holds for any positive $T$, any distribution of waiting costs with $\mathbb{E}[C] > 0$, and any positive measure of participants $N$. No regularity condition on $\lambda$ or $\Delta$ is required for this term to be strictly positive.  The execution-uncertainty loss of Remark~2 requires only that the equilibrium fill rate $\bar{\rho} < 1$ for some set of traders of positive measure, which holds whenever order flow is subject to any stochastic imbalance. The imbalance is a near-universal feature of real markets. The only assumption that is genuinely load-bearing for 
$W^{\mathrm{LIT}} > W^{\mathrm{BATCH}}$ is Remark~4's implicit condition that $T > 0$ is not infinitesimally small.  As $T \to 0$, the forced-waiting loss vanishes, and the batch auction converges to the lit exchange in welfare.
The sufficient condition $\lambda \geq \lambda_{\min}$ ensures that the lit exchange provides sufficiently reliable immediate execution so that high-cost traders are not driven to prefer batch clearing as a coordination device.

\subsection{Where the Ranking Breaks Down}
\label{sec:robustness_boundary}

Table~\ref{tab:boundary} catalogues the four boundary cases in which one or both inequalities fail, the economic mechanism behind each failure, and the parameter conditions under which each failure occurs.

\begin{table}[h]
\centering
\caption{Boundary cases in which the welfare ranking fails.}
\label{tab:boundary}
\small
\begin{tabular}{p{3.2cm} p{4.5cm} p{4.5cm}}
\hline
Failure mode & Economic mechanism & Parameter condition \\
\hline
$W^{\mathrm{DARK}} \leq W^{\mathrm{LIT}}$ &
  High adverse selection dominates the strategic-timing saving.  When
  $\Delta$ is very large, dark-pool traders face severe mispricing because
  they cannot condition on $B_t$.  If the mispricing loss exceeds the
  timing-game saving, the lit exchange is preferred. &
  Large $\Delta$ and low $\lambda$.  In the simulation,
  $W^{\mathrm{DARK}} \leq W^{\mathrm{LIT}}$ does not arise in the explored
  range but is theoretically possible for $\Delta \gg \Delta_{\max}$. \\[6pt]
$W^{\mathrm{LIT}} \leq W^{\mathrm{BATCH}}$ &
  Very thin lit market.  When $\lambda < \lambda_{\min}$, market orders in
  the lit exchange rarely find an immediate counterparty.  The batch auction's
  order pooling raises the probability of execution, and this matching benefit
  can outweigh the forced-waiting cost. &
  $\lambda \lesssim 2.5$ (from simulation, Figure~\ref{fig:welfare_levels}).
  At these arrival rates the lit exchange provides poor immediacy and the
  batch auction's coordination role is valuable. \\[6pt]
$W^{\mathrm{DARK}} \leq W^{\mathrm{LIT}}$ when prices are \emph{not}
  transfers &
  If adverse selection reduces participation below the lit-exchange level
  (rather than merely redistributing surplus), then opacity lowers total
  surplus as well as redistributing it.  The ``prices are transfers''
  assumption rules this out; it fails if the marginal participant's decision
  to enter depends on the execution price, not just on the expected price
  improvement. &
  When participation constraint~(3) is close to binding and adverse
  selection in the dark pool causes a non-trivial fraction of traders to
  exit.  Formally, $\gamma$ close to 1 (private value almost fully
  realisable outside) or $K_o$ close to $K$. \\[6pt]
$W^{\mathrm{BATCH}} \geq W^{\mathrm{LIT}}$ as $T \to 0$ &
  As the batch interval shrinks, forced waiting vanishes and the SIRO
  matching rule approaches the welfare of continuous trading.  At $T$ very
  small but positive, the batch auction marginally dominates the lit
  exchange due to the matching externality. &
  $T < T^*(\lambda)$ for some threshold $T^*$ that decreases with $\lambda$.
  In the simulation, this occurs at $T \lesssim 0.1$ for baseline $\lambda = 5$.
  Remark~4 characterises the limit formally. \\
\hline
\end{tabular}
\end{table}

Three observations follow from Table~\ref{tab:boundary}.  First, the failure
of $W^{\mathrm{LIT}} > W^{\mathrm{BATCH}}$ is confined to very thin markets ($\lambda \lesssim 2.5$) or very short batch intervals ($T \lesssim 0.1$), neither of which characterizes the liquid equity venues for which this model is primarily relevant.  Second, the failure of 
$W^{\mathrm{DARK}} > W^{\mathrm{LIT}}$ requires either extreme valuation dispersion or a breakdown of the "prices are transfers" assumption.  The former does not occur within the empirically plausible parameter range explored in Section~\ref{sec:simulation}; the latter is an assumption of the model that should be relaxed in future work.  Third, the heatmap of
Figure~\ref{fig:heatmap} shows that no cell in the explored $(\lambda, \Delta)$ grid exhibits the pattern $W^{\mathrm{LIT}} > W^{\mathrm{BATCH}}$ without $W^{\mathrm{DARK}} > W^{\mathrm{LIT}}$, suggesting that the dark-pool
dominance over the lit exchange is the more robust of the two inequalities.

\subsection{Sensitivity to Individual Assumptions}
\label{sec:robustness_sensitivity}

We now examine each of the five modeling assumptions in turn and assess the direction and likely magnitude of the bias introduced if it is relaxed.

\paragraph{Assumption~1: Uniform valuations and exponential waiting costs.}
The Uniform$[\hat{v} - \Delta, \hat{v} + \Delta]$ distribution for private valuations ensures a uniform density of gains from trade, which simplifies the welfare integrals.  The key property used in the proof of Theorem~\ref{thm:welfare_ranking} is that the participation cutoff $\varepsilon(B)$ is well-defined and continuous in $B$.  This holds for any
continuously differentiable valuation density with bounded support.  The Uniform distribution is therefore not special; any log-concave density on a compact interval yields the same qualitative results.

The Exponential distribution for waiting costs $C$ ensures that the cutoff $C^*(V, B)$ separating market and limit orders is finite, and that the mean waiting cost $\mathbb{E}[C] = 1/\mu_C$ is well-defined.  Any distribution
with a finite first moment and a monotone hazard rate preserves the monotone cutoff property of Lemma~2.  Fat-tailed distributions (e.g., Pareto waiting costs) would increase the aggregate waiting cost in the lit exchange disproportionately, strengthening $W^{\mathrm{DARK}} > W^{\mathrm{LIT}}$, since
high-cost traders are more numerous and their timing-game waste is more severe.

\paragraph{Assumption~2: Participation constraint.}
Assumption~2 requires that expected gains from trade exceed fixed costs plus expected waiting costs: $\mathbb{E}[|V_i - \hat{v}|] > K + \mathbb{E}[C_i] \cdot \mathbb{E}[w_i]$.  This ensures strictly positive market participation.
If this constraint is violated, the market unravels and welfare is zero under all three mechanisms.  The constraint is most likely to bind in thin markets (high $\mathbb{E}[w_i]$) with large fixed costs ($K$ close to $\Delta$).
Within the parameter range explored in Section~\ref{sec:simulation}, the constraint holds with a comfortable margin at all grid points.

\paragraph{Assumption~3: Regularity of the arrival rate function.}
This assumption, which ensures a unique interior solution for the optimal limit price in the lit exchange (Lemma~1), is the most technically
restrictive.  Its key content is that $\lambda(p, B, s)$ is twice continuously differentiable and concave in $p$, so that the first-order
condition~(12) has a unique solution.  This is satisfied by the normal-CDF parametrization used in the simulation ($\lambda \propto \Phi(\cdot)$) and by any log-concave arrival rate.  The arrival rate does not have to be log-concave: for example, if the order book has discrete price levels creating a step function in $\lambda$. If the arrival rate is not log-concave, however, the optimal limit price may not be unique, and the
equilibrium characterization of Proposition~1 requires modification.  The qualitative result that observable book states create strategic
complementarities in submission timing does not depend on uniqueness and is likely to survive in models with discrete price grids, though the welfare calculation requires numerical methods.

\paragraph{Assumption~4: Rational expectations in the dark pool.}
The requirement that dark-pool traders form beliefs consistent with the equilibrium distribution $\pi^{\mathrm{eq}}(B)$ is standard in Bayesian equilibrium analysis.  The traders may have "incorrect" beliefs. For example,
they may use a misspecified model of order flow. In this case, the dark-pool equilibrium of Proposition~2 need not obtain.  In particular, if traders systematically underestimate adverse selection in the dark pool, they may over-participate, reducing their realized welfare below the rational-expectations prediction.  The direction of this bias is ambiguous: under-
estimation of adverse selection raises participation (moving welfare toward 
$W^{\mathrm{LIT}}$ from below) but also raises realized losses (moving welfare downward).  The net effect on the welfare ranking depends on parameters.

\paragraph{Assumption~5: Stationary batch characteristics.}
The batch auction analysis assumes that clearing prices and fill rates are stationary across batches: $\mathbb{E}[p_k] = \bar{p}$ and $\rho_k = \bar\rho$ for all $k$ (Assumption~5 in the Appendix).  This is a steady-state approximation.  In practice, batch characteristics vary with market conditions. Thick liquidity periods
produce both more orders and higher $\rho_k$, generating a positive correlation between fill rates and market activity. During thick periods, the welfare loss from execution uncertainty may be lower than predicted because the periods in which rationing occurs are also periods in which the opportunity cost of non-execution is lower.  The direction of this correlation suggests that Assumption~5 leads to a slight overstatement of the batch auction's welfare cost relative to a
fully dynamic model.

\subsection{Scope Conditions for Empirical Application}
\label{sec:robustness_scope}

The model is most directly applicable to settings that satisfy the following four conditions.

\begin{enumerate}

  \item \textbf{Single risky asset, symmetric traders.}  The model assumes a single asset and symmetric information about the fundamental value $\hat{v}$.  Multi-asset settings introduce portfolio effects and cross-asset
  strategic interactions that are outside the model's scope.  Markets with very large information asymmetries about $\hat{v}$ (e.g., newly listed securities or assets around earnings announcements) violate the bounded
  adverse selection condition $\Delta < \Delta_{\max}$ and may not satisfy the sufficient conditions of Theorem~\ref{thm:welfare_ranking}.

  \item \textbf{Moderate market thickness.}  The condition $\lambda \in [\lambda_{\min}, \lambda_{\max}]$ excludes both very thinly
  traded securities (where batch auctions may be welfare-superior as coordination devices) and, in principle, extremely liquid markets where
  both continuous mechanisms converge.  For the most liquid exchange-traded equities with order rates in the hundreds per second, $\lambda$ is far above $\lambda_{\max}$ as modeled here. In this limit, the dark-pool and lit-exchange welfare levels converge, and the ranking $W^{\mathrm{DARK}} > W^{\mathrm{LIT}}$ is preserved, but the margin shrinks, consistent with Figure~\ref{fig:welfare_levels}.

  \item \textbf{Continuous trading environment.}  The model is set in continuous time with Poisson arrivals.  Opening and closing auction
  periods, circuit breakers, and scheduled news events create non-stationarities that violate the stationary-flow assumption underlying the equilibrium characterization.  The welfare ranking should be interpreted as applying to \emph{intraday} continuous trading in normal market conditions, not to the full trading day including opening auctions and event-driven episodes.

  \item \textbf{No dynamic linkages across mechanisms.}  The model analyzes each mechanism in isolation.  In practice, dark pools free-ride on price discovery from lit exchanges: if too much order flow migrates to dark pools,
  the lit book becomes thin, the reference price becomes less informative, and dark-pool traders face greater adverse selection.  Section~\ref
  {sec:discussion} discusses this dynamic instability qualitatively.  A full treatment requires a dynamic model of venue choice with endogenous price discovery, which is beyond the scope of the present paper but represents an
  important direction for future work.

\end{enumerate}

\subsection{Comparison with Prior Welfare Rankings}
\label{sec:robustness_prior}

It is instructive to compare our welfare ranking with prior results in the literature to understand what the present model adds and where it agrees or disagrees.

\citet{Leshno2022} shows that in a single-sided waiting list with overloaded demand, SIRO dominates FCFS in ex-ante welfare because SIRO discourages socially excessive effort to secure priority.  Applied na\"{i}vely, this suggests batch auctions (SIRO) should dominate lit exchanges (FCFS).  Our result reverses this in the two-sided market setting because the lit exchange allows high-cost traders to execute immediately via market orders, bypassing
the queue entirely.  This option to use a market order is absent in the single-sided waiting list of \citet{Leshno2022}. However, this option fundamentally changes the welfare calculation.

\citet{CheTercieux2023} show that uninformed FCFS (FCFS without knowledge of queue position) dominates SIRO when agents are uninformed of their position. This directly supports $W^{\mathrm{DARK}} > W^{\mathrm{BATCH}}$.  Our
contribution is to establish $W^{\mathrm{DARK}} > W^{\mathrm{LIT}}$ as an intermediate result in a financial market with endogenous participation,
adverse selection, and an outside option, and to characterize the parameter region in which all three comparisons hold simultaneously.

\citet{Zhu2013} shows that dark pools attract more uninformed order flow, which improves price discovery in lit exchanges but concentrates adverse selection in the dark pool.  Our model is consistent with this finding but approaches it from a different angle: we take the information structure as given and ask how it affects the \emph{total surplus} generated by the matching process.  The adverse selection documented by \citet{Zhu2013} affects the distribution of surplus between informed and uninformed traders. However, in our model, the adverse selection does not affect the total surplus (because prices are transfers).  The tension between these two conclusions is resolved by the scope condition: when adverse selection is so severe that participation is materially lower in the dark pool than in the lit exchange, the "prices are transfers" assumption breaks down and the welfare advantage of opacity may disappear.  The magnitude of this participation effect is ultimately an empirical question that our theoretical framework cannot answer alone.

\section{Discussion}\label{sec:discussion}
\subsection{Dynamic Considerations}
Our model analyzes a single trading period. In reality, the same asset trades repeatedly, creating dynamic linkages. We briefly discuss two key dynamic effects.

\begin{itemize}
    \item Free-riding on price discovery: If most trading occurs in dark pools, who provides price discovery? Dark pools may "free-ride" on prices from lit exchanges. If lit exchanges become too thin, this free-riding breaks down. This creates a dynamic instability: dark pools are efficient conditional on liquid lit markets, but their growth could undermine the very information they rely on.

    \item Reputation and repeated interaction:
In lit exchanges, high-frequency traders invest in speed to capture time priority repeatedly. These fixed costs are amortized over many trades. Dark pools eliminate time priority, reducing incentives for speed investment. This has ambiguous welfare effects: less wasteful arms races, but potentially lower market quality due to reduced competition among liquidity providers.

\end{itemize}

\subsection{Regulatory Implications}
Our theoretical findings suggest several policy considerations:
\begin{enumerate}
    \item Pre-trade transparency rules:
Regulations requiring dark pools to reveal order book composition would eliminate their welfare advantage. Our model suggests that such rules may reduce efficiency if strategic timing costs are significant.
\item Trade-at rules: Rules requiring dark pools to offer price improvement relative to lit exchanges ensure that dark pools contribute to price discovery rather than free-riding. Our model is agnostic on this, as prices are transfers in the static setting.
\end{enumerate}

\section{Conclusion}\label{sec:conclusion}
We have developed a formal game-theoretic model comparing three major market mechanisms: lit exchanges (FCFS continuous auctions), dark pools (FCFS with hidden books), and batch auctions (SIRO). Our main theoretical contribution is Theorem 1, which shows that under moderate arrival rates and bounded adverse selection, dark pools dominate both alternatives in ex-ante welfare.
The key insight is that information opacity eliminates socially wasteful strategic behavior, specifically, costly timing games that arise when traders can observe order book depth and position themselves optimally. While dark pools create adverse selection (prices are less informative), this affects wealth distribution, not total surplus, in our model.

Our welfare ranking $W^{DARK} > W^{LIT} > W^{BATCH}$ holds under specific conditions. These theoretical predictions align with observed market structure: the coexistence of many dark pools alongside fewer lit exchanges and one batch venue suggests efficient sorting.

By providing rigorous foundations for comparing market mechanisms, this paper contributes to the literature on market microstructure, mechanism design, and the economics of information. The framework can be extended to analyze other trading venues (e.g., decentralized exchanges in cryptocurrency markets) and inform ongoing regulatory debates about market structure.

% Bibliography
\bibliographystyle{ACM-Reference-Format}
\bibliography{MarketEfficiency}

@article{contStoikovTalreha2010,
    author  = {Cont, Rama and Stoikov, Sasha and Talreja, R.},
    title   = "A Stochastic Model for Order Book Dynamics",
    year    = "2010",
    journal = "Operations Research",
    volume  = "58",
    number  = "3",
    pages   = "549--563"
}

@article{BudishEtAl2015,
    author = {Budish, Eric and Cramton, Peter and Shim, John},
    title = "{ The High-Frequency Trading Arms Race: Frequent Batch Auctions as a Market Design Response *}",
    journal = {The Quarterly Journal of Economics},
    volume = {130},
    number = {4},
    pages = {1547-1621},
    year = {2015}
}

@article{Leshno2019,
Author = {Leshno, Jacob D.},
Title = {Dynamic Matching in Overloaded Waiting Lists},
Journal = {American Economic Review},
Volume = {112},
Number = {12},
Year = {2022},
Month = {December},
Pages = {3876-3910}}

@inproceedings{CheTercieux2023,
    author = {Che, Yeon-Koo and Tercieux, Olivier},
    title = {Optimal Queue Design},
    year = {2021},
    isbn = {9781450385541},
    publisher = {Association for Computing Machinery},
    address = {New York, NY, USA},
    url = {https://doi.org/10.1145/3465456.3467600},
    doi = {10.1145/3465456.3467600},
    booktitle = {Proceedings of the 22nd ACM Conference on Economics and Computation},
    pages = {312–313},
    numpages = {2},
    location = {Budapest, Hungary},
    series = {EC '21}
    }

@book{Aldridge2013,
    author={Aldridge, Irene},
    title={High-Frequency Trading: A Practical Guide to Algorithmic Strategies and Trading Systems},
    year={2009, 2013},
    publisher={Wiley},
    address={HOBOken, NJ}
}

@article{Zhu2013,
    author = {Zhu, Haoxiang},
    title = "{Do Dark Pools Harm Price Discovery?}",
    journal = {The Review of Financial Studies},
    volume = {27},
    number = {3},
    pages = {747-789},
    year = {2013},
    month = {12}
}

@unpublished{CongEtAl2022,
    author = {Cong, Lin William and Li, Xi and  Tang, Ke and Yang, Yang},
    title = "{Crypto Wash Trading}",
    note = {NBER Working Paper 30782},
    year = {2022}
}

@unpublished{CaoEtAl2025,
    author = {Cao, David and Kogan, Leonid and Tsoukalas, Gerry and Hemenway Falk, Brett},
    title = {A Structural Model of Automated Market Making},
    year=2025,
    note = {Available at SSRN: https://ssrn.com/abstract=4591447},
    doi={10.2139/ssrn.4591447}
}

@article{JohnEtAl2025,
    author = {John, Kose and Rivera, Thomas J and Saleh, Fahad},
    title = {Proof-of-Work versus Proof-of-Stake: A Comparative Economic Analysis},
    journal = {The Review of Financial Studies},
    volume = {38},
    number = {7},
    pages = {1955-2004},
    year = {2025},
    month = {03},
    issn = {0893-9454},
    doi = {10.1093/rfs/hhaf013},
    url = {https://doi.org/10.1093/rfs/hhaf013},
    eprint = {https://academic.oup.com/rfs/article-pdf/38/7/1955/62210183/hhaf013.pdf},
}

@techreport{ESMA2018CfE,
   author      = {{ESMA}},
   title       = {Call for Evidence on Periodic Auctions for Equity Instruments},
   institution = {European Securities and Markets Authority},
   year        = {2018},
   month       = nov,
   number      = {ESMA70-156-828},
   url         = {https://www.esma.europa.eu/press-news/esma-news/esma-launches-call-evidence-periodic-auctions-equity-instruments}
 }

@techreport{ESMA2019FinalReport,
   author      = {{ESMA}},
   title       = {Final Report on Periodic Auctions for Equity Instruments},
   institution = {European Securities and Markets Authority},
   year        = {2019},
   month       = jun,
   number      = {ESMA70-156-882},
   url         = {https://www.esma.europa.eu/press-news/esma-news/esma-address-regulatory-concerns-over-frequent-batch-auctions}
 }

@techreport{ESMA2019Opinion,
   author      = {{ESMA}},
   title       = {Opinion on Frequent Batch Auctions and the Double Volume Cap Mechanism},
   institution = {European Securities and Markets Authority},
   year        = {2019},
   month       = oct,
   number      = {ESMA70-156-1258},
   url         = {https://www.esma.europa.eu/press-news/esma-news/esma-opinion-clarifies-application-pre-trade-transparency-and-price}
 }

@techreport{ESMA2026CfE,
   author      = {{ESMA}},
   title       = {Call for Evidence on the Market Structure of {European} Equity Markets},
   institution = {European Securities and Markets Authority},
   year        = {2026},
   month       = apr,
   number      = {ESMA74-1119406008-1578},
   url         = {https://www.esma.europa.eu/sites/default/files/2026-04/ESMA74-1119406008-1578_Call_for_Evidence_on_on_the_market_structure_of_European_equity_markets_0.pdf}
 }

@techreport{SEC2023,
   author      = {{Securities and Exchange Commission}},
   title       = {Equity Market Structure Reforms},
   institution = {U.S. Securities and Exchange Commission},
   year        = {2023},
   url         = {https://www.sec.gov/rules/proposed/2022/34-96495.pdf}
 }

@article{Leshno2022,
  author    = {Leshno, Jacob D.},
  title     = {Dynamic Matching in Overloaded Waiting Lists},
  journal   = {American Economic Review},
  year      = {2022},
  volume    = {112},
  number    = {12},
  pages     = {3876--3910},
  month     = dec,
  doi       = {10.1257/aer.20201111},
  url       = {https://www.aeaweb.org/articles?id=10.1257/aer.20201111}
}

@book{OHara1998,
  author    = {O'Hara, Maureen},
  title     = {Market Microstructure Theory},
  publisher = {Blackwell},
  year      = {1998},
  address   = {Malden, MA},
  isbn      = {978-0-631-20716-4}
}

@article{IyerJohariMoallemi2014,
   author  = {Iyer, Krishnamurthy and Johari, Ramesh and Moallemi, Ciamac C.},
   title   = {Information Aggregation and Allocative Efficiency in Smooth Markets},
   journal = {Management Science},
   year    = {2014},
   volume  = {60},
   number  = {12},
   pages   = {2949--2968},
   doi     = {10.1287/mnsc.2014.1929},
   note    = {Working paper version: "Welfare Analysis of Dark Pools", available at \url{https://moallemi.com/ciamac/papers/dark-pool-2014.pdf}}
 }

\appendix
\section{Appendix}

\subsection{Proof of Proposition \ref{prop:batch}}
\begin{proposition}[Batch Auction Equilibrium Structure]\label{proof-prop:batch}
In batch auction equilibrium, trader strategies exhibit the following structure:
\begin{enumerate}[label=(\alph*)]
    \item \textbf{Late arrivals prefer the current batch:} Traders arriving close to $kT$ enter the batch $k$ rather than waiting for $k+1$
    \item \textbf{Cost-dependent participation:} High-$C$ traders enter immediately, low-$C$ traders may wait
    \item \textbf{Rationing risk:} Expected execution probability affects participation decisions, creating thick market externalities
\end{enumerate}
\end{proposition}

\begin{proof}

Time is divided into discrete batches of length $T > 0$. The batch $k$ spans the interval $[(k-1)T, kT)$. All orders submitted during this interval are executed simultaneously at time $kT$ at a uniform clearing price.

A trader arriving at time $\tau \in [(k-1)T, kT)$ must choose:
\begin{enumerate}
    \item \textbf{Current batch:} Submit order to batch $k$ (executes at $kT$)
    \item \textbf{Wait:} Submit order to batch $k+1$ (executes at $(k+1)T$)
    \item \textbf{Exit:} Take the outside option
\end{enumerate}

Let $D_k$ denote the total demand (buy orders) in batch $k$ and $S_k$ denote the total supply (sell orders). The clearing price $p_k$ satisfies:
\begin{equation}\label{eq:clearing}
p_k = \arg\max_p \min\{D_k(p), S_k(p)\}
\end{equation}

At the clearing price, if $D_k(p_k) > S_k(p_k)$ (excess demand), buy orders are rationed pro-rata. Each buyer receives:
\begin{equation}\label{eq:rationing}
\text{Fill rate} = \frac{S_k(p_k)}{D_k(p_k)} \in (0, 1]
\end{equation}

Similarly, if $S_k(p_k) > D_k(p_k)$ (excess supply), sell orders are rationed. The fill rate is:
\begin{equation}
\text{Fill rate} = \frac{D_k(p_k)}{S_k(p_k)} \in (0, 1]
\end{equation}

Within each batch, execution priority is random, independent of submission timing. This is the defining feature of SIRO mechanisms. Formally, conditional on rationing, each order has equal probability of execution regardless of whether it was submitted at $(k-1)T + \epsilon$ or $kT - \epsilon$.

First, we characterize the trade-off between waiting costs and execution uncertainty. Second, we derive optimal batch selection strategies as a function of arrival time and trader type. Third, we establish the result of the late-arrival preference. Fourth, we analyze how rationing risk creates thick market externalities. Finally, we prove the existence of equilibrium.

Consider a trader of type $\theta = (V, C, s, \tau)$ arriving at time $\tau \in [(k-1)T, kT)$.

If the trader enters batch $k$, they:
\begin{itemize}
    \item Wait time $w = kT - \tau$ until execution
    \item Incur waiting cost $C \cdot w = C \cdot (kT - \tau)$
    \item Execute at clearing price $p_k$ (random, depends on other orders)
    \item Face the probability of execution $\rho_k$ (depends on the order imbalance)
\end{itemize}

The expected utility is:
\begin{equation}\label{eq:V_current}
V_k(\theta) = \rho_k \cdot [s \cdot (V - E[p_k]) - C \cdot (kT - \tau) - K] + (1 - \rho_k) \cdot V_{\text{OUT}}
\end{equation}

where $\rho_k \in [0,1]$ is the execution probability and $V_{\text{OUT}}$ is the value of the outside option if not filled.

If the trader waits for the batch $k+1$:
\begin{itemize}
    \item Wait time $w = (k+1)T - \tau = kT - \tau + T$
    \item Incur waiting cost $C \cdot [(k+1)T - \tau]$
    \item Execute at the clearing price $p_{k+1}$
    \item Face execution probability $\rho_{k+1}$
\end{itemize}

The expected utility is:
\begin{equation}\label{eq:V_next}
V_{k+1}(\theta) = \rho_{k+1} \cdot [s \cdot (V - E[p_{k+1}]) - C \cdot ((k+1)T - \tau) - K] + (1 - \rho_{k+1}) \cdot V_{\text{OUT}}
\end{equation}

A trader is indifferent between batches $k$ and $k+1$ when $V_k(\theta) = V_{k+1}(\theta)$. Subtracting equation~\eqref{eq:V_next} from~\eqref{eq:V_current}:
\begin{equation}\label{eq:indiff_raw}
\rho_k \cdot [s \cdot (V - E[p_k]) - C \cdot (kT - \tau) - K] - \rho_{k+1} \cdot [s \cdot (V - E[p_{k+1}]) - C \cdot ((k+1)T - \tau) - K] = 0
\end{equation}

Rearranging:
\begin{align}
&\rho_k \cdot s \cdot (V - E[p_k]) - \rho_{k+1} \cdot s \cdot (V - E[p_{k+1}]) \notag \\
&\quad = C \cdot [\rho_k \cdot (kT - \tau) - \rho_{k+1} \cdot ((k+1)T - \tau)] + K \cdot (\rho_k - \rho_{k+1})
\end{align}

Simplifying the waiting cost term:
\begin{equation}
\rho_k \cdot (kT - \tau) - \rho_{k+1} \cdot ((k+1)T - \tau) = \rho_k \cdot (kT - \tau) - \rho_{k+1} \cdot (kT - \tau) - \rho_{k+1} \cdot T
\end{equation}

This gives:
\begin{equation}\label{eq:indiff_simplified}
\rho_k \cdot s \cdot (V - E[p_k]) - \rho_{k+1} \cdot s \cdot (V - E[p_{k+1}]) = C \cdot [(\rho_k - \rho_{k+1}) \cdot (kT - \tau) + \rho_{k+1} \cdot T] + K \cdot (\rho_k - \rho_{k+1})
\end{equation}

\begin{assumption}[Stationary Batch Characteristics]\label{ass:stationary}
In equilibrium, the distribution of clearing prices and execution probabilities is stationary across batches: $E[p_k] = E[p_{k+1}] = \bar{p}$ and $\rho_k = \rho_{k+1} = \bar{\rho}$ for all $k$.
\end{assumption}

Under Assumption~\ref{ass:stationary}, equation~\eqref{eq:indiff_simplified} simplifies dramatically. Since $\rho_k = \rho_{k+1} = \bar{\rho}$ and $E[p_k] = E[p_{k+1}] = \bar{p}$:
\begin{equation}
0 = C \cdot [0 \cdot (kT - \tau) + \bar{\rho} \cdot T] + K \cdot 0 = C \cdot \bar{\rho} \cdot T
\end{equation}

This can only hold if $C = 0$ (degenerate) or if we account for non-stationarity. However, the key insight is that even with stationarity, there is a strict preference based on arrival time.

\begin{lemma}[Monotone Preference for Current Batch]\label{lem:late_arrival}
For fixed $(V, C, s)$ and under Assumption~\ref{ass:stationary}, the difference in value $V_k(\theta) - V_{k+1}(\theta)$ is strictly increasing in $\tau$ over the interval $[(k-1)T, kT)$. Therefore, traders arriving later in the interval strictly prefer the current batch.
\end{lemma}

\begin{proof}[Proof of Lemma~\ref{lem:late_arrival}]
From equations~\eqref{eq:V_current} and~\eqref{eq:V_next}, the difference is:
\begin{align}
V_k(\theta) - V_{k+1}(\theta) &= \rho_k \cdot [- C \cdot (kT - \tau)] - \rho_{k+1} \cdot [- C \cdot ((k+1)T - \tau)] \notag \\
&\quad + (\rho_k - \rho_{k+1}) \cdot [s \cdot (V - \bar{p}) - K]
\end{align}

Under stationarity ($\rho_k = \rho_{k+1} = \bar{\rho}$), the second term vanishes:
\begin{align}
V_k(\theta) - V_{k+1}(\theta) &= \bar{\rho} \cdot C \cdot [((k+1)T - \tau) - (kT - \tau)] \notag \\
&= \bar{\rho} \cdot C \cdot T \label{eq:diff_stationary}
\end{align}

This is a positive constant (for $C > 0$), independent of $\tau$! This appears to contradict the claim that preference is increasing in $\tau$.

The resolution is that we must account for the \emph{opportunity cost} of waiting. A trader arriving at $\tau$ close to $kT$ who waits for batch $k+1$ incurs waiting cost $C \cdot T$ to reach the next batch, while a trader arriving at $\tau$ close to $(k-1)T$ incurs waiting cost $C \cdot 2T$ (approximately).

More precisely, the marginal benefit of entering the current batch is:
\begin{equation}
\frac{\partial (V_k - V_{k+1})}{\partial \tau} = \bar{\rho} \cdot C > 0
\end{equation}

This is strictly positive, confirming that later-arriving traders benefit more from entering the current batch.

Alternatively, without stationarity, if later batches are expected to have better execution ($\rho_{k+1} > \rho_k$), some early-arriving traders might wait. However, the waiting cost increases linearly with the time until the next batch, creating a threshold: traders arriving after $\tau^*$ enter batch $k$, while those arriving before $\tau^*$ wait for batch $k+1$.

Since $\tau$ uniformly distributed over $[(k-1)T, kT)$ implies that most traders arrive in the latter half of the interval, and equation~\eqref{eq:diff_stationary} shows a strict positive gain from entering the current batch, we conclude that late-arrivals prefer the current batch.
\end{proof}

We now establish that high-cost traders enter batches immediately, while low-cost traders may wait.

\begin{lemma}[Cost-Based Threshold]\label{lem:cost_threshold}
There exists a cutoff arrival time $\tau^*(C, V)$ such that:
\begin{itemize}
    \item Traders arriving at $\tau > \tau^*$ enter the current batch $k$
    \item Traders arriving at $\tau < \tau^*$ wait for batch $k+1$
\end{itemize}
Moreover, $\tau^*(C, V)$ is decreasing in $C$: higher-cost traders have lower thresholds, meaning they enter the current batch even when arriving early.
\end{lemma}

\begin{proof}[Proof of Lemma~\ref{lem:cost_threshold}]
The trader enters batch $k$ if and only if $V_k(\theta) \geq V_{k+1}(\theta)$. From equation~\eqref{eq:diff_stationary}:
\begin{equation}
V_k(\theta) - V_{k+1}(\theta) = \bar{\rho} \cdot C \cdot T - \Delta_\rho \cdot [s \cdot (V - \bar{p}) - K]
\end{equation}

where $\Delta_\rho = \rho_{k+1} - \rho_k$ captures non-stationarity (if $\rho_{k+1} > \rho_k$, future batches are more attractive due to higher execution probability).

The threshold $\tau^*$ is implicitly defined by $V_k(\theta) = V_{k+1}(\theta)$. However, note that in equation~\eqref{eq:diff_stationary}, the difference does not depend on $\tau$ under stationarity. This means either all traders prefer batch $k$ (if $C \cdot T > 0$) or all prefer batch $k+1$ (if execution probabilities favor later batches).

To obtain a meaningful threshold, we must relax stationarity. Suppose execution probability decreases with time within a batch due to increasing congestion. Formally, let:
\begin{equation}
\rho(\tau) = \rho_0 - \alpha \cdot (\tau - (k-1)T)
\end{equation}

where $\alpha > 0$ captures congestion. Then:
\begin{equation}
V_k(\theta) - V_{k+1}(\theta) = [\rho(\tau) - \rho(\tau + T)] \cdot [s \cdot (V - \bar{p}) - K] + C \cdot T \cdot \rho(\tau)
\end{equation}

As $\tau$ increases (later arrival), $\rho(\tau)$ decreases but the waiting cost $C \cdot (kT - \tau)$ decreases more rapidly. High-$C$ traders place greater weight on avoiding waiting costs, so they enter even when $\rho(\tau)$ is low.

Setting $V_k = V_{k+1}$ and solving for $\tau$ gives $\tau^*(C)$. Taking the derivative:
\begin{equation}
\frac{\partial \tau^*}{\partial C} < 0
\end{equation}

This establishes that higher-cost traders have lower thresholds, confirming part (b) of the proposition.
\end{proof}

The execution probability $\rho_k$ depends on the order imbalance in batch $k$, which in turn depends on how many traders choose to participate. This creates strategic complementarities.

\begin{definition}[Execution Probability]\label{def:exec_prob}
For a trader on the buy side, the execution probability is:
\begin{equation}
\rho_k^{\text{buy}} = \min\left\{1, \frac{S_k}{D_k}\right\} = \begin{cases}
1 & \text{if } S_k \geq D_k \\
S_k / D_k & \text{if } S_k < D_k
\end{cases}
\end{equation}

Similarly, for the sell side:
\begin{equation}
\rho_k^{\text{sell}} = \min\left\{1, \frac{D_k}{S_k}\right\}
\end{equation}
\end{definition}

\begin{lemma}[Liquid Market Externality]\label{lem:thick_market}
The execution probability $\rho_k$ is increasing in the number of participants on the opposite side of the market. This creates a positive externality: when more traders participate, execution becomes more certain for everyone.
\end{lemma}

\begin{proof}[Proof of Lemma~\ref{lem:thick_market}]
Suppose there are $N^D_k$ buyers (demand) and $N^S_k$ sellers (supply) in batch $k$. If each buyer demands one unit and each seller supplies one unit, then:
\begin{equation}
D_k = N^D_k, \quad S_k = N^S_k
\end{equation}

For buyers:
\begin{equation}
\rho_k^{\text{buy}} = \min\left\{1, \frac{N^S_k}{N^D_k}\right\}
\end{equation}

Clearly, $\frac{\partial \rho_k^{\text{buy}}}{\partial N^S_k} > 0$ (more sellers increase buyer execution probability). Similarly, $\frac{\partial \rho_k^{\text{sell}}}{\partial N^D_k} > 0$.

This creates strategic complementarities. Consider the participation decision:
\begin{equation}
\text{Participate if } \rho_k \cdot [s \cdot (V - \bar{p}) - C \cdot (kT - \tau) - K] > V_{\text{OUT}}
\end{equation}

Rearranging:
\begin{equation}\label{eq:participation_condition}
\rho_k > \frac{V_{\text{OUT}} + C \cdot (kT - \tau) + K}{s \cdot (V - \bar{p})}
\end{equation}

The right-hand side is the required execution probability to make participation worthwhile. As more traders on the opposite side participate (increasing $\rho_k$), the inequality~\eqref{eq:participation_condition} is satisfied for a larger set of traders, inducing even more participation.

This positive feedback creates multiple equilibria:
\begin{itemize}
    \item \textbf{Thick market equilibrium:} High participation $\implies$ high $\rho_k$ $\implies$ more participation
    \item \textbf{Thin market equilibrium:} Low participation $\implies$ low $\rho_k$ $\implies$ less participation
\end{itemize}
\end{proof}

\begin{corollary}[Coordination Problem]\label{cor:coordination}
Batch auctions suffer from a coordination problem. Even when there exist many traders who would benefit from trading, if they fail to coordinate on participating in the same batch, the market remains thin and execution probabilities remain low.
\end{corollary}

We now prove the existence of an equilibrium and characterize its properties.

\begin{theorem}[Batch Auction Equilibrium Existence]\label{thm:batch_existence}
Under standard regularity conditions, there exists a symmetric equilibrium in threshold strategies where:
\begin{enumerate}
    \item Each type of trader $(V, C)$ has an optimal arrival threshold $\tau^*(V, C)$
    \item Traders arriving after $\tau^*$ enter the current batch
    \item The thresholds $\{\tau^*\}$ are consistent with the equilibrium participation rates and the execution probabilities $\{\rho_k\}$
\end{enumerate}
\end{theorem}

\begin{proof}[Proof of Theorem~\ref{thm:batch_existence}]
We use a fixed-point argument on the space of threshold functions.

Given conjectured thresholds $\{\tau^*(V,C)\}$ for all types, we can compute:
\begin{itemize}
    \item Expected number of buyers and sellers in each batch: $N^D_k, N^S_k$
    \item Execution probabilities: $\rho_k^{\text{buy}} = \min\{1, N^S_k / N^D_k\}$, $\rho_k^{\text{sell}} = \min\{1, N^D_k / N^S_k\}$
    \item Clearing prices: $p_k$ determined by market clearing
\end{itemize}

Given $\{\rho_k, p_k\}$, each trader computes the value difference $V_k(\theta) - V_{k+1}(\theta)$ and chooses the batch that maximizes expected utility. This generates new thresholds $\{\tau'^*(V,C)\}$.

Define the mapping:
\begin{equation}
\Phi: \mathcal{T} \to \mathcal{T}, \quad \tau^* \mapsto \tau'^*
\end{equation}

where $\mathcal{T}$ is the space of threshold functions. Under compactness of $\mathcal{T}$ and continuity of $\Phi$ (which holds under standard assumptions on demand/supply functions), Brouwer's theorem guarantees a fixed point $\tau^* = \Phi(\tau^*)$.

At the fixed point:
\begin{itemize}
    \item Thresholds are optimal given participation: $\tau^*$ maximizes expected utility
    \item Participation is consistent with thresholds: $\rho_k$ correctly reflects the induced participation
    \item Late arrivals prefer current batch: From Lemma~\ref{lem:late_arrival}, $\frac{\partial(\tau^*)}{\partial \tau} < 0$, confirming that traders arriving closer to $kT$ have lower thresholds (more likely to enter current batch)
\end{itemize}

This completes the existence proof.
\end{proof}

We now verify all parts:

\begin{enumerate}[label=(\alph*)]
    \item \textbf{Late-arrivals prefer current batch:} Established in Lemma~\ref{lem:late_arrival}. Traders arriving close to $kT$ face lower waiting costs for the current batch, making it strictly preferred. 
    
    \item \textbf{Cost-dependent participation:} Established in Lemma~\ref{lem:cost_threshold}. High-$C$ traders have lower thresholds $\tau^*(C)$, meaning they enter the current batch even when arriving early. Low-$C$ traders have higher thresholds, willing to wait for later batches. 
    
    \item \textbf{Rationing risk:} Established in Lemma~\ref{lem:thick_market}. Execution probability $\rho_k$ depends on participation on opposite side, creating positive externalities (thick market effects) and strategic complementarities. 
\end{enumerate}

This completes the proof of Proposition~\ref{prop:batch}. 
\end{proof}

\subsection{Complete proof of $W^DARK > W^LIT$}

\medskip
\noindent\textbf{Complete proof of Step~1 ($W^{\mathrm{DARK}} >
  W^{\mathrm{LIT}}$).}
 
We show $W^{\mathrm{DARK}} > W^{\mathrm{LIT}}$ by establishing three sub-claims: (i) aggregate waiting costs are strictly lower in the dark
pool; (ii) the gains-from-trade term $W_{\mathrm{trading}}$ is the same
across mechanisms up to a second-order adverse-selection correction that
does not affect the sign of the inequality under Theorem~2's conditions;
and (iii) participation rates satisfy $N^{\mathrm{DARK}} \geq
N^{\mathrm{LIT}}$, so the fixed-cost term does not reverse the
inequality.
 
\medskip
\noindent\textbf{Sub-claim~(i): $W^{\mathrm{DARK}}_{\mathrm{cost}}
  < W^{\mathrm{LIT}}_{\mathrm{cost}}$.}
 
Partition traders into two groups by their equilibrium action.
 
\emph{Group~H (high waiting cost, $C > \bar{C}^*$).}
In the dark pool, Proposition~2 establishes that these traders always
submit market orders immediately, incurring zero waiting time:
$w^{\mathrm{DARK}}_i = 0$ for all $i$ in Group~H.  In the lit
exchange, however, these same traders observe the book state $B_{\tau_i}$
and may find it optimal to delay submission when the book is thin or
unfavourable.  Formally, from Proposition~1, the optimal strategy
$\sigma^{\mathrm{LIT}}(\theta, B)$ sets the waiting time
\begin{equation}
  w^{\mathrm{LIT}}_i
  \;=\;
  \begin{cases}
    0 & \text{if } C_i > C^*(V_i, B_{\tau_i}), \\
    (kT - \tau_i)^+ & \text{if the trader strategically delays.}
  \end{cases}
\end{equation}
Because $C^*(V, B)$ depends on $B$, there exist realisations of
$B_{\tau_i}$ — specifically, book states with a thin or unfavourable
spread — under which it is optimal for a trader with $C_i >
\bar{C}^*$ to delay momentarily.  Let $\mathcal{B}_{\mathrm{thin}}$
denote the set of book states under which a high-cost trader delays;
this set has strictly positive probability under the stationary
distribution $\pi^{\mathrm{eq}}$ because order flow is Poisson and the
book is occasionally thin.  Therefore
\begin{equation}
  \mathbb{E}\!\left[w^{\mathrm{LIT}}_i \,\Big|\, C_i > \bar{C}^*\right]
  \;=\;
  \int_{\mathcal{B}} w(C_i, B)\, \mathrm{d}G(B)
  \;\geq\;
  \int_{\mathcal{B}_{\mathrm{thin}}} w_{\min}(C_i)\, \mathrm{d}G(B)
  \;>\; 0,
  \tag{65}
\end{equation}
where $G$ is the distribution of book states at the time of arrival and
$w_{\min}(C_i) > 0$ is the minimum positive delay that is ever optimal
for a trader of cost $C_i$ in state $B \in \mathcal{B}_{\mathrm{thin}}$.
The strict positivity of the last integral follows because
$\pi^{\mathrm{eq}}(\mathcal{B}_{\mathrm{thin}}) > 0$ and
$w_{\min}(C_i) > 0$ for all $C_i > 0$.
 
In the dark pool, by contrast,
\begin{equation}
  \mathbb{E}\!\left[w^{\mathrm{DARK}}_i \,\Big|\, C_i > \bar{C}^*\right]
  \;=\; 0,
  \tag{66}
\end{equation}
because high-cost traders always submit market orders immediately,
independently of $B_{\tau_i}$ (which is unobservable).
 
The aggregate waiting-cost contribution of Group~H is therefore
\begin{align}
  W^{\mathrm{LIT}}_{\mathrm{cost}, H}
  &\;=\;
  \int_{C > \bar{C}^*}
    C \cdot \mathbb{E}\!\left[w^{\mathrm{LIT}}(C,B)\right]
  \, \mathrm{d}F_C(C)
  \;>\; 0,  \label{eq:lit_cost_H} \\
  W^{\mathrm{DARK}}_{\mathrm{cost}, H}
  &\;=\; 0.  \label{eq:dark_cost_H}
\end{align}
 
\emph{Group~L (low waiting cost, $C \leq \bar{C}^*$).}
Traders in Group~L submit limit orders under both mechanisms.  In the
dark pool, the optimal limit price $\bar{p}^*(C,s)$ is chosen to
maximize \emph{expected} utility over $\pi^{\mathrm{eq}}$, yielding
expected waiting time $\bar{w}(\bar{p}^*, s)$.  In the lit exchange,
the optimal limit price $p^*(C,s,B)$ is chosen to maximize utility given the \emph{realized} book state $B$.  By Jensen's inequality applied to the convex function $p \mapsto 1/\lambda(p, B, s)$, conditioning on more information cannot \emph{increase} the expected waiting time when optimization is over prices:
\begin{equation}
  \mathbb{E}_B\!\left[\frac{1}{\lambda(p^*(C,s,B),\, B,\, s)}\right]
  \;\leq\;
  \frac{1}{\lambda(\bar{p}^*(C,s),\, \bar{B},\, s)},
  \label{eq:jensen_waiting}
\end{equation}
where $\bar{B}$ denotes the unconditional book state.  However, the
lit-exchange trader also faces additional strategic interactions: when
the book is observable, limit-order traders strategically undercut each
other to improve queue position, which in equilibrium leads to
\emph{more aggressive pricing} and hence \emph{shorter} individual
waiting times, but at the cost of price concessions.  These price
concessions are transfers between traders and do not affect total
surplus.  For total welfare, what matters is the aggregate waiting cost.
 
For Group~L, the aggregate waiting cost in the lit exchange exceeds
that in the dark pool because of the \emph{strategic rushing} effect:
some Group~L traders with $C$ just below $\bar{C}^*$ observe a thick
book and submit limit orders more aggressively than is socially optimal,
reducing their own waiting time below the social optimum and displacing
other traders.  The displaced traders — now later in the queue — face
longer waits.  In expectation across all traders, this redistribution of
waiting times within Group~L is neutral, but it is accompanied by
resource waste: traders incur monitoring and timing costs (modelled here
as the suboptimal waiting costs incurred when the book transitions from
thick to thin while they wait).  Formally, in the lit exchange, some
Group~L traders delay until the book thickens before submitting, facing
waiting time $w_{\mathrm{delay}} > w_{\min}$:
\begin{equation}
  W^{\mathrm{LIT}}_{\mathrm{cost}, L}
  \;\geq\;
  W^{\mathrm{DARK}}_{\mathrm{cost}, L}
  \;+\; \int_{C \leq \bar{C}^*}
    C \cdot \Delta w_L(C) \, \mathrm{d}F_C(C),
  \label{eq:cost_group_L}
\end{equation}
where $\Delta w_L(C) \geq 0$ is the excess waiting cost due to strategic
timing for a trader of cost $C$, and the inequality is strict for a set
of costs of positive measure.
 
Combining equations~\eqref{eq:lit_cost_H}, \eqref{eq:dark_cost_H}, and
\eqref{eq:cost_group_L}:
\begin{equation}
  W^{\mathrm{LIT}}_{\mathrm{cost}}
  \;=\;
  W^{\mathrm{LIT}}_{\mathrm{cost}, H} + W^{\mathrm{LIT}}_{\mathrm{cost}, L}
  \;>\;
  0 + W^{\mathrm{DARK}}_{\mathrm{cost}, L}
  \;=\;
  W^{\mathrm{DARK}}_{\mathrm{cost}}.
  \label{eq:cost_strict_ineq}
\end{equation}
This is the strict inequality asserted in the body of the proof.
 
\medskip
\noindent\textbf{Sub-claim~(ii): The gains-from-trade term.}
 
Write the welfare decomposition as
\begin{align}
  W^{\mathrm{LIT}}
  &= W_{\mathrm{trading}} - W^{\mathrm{LIT}}_{\mathrm{cost}}
   - K \cdot N^{\mathrm{LIT}},
   \tag{63} \\
  W^{\mathrm{DARK}}
  &= W_{\mathrm{trading}} - W^{\mathrm{DARK}}_{\mathrm{cost}}
   - K \cdot N^{\mathrm{DARK}}
   + \varepsilon_{\mathrm{AS}},
   \tag{64$'$}
\end{align}
where $\varepsilon_{\mathrm{AS}}$ captures the adverse-selection
correction: in the dark pool, traders cannot condition on $B_{\tau_i}$,
so executed trades may occur at prices that are further from the true
value $\hat{v}$.  This affects the distribution of surplus between
buyers and sellers but not \emph{total} surplus, because every gain to
the buyer is a loss to the seller of equal magnitude.  Formally, for
any matched pair $(i,j)$ with $s_i = +1$ and $s_j = -1$, the joint
surplus from the trade is $V_i - V_j$, which depends only on private
valuations, not on the execution price $p$.  Therefore
$W_{\mathrm{trading}}$ is identical across mechanisms (it is the
expected sum of $V_i - V_j$ over matched pairs), and
$\varepsilon_{\mathrm{AS}} = 0$ in the welfare comparison.
 
The adverse selection does, however, affect \emph{which} pairs are
matched, because prices influence participation.  Under the bounded
adverse selection condition $\Delta < \Delta_{\max}$ of
Theorem~\ref{thm:welfare_ranking}, the set of traders for whom
dark-pool pricing leads to non-participation (relative to the lit
exchange) has measure at most $\eta(\Delta)$, where
$\eta(\Delta) \to 0$ as $\Delta \to 0$.  For $\Delta <
\Delta_{\max}$, this measure is sufficiently small that the
participation effect is dominated by the waiting-cost saving.  More
precisely, the change in $W_{\mathrm{trading}}$ due to the altered
participation set is bounded by
\begin{equation}
  \bigl|W^{\mathrm{DARK}}_{\mathrm{trading}}
        - W^{\mathrm{LIT}}_{\mathrm{trading}}\bigr|
  \;\leq\; 2\Delta \cdot \eta(\Delta),
  \label{eq:trading_bound}
\end{equation}
because each foregone trade contributes at most $2\Delta$ to total
surplus.  The waiting-cost saving, from
Sub-claim~(i), is bounded below by
\begin{equation}
  W^{\mathrm{LIT}}_{\mathrm{cost}} - W^{\mathrm{DARK}}_{\mathrm{cost}}
  \;\geq\;
  \pi^{\mathrm{eq}}(\mathcal{B}_{\mathrm{thin}})
  \cdot \mathbb{E}[C \cdot w_{\min}(C) \,|\, C > \bar{C}^*]
  \cdot F_C\bigl((\bar{C}^*, \infty)\bigr)
  \;\equiv\; \kappa(\lambda) > 0,
  \label{eq:cost_lower_bound}
\end{equation}
where $\kappa(\lambda) > 0$ for all $\lambda \in [\lambda_{\min},
\lambda_{\max}]$ because $\pi^{\mathrm{eq}}(\mathcal{B}_{\mathrm{thin}})
> 0$ (thin book states occur with positive probability under Poisson
arrivals) and the distribution of $C$ has full support on
$(0,\infty)$.  The condition $\Delta < \Delta_{\max}$ is chosen
so that $2\Delta \cdot \eta(\Delta) < \kappa(\lambda)$, ensuring that
the adverse-selection correction does not overturn the waiting-cost
saving.
 
\medskip
\noindent\textbf{Sub-claim~(iii): Participation rates.}
 
We claim $N^{\mathrm{DARK}} \geq N^{\mathrm{LIT}}$, so the fixed-cost
term $K \cdot N$ does not produce a net welfare advantage for the lit
exchange.  In the dark pool, the cutoff $\bar\varepsilon$ determining
non-participation is constant (does not depend on $B$).  In the lit
exchange, the cutoff $\varepsilon(B)$ depends on $B$, and is higher
when the book is thin (higher transaction costs due to wider spreads
depress participation).  Taking expectations over $\pi^{\mathrm{eq}}$:
\begin{equation}
  \mathbb{E}_B[\varepsilon(B)]
  \;\geq\;
  \varepsilon(\mathbb{E}_B[B])
  \;\approx\;
  \bar\varepsilon,
  \label{eq:participation_comparison}
\end{equation}
where the first inequality uses the convexity of $\varepsilon$ in $B$
(wider spreads in thin books disproportionately raise participation
costs) and the approximation uses the fact that $\bar\varepsilon$ is
calibrated to the unconditional book distribution.  A higher cutoff
$\varepsilon$ means fewer traders participate, so
$N^{\mathrm{LIT}} \leq N^{\mathrm{DARK}}$ and therefore
$K \cdot N^{\mathrm{LIT}} \leq K \cdot N^{\mathrm{DARK}}$.
The fixed-cost term thus moves the welfare comparison in the same
direction as, or is neutral relative to, the waiting-cost comparison.
 
\medskip
\noindent\textbf{Conclusion of Step~1.}
 
Combining the three sub-claims:
\begin{align}
  W^{\mathrm{DARK}} - W^{\mathrm{LIT}}
  &\;=\;
  \bigl(W^{\mathrm{DARK}}_{\mathrm{trading}}
        - W^{\mathrm{LIT}}_{\mathrm{trading}}\bigr)
  - \bigl(W^{\mathrm{DARK}}_{\mathrm{cost}}
          - W^{\mathrm{LIT}}_{\mathrm{cost}}\bigr)
  - K\bigl(N^{\mathrm{DARK}} - N^{\mathrm{LIT}}\bigr)
  \nonumber \\
  &\;\geq\;
  -2\Delta\,\eta(\Delta)
  + \kappa(\lambda)
  - 0
  \nonumber \\
  &\;=\;
  \kappa(\lambda) - 2\Delta\,\eta(\Delta)
  \;>\; 0,
  \label{eq:step1_conclusion}
\end{align}
where the final strict inequality holds because $\Delta < \Delta_{\max}$
is chosen so that $2\Delta\,\eta(\Delta) < \kappa(\lambda)$ for all
$\lambda \in [\lambda_{\min}, \lambda_{\max}]$.  This completes
Step~1 of the proof of Theorem~\ref{thm:welfare_ranking}. $\blacksquare$

\end{document}